\documentclass{cpamart1}     
\received{Month 200X}       
\volume{000}
\startingpage{1}                      

\authorheadline{Bal, Hoskins, Quinn, and Rachh}
\titleheadline{Klein-Gordon singular waveguides}

\usepackage{amssymb,amsthm,amsrefs} 
\usepackage{latexsym}
\usepackage{indentfirst}
\usepackage{graphicx}
\usepackage{placeins}
\usepackage{booktabs}
\usepackage{algorithm}
\usepackage{algorithmic}
\usepackage{multirow}
\usepackage{graphicx,color}
\usepackage{diagbox}
\usepackage{newtxtext}
\usepackage{bm}
\usepackage{dsfont}
\usepackage{subfigure}

\usepackage{mathptmx}
\usepackage{times}

\newtheorem{theorem}{Theorem}[section]
\newtheorem{lemma}[theorem]{Lemma}

\newtheorem{prop}[theorem]{Proposition}
\newtheorem*{theorem:invL}{Theorem~\ref{thm:invL}}
\newtheorem*{theorem:reg0}{Theorem~\ref{thm:reg0}}
\newtheorem*{theorem:invE}{Theorem~\ref{thm:invE}}
\newtheorem*{thm:outgoing}{Theorem~\ref{thm:outgoing}}
\newtheorem*{thm:outgoing2}{Theorem~\ref{thm:outgoing2}}
\theoremstyle{definition}

\theoremstyle{remark}
\newtheorem{remark}[theorem]{Remark}

\newcommand\mtx[1]{\bm{\mathsf{#1}}}

\newcommand{\sgn}{\mathrm{sgn}}
\newcommand{\mD}{\mathcal{D}}
\newcommand{\mS}{\mathcal{S}}

\newcommand{\eps}{\epsilon}
\newcommand{\nbuffer}{n_{\textrm{buffer}}}
\newcommand{\nover}{n_{\textrm{over}}}
\newcommand{\niter}{n_{\textrm{iter}}}

\newcommand{\norm}[1]{\lVert#1\rVert}

\usepackage{slashed}

\newcommand{\bbm}{\begin{bmatrix}}
\newcommand{\ebm}{\end{bmatrix}}

\newcommand{\Rm}{\mathbb R}

\newcommand{\nb}{n_{\text{buffer}}}
\newcommand{\nc}{n_c}
\newcommand{\no}{n_{\text{over}}}

\newcommand{\J}{J}
\newcommand{\g}{\chi}

\newcommand{\astar}{\alpha_*}
\newcommand{\cB}{\mathcal{B}}
\newcommand{\cA}{\mathcal{A}}
\newcommand{\cS}{\mathcal{S}}

\newcommand{\hc}{S}
\newcommand{\finint}{I_1}
\newcommand{\finintU}{I_U}
\newcommand{\finintUc}{I_{U^c}}

\newcommand{\so}{\nu}

\newcommand{\grad}{\nabla}

\newcommand{\flatL}{\mathcal{L}_0}
\newcommand{\cL}{\mathcal{L}}
\newcommand{\cM}{\mathcal{M}}

\newcommand{\cP}{\mathcal{P}}

\newcommand{\gc}{C}

\usepackage[utf8]{inputenc}
\usepackage{xcolor}

\begin{document}                        


\title{Integral formulation of Klein-Gordon singular waveguides.}

\author{Guillaume Bal}{University of Chicago}
\author{Jeremy Hoskins}{University of Chicago}
\author{Solomon Quinn}{University of Chicago}
\author{Manas Rachh}{Flatiron Institute}





\begin{abstract}
We consider the analysis of singular waveguides separating insulating phases in two-space dimensions. The insulating domains are modeled by a massive Schr\"odinger equation and the singular waveguide by appropriate jump conditions along the one-dimensional interface separating the insulators. We present an integral formulation of the problem and analyze its mathematical properties. We also implement a fast multipole and sweeping-accelerated iterative algorithm for solving the integral equations, and demonstrate numerically the fast convergence of this method. Several numerical examples of solutions and scattering effects illustrate our theory.
\end{abstract}

\maketitle   



 \tableofcontents



\section{Introduction}\label{sec:intro}
%
The time-harmonic Klein-Gordon equation,
$$ -\Delta u + m^2 u = E^2 u$$
arises naturally in a wide variety of contexts including 
condensed matter and particle physics, classical mechanics, and optics.
When $|E| < m$ it models an {\it insulating medium}: solutions decay exponentially quickly away from a source. In the last forty years there has been particular interest in the case in which two insulators are brought together, meeting at an interface. In this setting, depending on the physical parameters, it is possible to generate surface waves which are localized near, and propagate along, the interface. Such surface waves at the interface of two insulating media have applications in photonics, geophysics, water waves, and condensed matter physics, see~\cite{slobozhanyuk2017three, avron1983homotopy, bernevig2013topological, watanabe2020counting, delplace2017topological, souslov2019topological, witten2016three, sato2017topological}, for example.

In two dimensions, this can be modelled by the following set of partial differential equations (PDEs)

{\small
\begin{align}\label{eq:pde}
    \begin{cases}
    -\Delta u(x) + m^2 u(x) - E^2 u(x) = f_2(x), \quad & x \in \Omega_2,\\
        -\Delta u(x) + m^2 u(x) - E^2 u(x) = f_1(x), \quad & x \in \Omega_1,\\
        \lim_{y \to x \in \Omega_2}u(y) =\lim_{y \to x \in \Omega_1}u(y), \quad & x \in \Gamma, \\
        \lim_{y \to x \in \Omega_2}\hat{n}(x)\cdot\nabla u(y) -\lim_{y \to x \in \Omega_1}\hat{n}(x) \cdot\nabla u(y)=-2mu(x), \quad & x \in \Gamma, \\
    \end{cases}
\end{align}
}
where $\hat{n}(x)$ denotes the unit normal to $\Gamma$ at $x\in \Gamma$ pointing in the direction of $\Omega_2$, the domains $\Omega_{1}$ and $\Omega_2$ denote the supports of the first and second insulators, $f_1$ and $f_2$ are source terms, the constant $m>0$ we refer to as the `mass', and the constant $E$ we refer to as the energy. In the sequel we will assume that $|E|^2 < m^2$ (the PDE in the bulk in this regime is often referred to as the Yukawa equation), that $\overline{\Omega_1 \cup \Omega_2}=\Rm^2$ is the entire plane, and that $\Omega_1$ and $\Omega_2$ meet along an interface $\Gamma=\partial\Omega_1=\partial\Omega_2$. Moreover, we assume that $\Gamma$ is a single smooth simple curve which is {\it asymptotically flat} in both directions and has a positive opening angle at infinity (see Section \ref{sec:detailed} for precise definitions and more detailed discussions). Intuitively, we assume $\Gamma$ has smooth parameterization $\gamma:\mathbb{R}\to \Gamma$ which asymptotically approach two different rays as $t$ goes to $-\infty$ and $\infty.$  For uniqueness of solutions, one should also supplement these equations with suitable boundary conditions at infinity and radiation boundary conditions along the interface.

The numerical simulation of surface waves on infinite boundaries is a computationally challenging 
task, owing to the absence of decay, or slow algebraic decay, of solutions along the interface. For the boundary value problem in Equation~\ref{eq:pde}, on smooth closed curves, surface waves localized near the interface propagate for a finite set of values of $E \in (-m,m)$. Barring these particular values of $E$, standard integral representations can be used to obtain solutions efficiently. Even for values of $E$ for which the localized surface waves are present, the naive integral equation formulation has only a finite dimensional nullspace, which is easily addressed using finite dimensional range completions~\cite{sifuentes2014randomized}.

On infinite interfaces, however, interface waves exist for all values of $E \in (-m,m)$. Moreover, when the classic integral representations for smooth closed curves are extended to infinite interfaces, the resulting equations are not Fredholm second-kind, and the interface waves manifest as a part of the continuous spectrum of the integral operator passing through zero. This issue cannot be addressed using any finite dimensional completions, and presents significant difficulties in the design of numerical methods for their solution.

\subsection{Contributions}
In this paper, we construct a novel {\it analytic} preconditioner for the standard integral representation used on smooth closed curves, which addresses the continuous spectrum when the interface is infinite. Our approach is based on introducing an auxiliary variable defined via the fundamental solution of a certain time-harmonic wave equation on the interface $\Gamma$. We prove bounded invertibility of this preconditioned integral equation for a range of masses and interfaces.

In addition to being analytically well-posed, this representation naturally lends itself to the efficient numerical solution of the
boundary value problem~\ref{eq:pde}. In particular, the resulting discretized integral equation is well-conditioned, and is amenable to standard fast multipole and sweeping methods. Thus, using this approach, Equation~\ref{eq:pde} can be solved quickly and accurately. Even though the tools used in the numerical solution are quite standard, to the best of our knowledge, this is the first high-order accurate and fast numerical solver for the simulation of surface waves that arise in the solution of problems such as Equation~\ref{eq:pde}.

The proposed approach is not limited to Equation~\ref{eq:pde}, and should apply more broadly to the case of piecewise constant insulators with interfaces supporting surface waves. For example, similar problems arise in the study of Dirac models for graphene \cite{3,bernevig2013topological} and linearized shallow water equations used in models of equatorial waves \cite{bal2024topological, delplace2017topological}.

\subsection{Relation to other work}
 
The PDE \eqref{eq:pde} with similar boundary conditions arises in several contexts. Indeed, one interpretation of it is as a Schr\"odinger equation with singular interactions supported on $\Gamma$. The paper \cite{brasche1994schrodinger} gives a detailed spectral analysis of this (and other) Schr\"odinger operators. In the periodic setting, we refer the reader to \cite{figotin1998spectral} for a rigorous treatment of singular potentials. Equations of this form also arise in the study of ``leaky quantum graphs'' \cite{exner2007leaky}. When $\Gamma$ is a compact perturbation of the flat interface, \cite{exner2007leaky} for instance derives asymptotic expansions for generalized eigenfunctions of the Helmholtz operator and obtains expressions for reflection coefficients in a corresponding scattering theory. We note that in our setting, no restriction is placed on the angle between the two branches of $\Gamma$.

For compact interfaces, \cite{holzmann2019boundary} considers the related problem of determining the point spectra of elliptic second-order partial differential operators with ``singular interactions''. For this problem, the authors propose an integral formulation in \cite{holzmann2019boundary}, which is solved via a standard Galerkin method.

More broadly, for interfaces which involve compact perturbations of a halfspace, there are a number of approaches which have been developed for related acoustic and electromagnetic scattering problems~\cite{bruno2017windowed, lai2018new, cai2000fast, weyl1919, okhmatovski2004evaluation, chandler1997impedance}. In particular, such methods enable integral representations that are well-conditioned on smooth closed curves to be extended in a straightforward manner to infinite interfaces. In practice, these techniques have been employed for problems involving turning waveguides (see \cite{10036449} and the references therein). Additionally, a recent set of papers \cite{epstein2023solvingI,epstein2023solving,epstein2024solving} suggest an alternative method based on `gluing' layered medium Green's functions. While this gluing method has a complete analysis, and integral formulation, for a broader class of problems, the method of this paper is substantially simpler to implement, and extends more naturally to multiple non-flat interfaces. Additional numerical approaches to surface wave problems have been studied in a number of other physical contexts. For example, they appear in the solution of Maxwell's equation in transmission between dielectric media when the ratio of permittivities approach a negative real number, see  for example \cite{raether1988surface,tzarouchis2018light}, and \cite{helsing2018spectra,helsing2020extended,helsing2021dirac} for numerical methods.

Finally, we note that surface-wave preconditioners, similar to what is used in this paper, have also been employed in other contexts for solving high-frequency scattering problems in acoustics, electromagnetics, and elasticity, see~\cite{darbas2013combining,antoine2008advances,antoine2005alternative,antoine2006improved,chaillat2015approximate,chaillat2014new,kriegsmann1987new} and the references therein. These preconditioners, frequently referred to as on-surface radiation conditions, are typically used to improve the performance of iterative solvers in complicated geometries and in the high-frequency regime, unlike in the present context where they are used for resolving surface waves inherent to the governing equations.

\subsection{Paper outline}
In Section~\ref{sec:math-prelim}, we present a detailed formulation of the problem, discuss its connection to a related Dirac equation and topological insulators, and review relevant mathematical properties of integral operators of the Yukawa equation. Section~\ref{sec:analytical} presents the main results, and Section~\ref{sec:proofs} discusses their proofs.
In Section \ref{sec:num} we describe an algorithm based on our boundary integral equations, and in Section \ref{sec:illustrations} present several numerical examples. We provide some concluding thoughts and describe plans for future research in Section \ref{sec:conc}.

\section{Mathematical preliminaries \label{sec:math-prelim}}
\subsection{Detailed formulation of the problem}\label{sec:detailed}
In this section we give a more precise statement of the problem under consideration, and summarize the associated conditions on the interface. Towards that end, suppose we are given a smooth simple curve $\Gamma$ separating the plane into a lower region $\Omega_1$ and an upper region $\Omega_2.$ Let $\gamma:\mathbb{R}\to \mathbb{R}^2$ be an arclength parameterization. Moreover, with $\hat{n}(t)$ the normal vector to $\gamma$ at $t\in\mathbb{R}$ pointing in the direction of $\Omega_2$ (using the same notation $\hat n$ for $\hat n(t)$ and $\hat n(\gamma(t))$ to simplify), we assume that $(\gamma'(t),\hat n(t))$ has positive orientation. For concreteness, we additionally assume that $\gamma \in C^\infty (\mathbb{R}; \mathbb{R}^2)$ with
\begin{align}
&|\gamma'(t)| =1, \qquad |\gamma^{(j)}(t)| \le \gc_j e^{-\beta |t|}, \quad j=2,3,4, \dots \label{eq:beta}
\end{align}
for some positive real constants $\gc_j$ and $\beta$, and
\begin{align}
&\lim_{t \to \infty} |\gamma(\pm t)| = \infty, \qquad \lim_{t\to \infty } |\gamma(t) - \gamma(-t)| =\infty \,.  \label{eq:gammainfty}
\end{align}

\begin{remark}
Almost all of the results in this paper would hold if we assume only that $\gamma \in C^3$ with exponentially decaying second and third derivatives. The decay of higher-order derivatives is used in the regularity result Theorem \ref{thm:reg0} but nowhere else.
\end{remark}
The fact that $\gamma$ is one-to-one (along with assumptions \eqref{eq:beta} and \eqref{eq:gammainfty} above) 
implies the existence of some $c>0$ such that
\begin{align}\label{eq:c}
    \frac{|\gamma (t) -\gamma (s)|}{|t-s|} \ge c, \qquad s,t \in \mathbb{R}.
\end{align}
Given a suitably-differentiable function $u,$ for ease of exposition we set
\[
\begin{array}{rcl}
\left[[\hat{n}\cdot \nabla u \right]](t) &:=& \lim_{s\to 0^+} \hat{n}(t)\cdot \nabla \left[u(\gamma(t)+s\,\hat{n}(t)) -u(\gamma(t)-s\,\hat{n}(t)) \right]  \, ,\\[2mm]
\left[[u\right]](t) &:=& \lim_{s\to 0^+} \left[u(\gamma(t)+s\,\hat{n}(t)) -u(\gamma(t)-s\,\hat{n}(t)) \right].
\end{array}
\]
Finally, we suppose that we are given a positive real number $m$ as well as a real number $E$ such that $|E|<m.$ In the following we set $\omega=\sqrt{m^2-E^2}.$

\begin{figure}[h!]
    \centering
    \includegraphics[scale=.35]{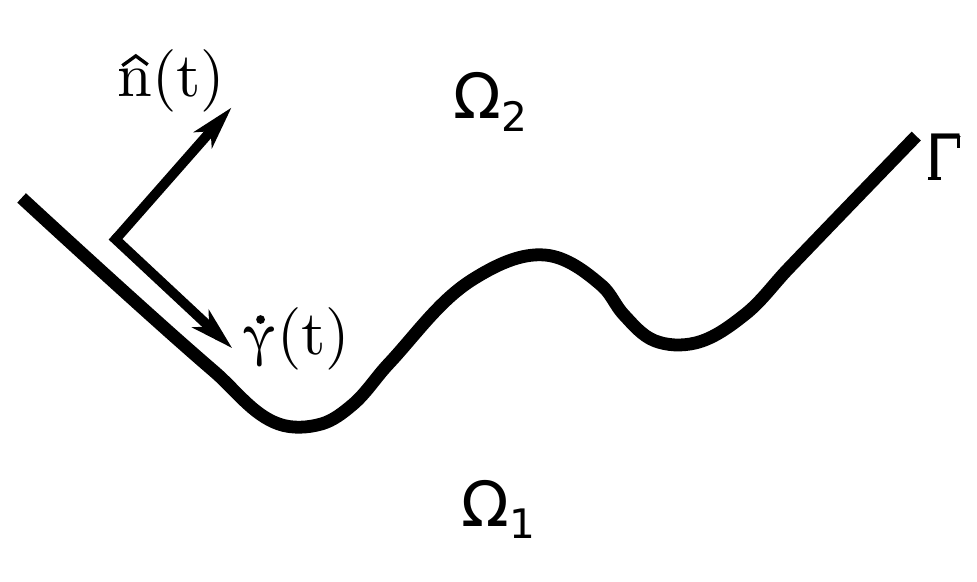}
    \caption{Geometry}
    \label{fig:geom}
\end{figure}

In this paper, we consider the time-harmonic Klein Gordon equation with piecewise discontinuous masses meeting at a single one-dimensional interface:
\begin{align}\label{eqn:pde}
    \begin{array}{rcll}
    -\Delta u(x) + \omega^2 u(x) &=& f_2(x),\quad &x\in \Omega_2,\\[2mm]
    -\Delta u(x) + \omega^2 u(x) &=& f_1(x),\quad &x\in \Omega_1.
\end{array}    
\end{align}
Here the functions $f_1$ and $f_2$ correspond to compactly-supported sources in the lower and upper regions, respectively. For ease of exposition in the following we will denote by $f$ the function which is equal to $f_2$ in $\Omega_2$ and to $f_1$ in $\Omega_1.$ Along the interface we enforce continuity of $u$ as well as a jump condition in the normal derivative, as below
\begin{align}\label{eqn:jump}
 \begin{array}{rcll}
    \left[[\hat{n}\cdot \nabla u \right]](\gamma(t)) &=& -2m u(\gamma(t)), \quad &t \in \mathbb{R},\\
        \left[[u\right]](\gamma(t)) &=& 0, \quad &t \in \mathbb{R}.
    \end{array}
\end{align}

\begin{remark}
The continuity condition is relatively standard. The jump condition in the normal derivative is possibly less so. Our principle motivation comes from the consideration of topological insulators, as is briefly discussed in Section \ref{subsec:rel}. Boundary conditions such as this also arise in the study of ``leaky'' waveguides, particularly in the context of ``leaky quantum waveguides'' (see \cite{exner2007leaky} and the references therein).
\end{remark}

With the above assumptions, there exist solutions $u$ of the PDE which propagate along $\Gamma$. In order to enforce the condition that the surface wave should travel outwards (i.e. no energy should come in from infinity) we impose additional {\it radiation conditions},
\begin{align}
    \lim_{t \rightarrow \pm \infty} (\pm\partial_t - iE) u (\gamma (t) + r \hat{n} (t))&=0, \qquad r \in \mathbb{R}, \label{eqn:outgoing}\\
    \lim_{d(x,\Gamma) \rightarrow \infty} u (x) &= 0 \label{eqn:decayu},
\end{align}
where $d (x,\Gamma) := \min \{|x-y| : y \in \Gamma\}$ is the distance between $x$ and $\Gamma$. The requirement \eqref{eqn:outgoing} 
is known as an {\it outgoing radiation condition}, and intuitively means that $u (\gamma (t)+r \hat{n} (t))$ is approximately of the form $C e^{i E|t|}$ when $|t|$ is large. Here $C \in \mathbb{C}$ is a constant depending on the sources $f_1$ and $f_2,$ as well as the interface $\Gamma$. We refer to \cite{dyatlov2019mathematical} for more details on this topic in similar settings.

\medskip

Our main objective in this paper is the analysis of the problem \eqref{eqn:pde} with boundary conditions (\ref{eqn:jump}, \ref{eqn:outgoing}, \ref{eqn:decayu}).

\subsection{Relation to Dirac equations and topological insulators}\label{subsec:rel}

The above equations are closely related to certain Dirac equations arising in the study of topological insulators. In this section we briefly outline this connection. We begin by recalling that in two dimensions, the time-harmonic Dirac equation is given by
\begin{align}\label{eq:dirac}
 -i\sigma_3 \partial_x \psi -i \sigma_1 \partial_y \psi + m \sigma_2 \psi = E\psi,
\end{align}
where $\psi: \mathbb{R}^2 \to \mathbb{C}^2,$ $m: \mathbb{R}^2 \to \mathbb{R},$ and $\sigma_1,\sigma_2,$ and $\sigma_3$ are the Pauli spin matrices defined by
\begin{align*}
    \sigma_1 = \begin{pmatrix} 0 & 1 \\1 &0 \end{pmatrix}, \quad \sigma_2 = \begin{pmatrix} 0 & -i \\i &0 \end{pmatrix},\quad{\rm and} \quad     \sigma_3 = \begin{pmatrix} 1 & 0 \\0 & -1 \end{pmatrix}.
\end{align*}
 Dirac equations arise naturally in the modelling of topological insulators, particularly in the transport observed at interfaces separating two insulators in different phases \cite{bernevig2013topological,3,4,bal2022asymmetric}. In particular, if two different insulators are brought together along an interface $\{y=0\}$, then this can be modelled by a jump in the mass, setting $m(x,y) = m_2$ for $y>0$ and $m(x,y)=-m_1$ for $y<0$. While each region on its own acts as an insulator (provided that $|E|<\min(|m_1|,|m_2|)=:m_*$), this is no longer the case for the combination, which admits absolutely continuous spectra in $(-m_*,m_*)$ when the signs of $m_1$ and $m_2$ are the same. 

Squaring the Dirac equation (\ref{eq:dirac}) one obtains
\[
-\partial_x^2\psi -\partial_y^2\psi + [m_2^2\theta(y)+m_1^2\theta(-y)]\psi-E^2\psi + (m_2+m_1)\delta_0(y) \sigma_3\psi = 0. 
\]
Here $\theta(y)$ is the Heaviside function, equal to $1$ for $y>0$ and to $0$ for $y\leq0$. The above equation is diagonal, i.e. the components of $\psi$ are not coupled.
Assume, for concreteness, that $m_1, m_2 > 0$.
The equation for the first component is
$$
 -\Delta\psi_1 + [m_2^2\theta(y)+m_1^2\theta(-y)]\psi_1-E^2\psi_1 + (m_2+m_1)\delta_0(y)\psi_1=0,
$$
which admits only the trivial solution $\psi_1 = 0$ (provided $\psi_1$ is not exponentially increasing in $|y|$).
The more interesting second component satisfies
$$
 -\Delta\psi_2 + [m_2^2\theta(y)+m_1^2\theta(-y)]\psi_2-E^2\psi_2 - (m_2+m_1)\delta_0(y)\psi_2=0,
$$
which is exactly (\ref{eqn:pde}, \ref{eqn:jump}) when $m_2 = m_1 = m$, at least when sources are neglected. 

Our main objective in this paper is to analyze such a scalar equation and focus on the resulting propagation of signals along the interface $\Gamma$ separating the insulators. In particular, we are interested in the setting  where $\Gamma$ is curved. The corresponding analysis to the vector-valued Dirac equation is postponed to a later study~\cite{bal2023integral}.

For an analysis of the temporal propagation of wavepackets along a curved interface in the semiclassical regime (i.e, for wavepackets asymptotically localized in the near vicinity of the interface) in both topologically trivial (Klein-Gordon) and non-trivial (Dirac) settings, see \cite{bal2022semiclassical,bal2021edge}.

\subsection{Boundary integral operators and their properties}\label{sec:bio}

In this section we introduce several frequently-encountered boundary integral operators which will be useful in defining the boundary integral equations for the solution of PDE (\ref{eqn:pde}). We begin by recalling that for any $\omega$ in the right-half of the complex plane, the Green's function $G_\omega(x,y)$ for the PDE \begin{align}\label{eq:freePDE}
\begin{split}
    -\Delta u(x) + \omega^2 u(x) &= \delta(x-y),\\
    \lim_{|x|\to \infty} u(x) &= 0,
    \end{split}
\end{align}
is given by 
\begin{align}\label{eq:gf}
G_\omega(x,y) = \frac{1}{2\pi} K_0(\omega |x-y|),
\end{align}
where $K_0$ is the modified Bessel function of the second kind.

Given a function $\mu \in L^2(\mathbb{R})$ its single-layer potential $S_{\omega}[\mu]$ is defined by
$$S_\omega [\mu](x)=\int_\mathbb{R} G_\omega(x,\gamma(t)) \,\mu(t)\, {\rm d}t$$ 
and its double-layer potential $D_\omega[\mu]$ by
$$D_\omega [\mu](x) = -\int_\mathbb{R} \hat{n}(t) \cdot \nabla_x G_\omega(x,\gamma(t)) \,\mu(t)\, {\rm d}t$$
for $x \notin \Gamma.$ It is well-known (see \cite[Lemmas 3.3 and 3.5]{holzmann2019boundary} for example) that $S_\omega$ is continuous across $\Gamma$ while $D_\omega$ and $\hat{n} \cdot \grad S_\omega$ satisfy the following jump relations
\begin{align}
\lim_{s\to 0^+} D_\omega[\mu](\gamma(t) \pm s\,\hat{n}(t)) = \pm \frac{1}{2}\mu(t) -\int_\mathbb{R} \hat{n}(t') \cdot \nabla G_\omega(\gamma(t),\gamma(t')) \,\mu(t')\, {\rm d}t' \label{eq:jumpD}\\
\lim_{s\to 0^+} \hat{n}(t) \cdot \nabla S_\omega[\mu](\gamma(t) \pm s\,\hat{n}(t)) = \mp \frac{1}{2}\mu(t) +\int_\mathbb{R} \hat{n}(t) \cdot \nabla G_\omega(\gamma(t),\gamma(t')) \,\mu(t')\, {\rm d}t' \label{eq:jumpGradS}.  
\end{align}
For a more detailed discussion of potential theory for the Yukawa equation, see~\cite{colton2013integral} for example.

We note that the two integral operators appearing on the right-hand sides of the previous equations are compact (in $L^2(\Rm)$), as their kernels are continuous and rapidly decaying in $t$ and $t'$. 
With some abuse of notation, in the following we denote these operators by ${\mathcal D}_\omega$ and ${\mathcal S}'_\omega,$ respectively. 
We define $\cS_\omega$ and $\mD'_\omega$ by
\begin{align*}
    \cS_\omega [\mu] (t)&= \int_\mathbb{R} G_\omega(\gamma (t),\gamma(t')) \,\mu(t')\, {\rm d}t',\\
    \mD'_{\omega}[\mu](t)&=
    -\int_\mathbb{R} \hat{n}(t') \cdot \nabla^2 G_\omega(\gamma(t),\gamma(t')) \,\hat{n} (t)\,\mu(t')\, {\rm d}t',
\end{align*}
with $\nabla^2 G_\omega$ the Hessian of $G_\omega$.
We note that in this case, ${\mathcal S}_\omega,$ ${\mathcal D}_\omega,$ ${\mathcal S}'_\omega$ and $\mD'_\omega$ can be viewed as operators from $L^2(\mathbb{R}) \to L^2(\mathbb{R}).$ Finally, we remark that both $\mD_\omega$ and $\mS_\omega'$ are zero when restricted to any portion of the boundary which is flat. Moreover, for flat interfaces the kernels of $\mS_\omega$ and $\mD'_{\omega_2}-\mD'_{\omega_1}$ have the following {\it Sommerfeld integral representations}~\cite{sommerfeld1949partial, o2014efficient}
\begin{align}
&\mS_\omega[\mu](t) = \int_\mathbb{R} K_\omega(t-t')\,\mu(t')\,{\rm d}t',\\
&(\mD'_{\omega_2}-\mD'_{\omega_1})[\mu](t) = \int_\mathbb{R} H_{\omega_2,\omega_1}(t-t')\,\mu(t')\,{\rm d}t',
\end{align}
where
\begin{align}
&K_\omega(t) = \frac{1}{4\pi}\int_\mathbb{R} \frac{e^{i\xi t}}{\sqrt{\xi^2+\omega^2}} {\rm d}\xi, \label{eq:K}\\
&H_{\omega_2,\omega_1}(t) = \frac{1}{4\pi}\int_\mathbb{R} \left(\sqrt{\xi^2+\omega_2^2}-\sqrt{\xi^2+\omega_1^2} \right)e^{i \xi t}\,{\rm d}\xi.
\label{eq:H}
\end{align}

\section{Analytical results}\label{sec:analytical}
\subsection{The boundary integral equations}\label{sec:bie}
We now turn to the construction of a suitable boundary integral formulation of the PDE (\ref{eqn:pde}, \ref{eqn:jump}, \ref{eqn:outgoing}, \ref{eqn:decayu}).

We begin by defining the 
operator $Q$ via the following formula
\begin{equation} \label{eq:Q}
 Q[\rho](t) = \frac{m^2}E \int_{\Rm} e^{iE|t-t'|} \rho(t'){\rm d}t', \qquad \rho\in L^2 (\mathbb{R}),
\end{equation}
and set $\cL := I-2m\cS_\omega$ and $\cP := I+Q$ so that
\begin{equation}
\label{eq:LP}
\begin{aligned}
    \cL [\mu] (t) &= \mu (t) - 2m\int_\mathbb{R} G_\omega(\gamma (t),\gamma(t')) \,\mu(t')\, {\rm d}t',\\
    \cP [\rho] (t) &= \rho (t) + \frac{m^2}E \int_{\Rm} e^{iE|t-t'|} \rho(t'){\rm d}t'  \, ,
\end{aligned}
\end{equation}
for $\mu, \rho \in L^2 (\mathbb{R})$. 

With these definitions, we consider the following boundary integral equation
\begin{align}\label{eq:bif}
\cL \cP [\rho] (t) = 2m\,u_i (\gamma (t)), \qquad t \in \mathbb{R}
\end{align}
where $\rho:\mathbb{R}\to \mathbb{C}$ is an unknown {\it density} and
$u_i$ is given by 
\begin{align}\label{eq:ui1}
    u_i (x) := \int_{\mathbb{R}^2} G_\omega (x,y) f(y) {\rm d}y, \qquad x \in \mathbb{R}^2.
\end{align}
Here, as above, $f$ denotes the function which is equal to $f_2$ in $\Omega_2$ and $f_1$ in $\Omega_1.$
For ease of exposition, we assume that $f_j \in C_c (\Omega_j)$ so that the right-hand side of \eqref{eq:bif} is smooth and exponentially decaying in $t$. Our results easily extend to a finite collection of point sources in $\Omega_1 \cup \Omega_2$, which still produce a smooth and exponentially decaying right-hand side.

\begin{remark}
The operator $Q$ can naturally be interpreted as the fundamental solution operator of the one-dimensional Helmholtz equation
$$\Delta_\Gamma v + E^2 v = m^2 \rho,$$
where $\Delta_\Gamma$ is the Laplace-Beltrami operator of the interface curve $\Gamma.$ 
\end{remark}
\begin{remark}
Note that the operator $Q$ could have equally been defined by replacing $e^{iE|t-t'|}$ by $e^{-iE|t-t'|}$. The choice of sign in \eqref{eq:Q} amounts to a choice of {\it outgoing radiation condition}, namely that we filter out incoming radiation. Indeed,
if $\rho \in L^1$, then $\lim_{t \rightarrow \pm \infty} (Q[\rho](t) - \frac{m^2}{E}e^{i E |t|} \hat{\rho} (\pm E)) = 0$, meaning that solutions of a 
time-dependent Klein Gordon
problem would propagate to the right for $t>0$ and to the left for $t<0$ whenever $|t|$ is sufficiently large ($t$ is the space, not time, variable here). See \cite{dyatlov2019mathematical} for details on the notion of incoming and outgoing radiation.
\end{remark}

Our main results are Theorems \ref{thm:invL} and \ref{thm:invE} below, which establish the well-posedness of the integral formulation \eqref{eq:bif} for ``almost all'' choices of $m$ and $E$. 
These theorems require the following weighted $L^2$ spaces.
For $\alpha\in \mathbb{R}$,
letting $w_\alpha (t) := e^{\alpha|t|}$ we define $L^2_\alpha := \{\rho \in L^2 (\mathbb{R}) : w_\alpha \rho \in L^2(\mathbb{R})\}$. Moreover, we define the norm $\norm{\cdot}_{L^2_\alpha}$ by $\norm{\rho}_{L^2_\alpha}:= \norm{w_\alpha \rho}_{L^2}$.
\begin{theorem}\label{thm:invL}
Fix $m_0 > 0$ and $E_0 \in (-m_0, m_0) \setminus \{0\}$. Define $\omega_0 := \sqrt{m_0^2 - E_0^2}$ and set $m=\lambda m_0$ and $E = \lambda E_0$ for $\lambda \in \mathbb{R}$. Define the function $\astar$ by
\begin{align}\label{eq:alphamax}
    \astar (\omega) := \frac{\beta c \omega}{\beta + c\omega},
\end{align}
where we recall the definitions of $\beta$ and $c$ in \eqref{eq:beta} and \eqref{eq:c}. Then for any $0 < \alpha < \astar (\omega_0)$,
the integral equation \eqref{eq:bif} admits a unique solution $\rho \in L^2_\alpha$
for all but a finite number of $\lambda \in [1,\infty)$.
\end{theorem}
\begin{remark}\label{rem:lambda}
We note that the above dependence of $m$ and $E$ on $\lambda$ is equivalent to setting $m = m_0$ and $E = E_0$ while changing $\gamma$ by $\gamma_\lambda (t) := \lambda \gamma (t/\lambda)$. Thus increasing the value of $\lambda$ can be thought of as stretching out the interface.
\end{remark}

We continue to use the definition \eqref{eq:alphamax} of $\astar$ in the theorems below.
\begin{theorem}\label{thm:invE}
Fix $m>0$.
For any $0 < \eps_1, \eps_2 < m$ and
$0<\alpha<\astar\left(\sqrt{2\eps_2 m - \eps_2^2}\right)$, 
the integral equation \eqref{eq:bif} admits a unique solution $\rho \in L^2_\alpha$
for all but a finite number of $E \in [-m+\eps_2, -\eps_1] \cup [\eps_1, m-\eps_2]$.
\end{theorem}

When the interface $\Gamma$ is smooth and the sources $f_1$ and $f_2$ continuous, the solution $\rho \in L^2_\alpha$ from Theorems \ref{thm:invL} and \ref{thm:invE} in fact satisfies a higher degree of regularity. We introduce the weighted Sobolev spaces, which for $\alpha \in \mathbb{R}$ and $s \in \mathbb{N} \cup \{0\}$ are defined by 
\begin{equation}
\label{eq:hsalpha}
H_\alpha^s := \{ \rho \in L^2 (\mathbb{R}) : w_\alpha \rho^{(j)} \in L^2 (\mathbb{R}), 0 \le j \le s \},
\end{equation}
where $\rho^{(j)}$ denotes the $j$th derivative of $\rho$. Recall that $w_\alpha (t) := e^{\alpha |t|}$.

\begin{theorem}\label{thm:reg0}
Let $\so \in \mathbb{N} \cup \{0\}$ and $0 < \alpha < \astar (\omega)$. If the function $\rho \in L^2_\alpha$ satisfies $\cL \cP \rho\in H^\so_\alpha$, then $\rho\in H^\so_\alpha$. In particular, suppose $f_j \in C_c (\Omega_j)$ for $j=1,2$, then any function $\rho \in L^2_\alpha$ satisfying \eqref{eq:bif} also belongs to $\cap_{s \in \mathbb{N}} H_\alpha^s$.
\end{theorem}

We conclude this section by relating the solutions of \eqref{eq:bif} to the solutions of (\ref{eqn:pde}, \ref{eqn:jump}, \ref{eqn:outgoing}, \ref{eqn:decayu}).
\begin{theorem}\label{thm:outgoing}
Suppose $f_j \in C_c (\Omega_j)$ for $j=1,2$ and let $\alpha > 0$.
Suppose $\rho \in L^2_\alpha$ satisfies \eqref{eq:bif}, and set
\begin{align}\label{eq:ubif}
    \mu := \cP [\rho], \qquad u_s := S_\omega [\mu], \qquad u := u_i + u_s.
\end{align}
Then \eqref{eqn:pde}, \eqref{eqn:jump}, \eqref{eqn:outgoing} and \eqref{eqn:decayu} hold. In particular, $u$ constructed in this way is a solution of the PDE.
\end{theorem}

For proofs of the above theorems, see Section \ref{sec:proofs}.

\subsection{Intuition for the integral formulation}\label{sec:intuitive}
In this section we outline an intuitive derivation of the boundary integral equations, and sketch a proof of invertibility for the case of a flat interface. 

Recall that the PDE \eqref{eqn:pde} we wish to solve is 
\begin{align} \label{eq:equalm}
    -\Delta u(x) + \omega^2 u(x) = f(x),\quad x\in \mathbb{R}^2\setminus \Gamma,
\end{align}
where $f$ is the function equal to $f_2$ in $\Omega_2$ and to $f_1$ in $\Omega_1$.

Then if we define the functions $u_i$ and $u_s$ by
\begin{align}\label{eq:ui1_old}
    u_i (x) := \int_{\mathbb{R}^2} G_\omega (x,y) f(y) dy, \qquad u_s := u-u_i,
\end{align}
and assume the existence of a {\it density} $\mu \in L^2 (\mathbb{R})$ such that $u_s (x) = S_\omega [\mu] (x)$ for all $x \in \mathbb{R}^2 \setminus \Gamma$, it follows that
\begin{align}\label{eq:mu}
    -\mu (t) + 2m \cS_\omega [\mu] (t) = -2mu_i (\gamma (t)), \quad t \in \mathbb{R}.
\end{align}
Indeed, $u_i$ is smooth across $\Gamma$ and hence the definition of $u_s$ implies
\begin{align}\label{eq:us}
      \begin{cases}
    \Delta u_s(x) - \omega^2 u_s(x) = 0,\quad x\in \mathbb{R}^2\setminus \Gamma,\\
    \left[[\hat{n}\cdot \nabla u_s \right]] + 2m u_s = -2m u_i, \quad x \in \Gamma,
    \end{cases}  
\end{align}
with the added condition that $u_s + u_i$ is continuous across $\Gamma$. It is clear that the first line of \eqref{eq:us} holds for any choice of $\mu$, while
the boundary equation \eqref{eq:mu} is a consequence of the second line of \eqref{eq:us} together with
\eqref{eq:jumpGradS} and the definition of $\mu$.

We recognize the left-hand side of \eqref{eq:mu} as $-\cL [\mu]$, where we recall the definition of $\cL$ in \eqref{eq:LP}.
We will see below that for the flat-interface case, $\cL$ has a continuous spectrum which passes through zero. As such, in general the solution will not be unique without imposing additional conditions. Moreover, if the domain is truncated, the resulting truncated operator will at best be poorly-conditioned in the limit as the length of the boundary tends to infinity and in general will not converge as the size of the truncated domain grows. For a flat interface, the problem is easily analyzed in the Fourier domain. Indeed, observing that
\begin{align}\label{eq:fp}\mathcal{F}_{t \rightarrow \xi} \Big\{\frac{1}{E} e^{iE|t|}\Big\} =\lim_{\eps \downarrow 0}\mathcal{F}_{t \rightarrow \xi} \Big\{\frac{1}{E+i\eps} e^{i(E+i\eps)|t|}\Big\} = -\frac{2i}{\xi^2 - E^2},\end{align} 
and combining with $\mu = \cP [\rho]$, we get
\begin{equation}
\begin{aligned}
   \mathcal{F}\{ \cL [\mu]\}(\xi) &= \mathcal{F} \{\cL \cP [\rho]\} (\xi) \\
 &= \Big( 1 - \frac{m}{\sqrt{\xi^2 + \omega^2}}\Big)\Big(1- \frac{2im^2}{\xi^2 - E^2}\Big)\tilde{\rho} (\xi) \\
 &=: a(\xi)\tilde{\rho} (\xi),
\end{aligned}
\label{eq:symbolKP}
\end{equation}
where $\cP$ and $\cL$ are respectively defined in \eqref{eq:LP}, and $\tilde{\rho}$ denotes the Fourier transform of $\rho$. Here, we used \eqref{eq:K} to derive the Fourier representation of $\cL$.
The fact that $\cL\cP$ is a point-wise multiplication in the Fourier domain follows immediately from translation invariance of $\cL$ and $\cP$ (both kernels $k(t,t')$ are functions of only $t-t'$).
The singularities of the second factor of $a$ at $\xi = \pm E$ are canceled by the zeros of the first factor, making $a$ an analytic function.
Since $a$ is nonzero for all $\xi$ and converges to $1$ as $|\xi| \rightarrow \infty$, there exist constants $0<c<C$ such that $c \le |a| \le C$ uniformly in $\xi$.

We conclude that 
the operator $(\cL\cP)^{-1}$ is
bounded on $L^2 (\mathbb{R})$ with (bounded and analytic) Fourier symbol \begin{align}\label{eq:ainv}
    a^{-1} = \Big[\Big( 1 - \frac{m}{\sqrt{\xi^2 + \omega^2}}\Big)\Big(1- \frac{2im^2}{\xi^2 - E^2}\Big)\Big]^{-1} = \Big(1+ \frac{m}{\sqrt{\xi^2 + \omega^2}}\Big) \Big(1 + \frac{(2i+1)m^2}{\xi^2 -E^2 - 2im^2}\Big).
\end{align}
This means there exists a unique function $\rho \in L^2 (\mathbb{R})$ such that $\cL\cP [\rho] = 2m u_i$, and hence the boundary integral equation is invertible (with a bounded solution).

Note that the above derivation does not motivate the factor of $m^2/E$ in the definition \eqref{eq:Q} of $Q$; this is addressed by Remarks \ref{remark:m2_1} and \ref{remark:m2_2} below.


\subsection{Extension to two masses}\label{subsec:extension}
As suggested by Section \ref{subsec:rel}, the PDE \eqref{eqn:pde}-\eqref{eqn:jump} can be extended to model insulators with different `masses' in the upper and lower half-planes. To this end, let $m_1, m_2$ be positive real numbers, fix $E \in \mathbb{R}$ such that $|E| < \min (m_1, m_2)$, and define $\omega_j := \sqrt{m_j^2 - E^2}$ for $j=1,2$. The two-mass PDE is then
\begin{align}\label{eqn:pde2}
    \begin{array}{rcll}
    -\Delta u(x) + \omega_2^2 u(x) &=& f_2(x),\quad &x\in \Omega_2,\\[2mm]
    -\Delta u(x) + \omega_1^2 u(x) &=& f_1(x),\quad &x\in \Omega_1,
\end{array} 
\end{align}
and the jump conditions are
\begin{align}\label{eqn:jump2}
    \begin{array}{rcll}
    \left[[\hat{n}\cdot \nabla u \right]](\gamma(t)) &=& -(m_1 + m_2) u(\gamma(t)), \quad &t \in \mathbb{R},\\
        \left[[u\right]](\gamma(t)) &=& 0, \quad &t \in \mathbb{R} \, .
    \end{array}
\end{align}
The radiation conditions remain unchanged, and are given by
\begin{align}
    \lim_{t \rightarrow \pm \infty} (\pm\partial_t - iE) u (\gamma (t) + r \hat{n} (t))&=0, \qquad r \in \mathbb{R}, \label{eqn:outgoing2}\\
    \lim_{d(x,\Gamma) \rightarrow \infty} u (x) &= 0 \label{eqn:decayu2},
\end{align}
where $d (x,\Gamma) := \min \{|x-y| : y \in \Gamma\}$ as before is the distance between $x$ and $\Gamma$.

We first derive the corresponding boundary integral equation.
Define $\bar{m} := \frac{1}{2} (m_1 + m_2)$.
Analogous to \eqref{eq:Q}, let $Q_2 : L^2 (\mathbb{R}) \rightarrow L^2 (\mathbb{R})$ be given by
\begin{equation} \label{eq:Q2}
 Q_2 [\rho](t) = \frac{\bar{m}^2}E \int_{\Rm} e^{iE|t-t'|} \rho(t'){\rm d}t', \qquad \rho\in L^2 (\mathbb{R}).
\end{equation}
We define the operator $\cL_2: (L^2(\mathbb{R}))^2 \to (L^2(\mathbb{R}))^2$ by
\begin{align}\label{eq:L2}
    \cL_2 := I - 
    \begin{pmatrix}
        \mS'_{\omega_2}-\mS'_{\omega_1} +\bar{m}[\mS_{\omega_2}+\mS_{\omega_1} ] &
        \bar{m}\left[\mD_{\omega_2}+\mD_{\omega_1}\right] +\mD'_{\omega_2} - \mD'_{\omega_1}\\[2mm]
        -(\mS_{\omega_2}-\mS_{\omega_1}) & -(\mD_{\omega_2}-\mD_{\omega_1})
    \end{pmatrix}
\end{align}
and $\cP_2 : (L^2(\mathbb{R}))^2 \to (L^2(\mathbb{R}))^2$ by
\begin{align}\label{eq:P2}
\cP_2 := \left(I+ V\begin{pmatrix}Q_2 &0\\0& 0 \end{pmatrix} V^{-1}\right), \qquad V:= \begin{pmatrix} 
-1 & \frac{1}{2m_2}-\frac{1}{2m_1}\\
 \frac{1}{2m_2}-\frac{1}{2m_1} & 1
    \end{pmatrix}.
\end{align}

Our integral equation is then to 
find a density $\sigma \in (L^2 (\mathbb{R}))^2$ such that
\begin{align}\label{eq:bif2}
    \cL_2 \cP_2 [\sigma] = r,
\end{align}
where
\begin{equation}
\begin{aligned}\label{eq:ui}
    r :&= \begin{pmatrix}\left[[ \hat{n} \cdot \nabla u_i \right]]+2\bar{m} u_{i})\\
    -\left[[ u_i\right]]
    \end{pmatrix}, \\
u_i(x) :&= \begin{cases}
\int_{\Omega_2} G_{\omega_2}(x,y) f_2(y)\,{\rm d} y, \quad & x \in \Omega_2,\\
\int_{\Omega_1} G_{\omega_1}(x,y) f_1(y)\,{\rm d} y, \quad & x \in \Omega_1.
\end{cases}
\end{aligned}
\end{equation}

We can then relate the solutions of \eqref{eq:bif2} to solutions of 
(\ref{eqn:pde2}, \ref{eqn:jump2}, \ref{eqn:outgoing}, \ref{eqn:decayu}).

\begin{theorem}\label{thm:outgoing2}
Suppose $f_j \in C_c (\Omega_j)$ for $j=1,2$ and let $\alpha > 0$.
Suppose $\sigma \in \cap_{s \in \mathbb{N}} (H_\alpha^s)^2$ is a solution of \eqref{eq:bif2}, and set
\begin{align}\label{eq:mrus}
    \begin{pmatrix}
    \mu \\ \rho
    \end{pmatrix} := \cP_2 [\sigma], \qquad 
    u_s(x) = \begin{cases}
    D_{\omega_2}[\rho](x)+S_{\omega_2}[\mu](x), \quad & x \in \Omega_2,\\
     D_{\omega_1}[\rho](x)+S_{\omega_1}[\mu](x), \quad & x \in \Omega_1.\\
\end{cases} 
\end{align}
Then $u:=u_{i} + u_{s}$ satisfies \eqref{eqn:pde2}, \eqref{eqn:jump2}, \eqref{eqn:outgoing2} and \eqref{eqn:decayu2}.
\end{theorem}

We refer to Section \ref{sec:proofs} for the proof of this theorem. 
Note that if the data $f_j \in C^\infty_c (\Omega_j)$, then any solution $\sigma$ of \eqref{eq:bif2} must also belong to $\cap_{s \in \mathbb{N}} (H_\alpha^s)^2$. The proof of such a result follows from a simple modification of the proof of Theorem \ref{thm:reg0}.

Echoing Section \ref{sec:intuitive}, we now turn to the intuitive derivation of the two-mass integral equation \eqref{eq:bif2}.
This case is slightly more complicated, though the reasoning is similar to the equal mass case. Define the functions $u_i$ and $u_s$ by
\begin{align} \label{eq:ui2}
u_i(x) = \begin{cases}
\int_{\Omega_2} G_{\omega_2}(x,y) f_2(y)\,{\rm d} y, \quad & x \in \Omega_2,\\
\int_{\Omega_1} G_{\omega_1}(x,y) f_1(y)\,{\rm d} y, \quad & x \in \Omega_1,
\end{cases}, \qquad u_s := u-u_i.
\end{align}
Note that $u_i$ is no longer continuous across $\Gamma$. For notational convenience, on the boundary between $\Omega_1$ and $\Omega_2$ we set $u_i$ to be the average of the limits from above and below. If there are densities $\mu, \rho \in L^2 (\mathbb{R})$ such that
\begin{align}\label{eq:ansatz}
u_s(x) = \begin{cases}
    D_{\omega_2}[\rho](x)+S_{\omega_2}[\mu](x), \quad & x \in \Omega_2,\\
     D_{\omega_1}[\rho](x)+S_{\omega_1}[\mu](x), \quad & x \in \Omega_1,\\
\end{cases}    
\end{align}
then we eventually arrive at the following linear system of equations for $\mu$ and $\rho$,
{\small
\begin{align}
\rho(t)&+\Bigg[\mD_{\omega_2}[\rho](t)-\mD_{\omega_1}[\rho](t)\Bigg]+\Bigg[\mS_{\omega_2}[\mu](t)-\mS_{\omega_1}[\mu](t) \Bigg] = -\left[[ u_i\right]](\gamma(t)),\label{eqn:BIE_1} 
\end{align}
\begin{align}
    -\mu(t) + \Bigg[\mS'_{\omega_2}[\mu](t)-\mS'_{\omega_1}[\mu](t) \Bigg] +& \Bigg[\mD'_{\omega_2}[\rho](t) - \mD'_{\omega_1}[\rho](t)\Bigg]\nonumber\\
     +\bar{m} \Bigg[\mS_{\omega_2}[\mu](t)+\mS_{\omega_1}[\mu](t) \Bigg]&+\bar{m}\Bigg[\mD_{\omega_2}[\rho](t)+\mD_{\omega_1}[\rho](t)\Bigg]\label{eqn:BIE_2} \\
    &+2\bar{m} u_{i}(\gamma(t)) = -\left[[ \hat{n} \cdot \nabla u_i \right]](\gamma(t)),
    \nonumber
\end{align}
}
which are understood to hold for all $t \in \mathbb{R}$. Recall the definition $\bar{m} := \frac{1}{2} (m_1 + m_2)$.
To derive the above, observe that the definition of the scattered field $u_s$ implies that
\begin{align}\label{eqn:pde_us}
      \begin{cases}
    \Delta u_s(x) - \omega_2^2 u_s(x) = 0,\quad x\in \Omega_2,\\
    \Delta u_s(x) - \omega_1^2 u_s(x) = 0,\quad x\in \Omega_1,\\
    \left[[\hat{n}\cdot \nabla u_s \right]] + (m_2+m_1) u_s = -\left[[\hat{n}\cdot \nabla u_i \right]] -(m_2+m_1) u_i, \quad x \in \Gamma,
    \end{cases}  
\end{align}
with $u_s+u_i$ continuous across $\Gamma$ as before. The first two lines of \eqref{eqn:pde_us} hold for any choice of $\mu$ and $\rho$.
Enforcing continuity of $u$ at the interface, and using the jump relations (\ref{eq:jumpD}, \ref{eq:jumpGradS}) for the layer potentials, we 
obtain \eqref{eqn:BIE_1}.
Since $u$ on $\Gamma$ takes the form
{\small
\begin{align}\label{eq:u_Gamma}
u(\gamma(t)) = \frac{1}{2} \left[\mD_{\omega_2}[\rho](t)+\mD_{\omega_1}[\rho](t)\right]+\frac{1}{2}\Bigg[\mS_{\omega_2}[\mu](t)+\mS_{\omega_1}[\mu](t) \Bigg] + u_i(\gamma (t)), 
\end{align}
}
the derivative jump condition in (\ref{eqn:pde_us}) 
implies \eqref{eqn:BIE_2}.

Observe that (\ref{eqn:BIE_1}, \ref{eqn:BIE_2}) reads 
\begin{equation}\label{eq:ip2}
\cL_2 \begin{pmatrix}\mu\\ \rho\end{pmatrix} = r, 
\end{equation}
where $\cL_2$ and $r$ are defined in \eqref{eq:L2} and \eqref{eq:ui}, respectively.
As above, $\mathcal{L}_2$ may admit a continuous spectrum passing through zero, 
hence solving for $(\mu, \rho)$ in \eqref{eq:ip2} is in general an ill-posed problem.

To remedy this issue, we again turn our attention to the special case of a flat boundary, where Fourier representations of the relevant integral operators are given by \eqref{eq:K} and \eqref{eq:H}.
Taking the Fourier transform of (\ref{eqn:BIE_1}) in the flat case, we find that
\begin{align}\label{eqn:ft_rho}
    \tilde{\rho}(\xi) + \frac{1}{2}\left[\frac{1}{\sqrt{\xi^2+\omega_2^2}}-\frac{1}{\sqrt{\xi^2+\omega_1^2}} \right] \tilde{\mu}(\xi) = -\left[[\tilde{u}_i \right]](\xi).
\end{align}
Similarly, upon taking the Fourier transform of (\ref{eqn:BIE_2}), we obtain
\begin{align}\label{eqn:ft_mu}
    -\tilde{\mu}& -\frac{1}{2}\left[\sqrt{\xi^2+\omega_2^2}-\sqrt{\xi^2+\omega_1^2} \right]\tilde{\rho} + \frac{m_2+m_1}{4} \left[\frac{1}{\sqrt{\xi^2+\omega_2^2}}+\frac{1}{\sqrt{\xi^2+\omega_1^2}} \right] \tilde{\mu}\nonumber\\
    &\quad \quad = -\frac{m_2+m_1}{2}(\tilde{u}_{i,1}+\tilde{u}_{i,2}) - \left[\left[ \hat{n}\cdot \widetilde{\nabla u_i}\right]\right].
\end{align}
After solving (\ref{eqn:ft_rho}) for $\tilde{\rho}$ and substituting it into (\ref{eqn:ft_mu}) we see that
\begin{align}\label{eqn:flat_ft}
\left[-1+\frac{1}{4}(\xi_2-\xi_1)(\xi_2^{-1}-\xi_1^{-1})+\frac{m_2+m_1}{4}(\xi_2^{-1}+\xi_1^{-1}) \right]\tilde{\mu} = \tilde{\psi}
\end{align}
where $\xi_{1,2} = \sqrt{\xi^2+\omega_{1,2}^2}=\sqrt{\xi^2+m_{1,2}^2-E^2}$, and
\begin{align*}
\tilde{\psi}(\xi) =-\frac{m_2+m_1}{2}(\tilde{u}_{i,1}+\tilde{u}_{i,2}) - \left[\left[ \hat{n}\cdot \widetilde{\nabla u_i}\right]\right]-\frac{1}{2}(\xi_2-\xi_1)\left[[ \tilde{u}_i\right]](\xi).
\end{align*}

Let $R(\xi)$ denote  the Fourier multiplier from the left-hand side of (\ref{eqn:flat_ft}) defined by
\begin{align}\label{eq:Rxi}
    R(\xi):= -1+\frac{1}{4}(\xi_2-\xi_1)(\xi_2^{-1}-\xi_1^{-1})+\frac{m_2+m_1}{4}(\xi_2^{-1}+\xi_1^{-1}).
\end{align}
Differentiating \eqref{eq:Rxi} with respect to $\xi^2$ we see that
\begin{align*}
\frac{{\rm d}}{{\rm d}\xi^2} R(\xi) = -\frac{m_2+m_1}{8}\frac{\xi_2^3+\xi_1^3}{\xi_1^3\xi_2^3}+\frac{1}{8}\frac{(\xi_2^2-\xi_1^2)^2}{\xi_1^3\xi_2^3}.    
\end{align*}
Now, $\xi_2^2-\xi_1^2 = \omega_2^2-\omega_1^2,$ $\xi_{2}\ge \omega_2,$ $\xi_1\ge \omega_1,$ and $m_1+m_2 > \omega_1+\omega_2,$ 
from which it follows that
$$(m_2+m_1)(\xi_2^3+\xi_1^3) - (\xi_2^2-\xi_1^2)^2 \ge \omega_2^4+\omega_1^4 - \omega_2^4-\omega_1^4+2\omega_2^2\omega_1^2>0.$$
In particular, $R(\xi)$ is decreasing for $\xi <0$ and increasing for $\xi >0,$ from which it follows immediately that $\xi = \pm E$ are the only roots. 

Finally, we observe that the nullvectors $(\tilde\mu(\xi),\tilde\rho(\xi))^t$ associated with $\xi=\pm E$ are both
\begin{align}\label{eq:vnull}
    v = \begin{pmatrix}-1 \\ \frac{1}{2m_2}-\frac{1}{2m_1}\end{pmatrix}.
\end{align}

Motivated by this, we change variables to remove the singularity captured by the nullvectors in \eqref{eq:vnull}. 
Namely, we introduce the new unknowns, $\sigma_1$ and $\sigma_2$ defined implicitly by
\begin{align}\label{eq:rhomu}
    \begin{pmatrix}\mu \\ \rho \end{pmatrix} =\cP_2 \begin{pmatrix}\sigma_1 \\ \sigma_2 \end{pmatrix},
\end{align}
with $\cP_2$ given by \eqref{eq:P2}.
As with $Q$ for the one-mass case, the Fourier transform of $Q_2$ has singularities at $\xi=\pm E$. 
Thus the singularities of $\cP_2$ exactly cancel the zeros of $\cL_2$.
Substituting \eqref{eq:rhomu} into the Fourier transformed boundary integral equations (\ref{eqn:ft_rho}) and (\ref{eqn:ft_mu}), we obtain an invertible system for $\sigma_1$ and $\sigma_2.$ In particular, we recover \eqref{eq:bif2}, where
$\cL_2 \cP_2$ is bounded with bounded inverse.

\begin{remark}
We do not prove two-mass analogues of the well-posedness results Theorems \ref{thm:invL}  and \ref{thm:invE}.
    In principle, similar arguments should extend to the $m_1 \ne m_2$ case, though the presentation would be more involved due to the presence of terms involving double layer potentials and their derivatives.
\end{remark}

\section{Proofs of main analytical results}\label{sec:proofs}
This section is devoted to proving the statements from Section \ref{sec:bie}, and is laid out as follows.
In Section \ref{subsec:thm:invL}, we first turn our attention to the proof of the uniqueness result for the integral equation with a single mass, Theorem \ref{thm:invL}. The proof of the other uniqueness result for the single mass integral equation, Theorem \ref{thm:invE}, requires that $u$ obtained by any smooth enough solution to the integral equation \eqref{eq:bif} solves the PDE (\ref{eqn:pde}, \ref{eqn:jump}, \ref{eqn:outgoing}, \ref{eqn:decayu}). Thus in Section \ref{subsec:reg}, we first prove regularity results for the solution of the integral equation, Theorem \ref{thm:reg0}, and that smooth solutions to the integral equation satisfy the PDE, Theorem \ref{thm:outgoing}. Finally, we prove
the main uniqueness result for the single mass integral equation, Theorem \ref{thm:invE}, in Section \ref{subsec:invE}. Theorem \ref{thm:outgoing2}, which proves that smooth solutions of the integral equation corresponding to bulk regions with two different masses satisfy the PDE is presented in Section \ref{subsec:twomasses}

Recall that the operators $\cL$ and $\cP$ are defined by
\begin{align*}
    \cL [\mu] (t) &= \mu (t) - \frac{m}{\pi}\int_\mathbb{R} K_0 (\omega|\gamma (t)-\gamma(t')|) \, \mu (t') \, {\rm d}t',\\
    \cP [\rho] (t) &= \rho (t) + \frac{m^2}E \int_{\Rm} e^{iE|t-t'|} \rho(t'){\rm d}t'.
\end{align*}
To distinguish between the arbitrary and flat cases, let $\flatL :L^2 (\mathbb{R}) \rightarrow L^2 (\mathbb{R})$ be defined by
\begin{align}\label{eq:flatL}
    \flatL [\mu] (t) := \mu (t) - \frac{m}{2\pi}\int_\mathbb{R} \int_{\mathbb{R}} \frac{e^{i\xi (t-t')}}{\sqrt{\xi^2+\omega^2}} {\rm d}\xi \,\mu(t')\, {\rm d}t', \qquad \mu\in L^2 (\mathbb{R}),
\end{align}
so that $\flatL$ is the operator $\cL$ above when the boundary is flat.
Recall 
that the integral equation \eqref{eq:bif} that we wish to solve is
\begin{align*}
    \cL \cP [\rho] (t) = 2m\,u_i (\gamma (t)), \qquad t \in \mathbb{R}
\end{align*}
which can be rewritten as
\begin{align}\label{eq:cM}
    (I + \cM)[\rho] = 2m (\flatL \cP)^{-1} u_i, \qquad \cM := (\flatL \cP)^{-1}(\cL-\flatL)\cP.
\end{align}
Our main results concern the solvability of the above integral equation on the space 
$L^2_\alpha := \{\rho \in L^2 (\mathbb{R}) : w_\alpha \rho \in L^2(\mathbb{R})\}$, where $w_\alpha (t) := e^{\alpha|t|}$. We require that $w_\alpha$ not grow too rapidly, which 
is captured by the assumption $0 < \alpha < \astar (\omega)$, where
\begin{align*}
    \astar (\omega) := \frac{\beta c \omega}{\beta + c\omega}
\end{align*}
and the constants $\beta$ and $c$ are defined by \eqref{eq:beta} and \eqref{eq:c}.
Observe that the function $\astar$ is monotonically increasing on the interval $(0,\infty)$; thus in order to ensure that \eqref{eq:alphamax} holds for all $\omega$ in a given set $S \subset (0,\infty)$, it suffices to verify \eqref{eq:alphamax} for the infimum of $S$. Since $\beta$ and $c\omega$ are both positive with $0 < c \le 1$, we also see that $\astar (\omega) < \min \{ \beta, c \omega\} \le \min \{ \beta, \omega\}$.

Much of the analysis in this section relies on well-known properties of the modified Bessel functions of the second kind $K_0 : (0,\infty) \to (0,\infty)$ and $K_1 = -K_0'$. These functions are smooth, monotonically decreasing and satisfy $K_j (r) \sim e^{-r}$ as $r \to \infty$. Moreover, $K_0 (r) \sim -\log r$ and $K_1 (r) \sim 1/r$ as $r \downarrow 0$.

Throughout the proofs in this section, we will use the notation $C$ to denote an arbitrary positive constant whose value might change from one line to another. That is, $a \le Cb$ is equivalent to $a \lesssim b$. We will use $\norm{\cdot}$ to denote the operator norm.
More specifically, if $\cB_1$ and $\cB_2$ are Banach spaces and $\cA : \cB_1 \rightarrow \cB_2$ is a linear operator, then $\norm{\cA} := \sup \{\norm{\cA \psi}_2 : \norm{\psi}_1 = 1\}$, where $\norm{\cdot}_j$ denotes the norm on $\cB_j$.

\subsection{Proof of Theorem \ref{thm:invL} \label{subsec:thm:invL}}
As in the statement of the theorem, suppose that 
$m_0 > 0$ and $E_0 \in (-m_0, m_0)\setminus \{0\}$ are fixed, with $\omega_0 := \sqrt{m_0^2 - E_0^2}$. The parameters $m$ and $E$ of our integral equation \eqref{eq:bif} will depend on the parameter $\lambda \in [1,\infty)$ via the relations $m=\lambda m_0$, $E=\lambda E_0$, and $\omega = \lambda \omega_{0}$.

We begin with the following important lemmas.
\begin{lemma}\label{lemma:int}
    Suppose $0 < \alpha < \astar(\omega_0)$, define $\g : \mathbb{R}^2 \to [0,\infty)$ by
    \begin{align*}
        \g (s,t) := \begin{cases}
    \min\{|s|,|t|\}, &s \cdot t > 0\\
    0, &\text{else}
    \end{cases}
    \end{align*} 
    and let $K_1 = - K_0'$ denote the modified Bessel function of the second kind.
    Then
    \begin{align*}
        \finint := \int_{\mathbb{R}^2} \Big( e^{\alpha |t|} e^{\eta |s|} &|t-s|^3e^{-\beta \g (s,t)} K_1 (c \omega_0 |t-s|)\Big)^2 {\rm d} t {\rm d} s < \infty
    \end{align*}
    for all $\eta < \astar (\omega_0) - \alpha$.
\end{lemma}
\begin{proof}
    We first observe that $\g (s,t) = \frac{1}{2} (|s+t| - |s-t|) 1\{ |s+t| - |s-t| > 0\}$, where $1\{\cdot \}$ is the indicator function. Moreover, the function $\mathbb{R} \ni \xi \mapsto |\xi|^3 K_1 (c \omega_0 |\xi|)$ is continuous and exponentially decaying at infinity with rate $c \omega_0$. It follows from the change of variables $(\xi,\zeta) = (t-s, t+s)$ in the integral that
    \begin{align*}
        \finint \le C \int_{\mathbb{R}^2} \Big( e^{\alpha_\eta|\xi|} e^{\alpha_\eta|\zeta|} e^{-\frac{\beta}{2} (|\zeta| - |\xi|) 1\{ |\zeta| - |\xi| > 0\}} e^{-c \omega_0 |\xi|}\Big)^2 {\rm d} \xi {\rm d} \zeta,
    \end{align*}
    where $\alpha_\eta := (\alpha + \eta)/2$.
    Fixing $\delta > 0$, we now 
    control the integral over the regions $U := \{|\zeta| < (1+\delta) |\xi|\}$ and $U^c := \mathbb{R}^2 \setminus U$ separately. We see that
    \begin{align*}
        \finintU :&= \int_{U} \Big( e^{\alpha_\eta |\xi|} e^{\alpha_\eta |\zeta|} e^{-\frac{\beta}{2} (|\zeta| - |\xi|) 1\{ |\zeta| - |\xi| > 0\}} e^{-c \omega_0 |\xi|}\Big)^2 {\rm d} \xi {\rm d} \zeta\\
        &\le \int_U \left( e^{\alpha_\eta (2+\delta) |\xi|} e^{-c \omega_0 |\xi|} \right)^2 {\rm d} \xi {\rm d} \zeta = 4 (1+\delta) \int_{0}^\infty \xi \left( e^{\alpha_\eta (2+\delta) \xi} e^{-c \omega_0 \xi} \right)^2 {\rm d} \xi,
    \end{align*}
    which is finite provided $\alpha_\eta < c\omega_0/(2+\delta)$. Moreover,
    \begin{align*}
        \finintUc :&= \int_{U^c} \Big( e^{\alpha_\eta |\xi|} e^{\alpha_\eta |\zeta|} e^{-\frac{\beta}{2} (|\zeta| - |\xi|) 1\{ |\zeta| - |\xi| > 0\}} e^{-c \omega_0 |\xi|}\Big)^2 {\rm d} \xi {\rm d} \zeta\\
        &\le \int_{U^c} \left( e^{\frac{\alpha_\eta (2+\delta)}{1+\delta} |\zeta|} e^{-\frac{\beta \delta}{2 (1+\delta)} |\zeta|} e^{-c \omega_0 |\xi|} \right)^2 {\rm d} \xi {\rm d} \zeta
    \end{align*}
    is finite so long as $\alpha_\eta < \frac{\beta \delta}{2 (2+\delta)}$. The optimal choice of $\delta$ thus satisfies
    \begin{align*}
        \frac{c\omega_0}{2+\delta} = \frac{\beta \delta}{2 (2+\delta)}.
    \end{align*}
    Hence we take $\delta = 2c\omega_0/\beta$ and conclude that $\finint \le C (\finintU + \finintUc)$ is finite whenever 
    $\alpha_\eta < \frac{\beta c \omega_0}{2 (\beta + c\omega_0)}$. 
    We have thus shown that if $\alpha< \astar (\omega_0)$, then $\finint$ is finite 
    for any $\eta < \astar (\omega_0) - \alpha$. This completes the result.
\end{proof}

\begin{lemma}\label{lemma:invLP} 
For all $0 < \alpha < \astar (\omega_0)$, 
the operator $(\cL_0 \cP)^{-1}$ is bounded on $L^2_\alpha$ with $\norm{(\cL_0 \cP)^{-1}} \le C$ uniformly in $\lambda \in [1,\infty)$.
\end{lemma}
\begin{proof}
Fix $0 < \alpha < \astar (\omega_0)$. By \eqref{eq:ainv} we know that 
$(\flatL \cP)^{-1}=(1+R_0) (1+R_1)$, where $R_0$ and $R_1$ are convolutions by $r_0 := mK_0 (\omega |\cdot|)$ and $r_1 :=\frac{(i-2)m^2}{2\zeta} e^{i \zeta |\cdot|}$, respectively, with $\zeta := a_+ + i a_-$ and 
\begin{align*}
    a_\pm := \sqrt{\frac{\sqrt{4m_0^4+E_0^4}\pm E_0^2}{2}}.
\end{align*}
The function $r_0$ has a logarithmic singularity at $0$, is smooth otherwise and decays exponentially at infinity with rate $\omega$. The exponential decay rate of $r_1$ is $a_- > \omega_0$.
Recalling the relation $\omega= \sqrt{m^2 - E^2} = \lambda \omega_{0}$ and the fact that $\astar (\omega_0) < \omega_0$, it follows that there exists $r \in L^1$ such that $e^{\alpha |\cdot|} r(\cdot) \in L^1$ and $|r_j (t)| \le \lambda r (\lambda t)$ for all $\lambda \in [1,\infty)$ and $j \in \{0,1\}$.
Using the identity $\norm{g_1 * g_2}_{L^2} \le \norm{g_1}_{L^1} \norm{g_2}_{L^2}$
with $g_1 (t) = \lambda e^{\alpha|t|} r(\lambda t)$ and $g_2(t) = e^{\alpha |t|}|\rho| (t)$, we obtain that for $\rho \in L^2_\alpha$ and $j \in \{0,1\}$,
\begin{align*}
    \norm{R_j [\rho]}^2_{L^2_\alpha} \le &\int_{\mathbb{R}} e^{2\alpha |t|} \Big(\int_{\mathbb{R}} \lambda r(\lambda (t-t')) \rho (t') {\rm d} t' \Big)^2 {\rm d} t \\
    \le
    &\int_{\mathbb{R}} \Big(\int_{\mathbb{R}} \lambda e^{\alpha|t-t'|} r(\lambda (t-t')) e^{\alpha |t'|}|\rho| (t') {\rm d} t' \Big)^2 {\rm d} t \le C \norm{\rho}_{L^2_\alpha}^2
\end{align*}
is bounded uniformly in $\lambda \in [1,\infty)$.
This completes the proof.
\end{proof}

\begin{remark}\label{remark:m2_1}
    The above proof 
    motivates the factor of $m^2$ appearing in the definition of $\cP$. If $m^2$ were replaced by $m^p$ for some $p \ne 2$, then the exponential decay rate of $r_1$ would become 
    \begin{align*}
        a_- = \sqrt{\frac{\sqrt{4m_0^p+E_0^4}- E_0^2}{2}},
    \end{align*}
    which is no longer necessarily greater than $\omega_0$. Thus we would not be able to establish boundedness of $(\cL_0 \cP)^{-1}$ on $L^2_\alpha$ for the full range of $\alpha$ we consider.
\end{remark}

For the following lemma, recall the definition of $\cM$ in \eqref{eq:cM}.
\begin{lemma}\label{thm:MHS}
For all $0 < \alpha < \astar (\omega_0)$ and $\lambda \in [1,\infty)$, 
the operator $\cM$ is Hilbert-Schmidt (hence compact) on $L^2_\alpha$.
\end{lemma}
\begin{proof}
Fix $0 < \alpha < \astar (\omega_0)$, $0 < \eta < \astar (\omega_0) - \alpha$ and $\lambda \in [1,\infty)$. 
Define $A_1: L^2 \rightarrow L^2_\alpha, A_2: L^2 \rightarrow L^2, A_3: L^2_\alpha \rightarrow L^2$ by
\begin{align*}
    A_1 = W_\alpha^{-1}, \qquad A_2 = W_\alpha (\flatL - \cL) W_\eta, \qquad A_3 = W_\eta^{-1} \cP,
\end{align*}
where $W_\alpha$ is point-wise multiplication by $w_\alpha$. 
Our strategy is to bound each term on the right-hand side of $\cM = (\flatL \cP)^{-1} A_1 A_2 A_3$.
We immediately have $\norm{A_1 [\rho]}_{L^2_\alpha} = \norm{\rho}_{L^2}$, so $A_1$ is bounded.
Using that
\begin{align}\label{eq:P}
    |\cP [\rho]| (t) \le |\rho| (t) + \frac{m^2}{|E|}\norm{\rho}_{L^1} \le |\rho| (t) + C \norm{\rho}_{L^2_\alpha},
\end{align}
we have
$
    \norm{A_3 [\rho]}_{L^2} \le C \norm{\rho}_{L^2_\alpha},
$
hence $A_3$ is also bounded.
By Lemma \ref{lemma:invLP}, $(\flatL \cP)^{-1}$ is bounded on $L^2_\alpha$.


It remains to show that $A_2$ is Hilbert-Schmidt.
To do so, we will prove exponential decay of the kernel of $\flatL - \cL$ using the decay properties of $K_0$ and the fact that our interface rapidly converges to a straight line at infinity \eqref{eq:beta}.
Let $\ell_\Delta (t,s) = K_0 (\omega |\gamma (t) - \gamma(s)|) - K_0 (\omega |t-s|)$ so that $-\frac{m}{\pi} \ell_\Delta (t,s)$ is the kernel of $\cL - \flatL$.
Since $K_0$ and $K_1 = -K_0'$ are both positive and monotonically decreasing, and $|\gamma (t) - \gamma (s)| \le |t-s|$, it follows that
\begin{align*}
    0 \le \ell_\Delta (t,s) =\int_{\omega |\gamma (t) - \gamma (s)|}^{\omega |t-s|} K_1 (z) {\rm d}z &\le \omega (|t-s| - |\gamma (t) - \gamma (s)|) K_1 (\omega |\gamma (t) - \gamma (s)|).
\end{align*}
Hence
\begin{align}\label{eq:lD1}
    0 \le \ell_\Delta (t,s) \le \omega (|t-s| - |\gamma (t) - \gamma (s)|) K_1 (c\omega |t-s|),
\end{align}
where $c>0$ is defined by \eqref{eq:c}.
The factor $K_1 (c \omega |t-s|)$ decays exponentially in $|t-s|$ but has a singularity at $t = s$.
To show that $\ell_\Delta$ is bounded and decays exponentially in $|t|+|s|$, we will control the difference $|t-s| - |\gamma (t) - \gamma (s)|.$

By smoothness of $\gamma$, for every $s, t \in \mathbb{R}$ there exist
$r_1, r_2 \in [\min\{s,t\}, \max\{s,t\}]$ such that
\begin{align}\label{eq:gamma}
    \gamma (t) = \gamma (s) + (t-s)\gamma'(s) + \frac{1}{2} (t-s)^2 \gamma'' (s) + \frac{1}{3!} (t-s)^3 (\gamma_1 ''' (r_1), \gamma_2 ''' (r_2)).
\end{align}
Since $|\gamma '| \equiv 1$ and thus $\gamma ' \cdot \gamma '' \equiv 0$, we have
\begin{align*}
    \gamma ' (s) \cdot (\gamma (t) -\gamma (s)) = t-s+ \frac{1}{3!} (t-s)^3 \gamma ' (s) \cdot (\gamma_1 ''' (r_1), \gamma_2 ''' (r_2)).
\end{align*}
Multiplying both sides by $t-s$ and using \eqref{eq:gamma} to expand $(t-s) \gamma ' (s)$, it follows that
\begin{align*}
    |\gamma (t) -\gamma (s)|^2 = (t&-s)^2 + \frac{1}{3!} (t-s)^4 \gamma ' (s) \cdot (\gamma_1 ''' (r_1), \gamma_2 ''' (r_2)) \\
    &+ \frac{1}{2} (t-s)^2 \gamma '' (s) (\gamma (t) - \gamma (s)) + \frac{1}{3!} (t-s)^3 (\gamma_1 ''' (r_1), \gamma_2 ''' (r_2)) \cdot (\gamma (t) - \gamma (s)).
\end{align*}
Again using \eqref{eq:gamma} to expand $\gamma (t) - \gamma (s)$ on the bottom line, we obtain that
$$(t-s)^2 \le |\gamma (t) -\gamma (s)|^2 + C \J(s,t) (t-s)^4, \qquad \J(s,t) := |\gamma ''|_{s,t}^2 + |\gamma '''|_{s,t}$$ for some $C > 0$, where the semi-norm $|\cdot|_{s,t}$ is defined by
\begin{align*}
    |\eta|_{s,t} := \norm{\eta}_{L^\infty [\min\{s,t\},\max\{s,t\}]}.
\end{align*}
This implies $0 \le |t-s| - |\gamma (t) - \gamma (s)| \le C \J(s,t) |t-s|^3$ for some possibly larger constant $C$.
From the rapid decay \eqref{eq:beta} of $|\gamma ''|^2 + |\gamma '''|$ at infinity,
it follows that
\begin{align*}
    0 \le |t-s| - |\gamma (t) - \gamma (s)| \le C |t-s|^3 e^{-\beta \g (s,t)}, \qquad \g (s,t) := \begin{cases}
    \min\{|s|,|t|\}, &st > 0\\
    0, &\text{else}.
    \end{cases}
\end{align*}
Using \eqref{eq:lD1}, we have 
\begin{align}\label{eq:bdl}
    0 \le \ell_\Delta (t,s)
    \le C \omega |t-s|^3e^{-\beta \g (s,t)} K_1 (c \omega |t-s|).
\end{align}
By Lemma \ref{lemma:int}, we conclude that the kernel of $A_2$ is square-integrable: $$\int_{\mathbb{R}^2} w^2_\alpha (t) w^2_\eta (s) \ell^2_\Delta (t,s) {\rm d} t {\rm d} s < \infty.$$
This means
$A_2$ is Hilbert-Schmidt and the proof is complete.
\end{proof}

Henceforth, $\norm{\cdot}_2$ denotes the Hilbert-Schmidt norm on $L^2_\alpha$.

\begin{remark}\label{rem:E0}
Since the constant $C$ in \eqref{eq:P} is proportional to $1/|E|$, we have shown that $\norm{\cM}_2 \le C/|E|$ as $E \rightarrow 0$. 
The singularity at $E = 0$ should not be surprising, as outgoing and incoming conditions at infinity are the same in this case. It turns out that an appropriate linear combination of outgoing and incoming conditions produces an operator that behaves better in the $E \rightarrow 0$ limit. Indeed,
we could have instead defined $Q : L^2_\alpha \to L^\infty$ by $Q [\rho] (t) = \frac{im^2}{2E}\int_{-\infty}^\infty (e^{iE|t-t'|} - e^{-iE|t-t'|}) \rho(t') {\rm d}t'$, which is bounded 
uniformly in $E$ (for any $\alpha > 0$). 
The resulting solution would now look like $\sin (E |t|)$ as $|t| \rightarrow \infty$.
\end{remark}

Lemma \ref{thm:MHS} implies that if \eqref{eq:cM} does not have a unique solution $\rho$, then 
the kernel of $1+\cM$ is a nontrivial finite-dimensional subspace of $L^2_\alpha$.
An extension of Lemma \ref{thm:MHS} is the following




\begin{prop}\label{thm:hs}
For any $0 < \alpha < \astar (\omega_0)$, $\norm{\cM}_2 \le C/\sqrt{\lambda}$ 
uniformly in $\lambda \in [1,\infty)$.
\end{prop}
\begin{proof}
Suppose $\lambda \ge 1$.
As in the proof of Lemma \ref{thm:MHS},
we will bound each factor on the right-hand side of $\cM = (\flatL \cP)^{-1} A_1 A_2 A_3$.
By Lemma \ref{lemma:invLP}, the operator norm of $(\flatL \cP)^{-1}$ on $L^2_\alpha$ is bounded uniformly in $\lambda$.
From \eqref{eq:P} it is clear that $\norm{A_3} \le C\lambda$.
Since $A_1$ is independent of $\lambda$, it remains to bound $\norm{A_2}_2$.

The kernel of $A_2$ is $k(t,s) = \frac{m}{\pi} w_\alpha (t) \ell_\Delta (t,s) w_\eta (s)$, with $\ell_\Delta (t,s) = K_0 (\omega |\gamma (t) - \gamma(s)|) - K_0 (\omega |t-s|)$ as before.
From \eqref{eq:bdl} it follows that
\begin{align*}
    k^2 (t,s) \le C (m\omega)^2
    w_\alpha^2 (t) w_\eta^2 (s) (t-s)^6 e^{-2\beta \g(s,t)} K_1^2 (c\omega |t-s|).
\end{align*}
Performing the change of variables $(\xi,\zeta) := (t-s,t+s)$ and using the fact that $|\zeta+\xi| + |\zeta-\xi| \le |\zeta+\lambda \xi| + |\zeta-\lambda \xi|$ and $\g (\frac{\zeta-\lambda \xi}{2},\frac{\zeta+\lambda\xi}{2}) \le \chi (\frac{\zeta-\xi}{2},\frac{\zeta+\xi}{2})$ for all $\lambda \ge 1$ and $\xi, \zeta \in \mathbb{R}$, it follows that
\begin{align*}
    \norm{k}^2_{L^2} &\le C (m \omega)^2 \int_{\mathbb{R}^2}
    w_\alpha^2 \Big(\frac{\zeta+\lambda\xi}{2} \Big) w_\eta^2 \Big(\frac{\zeta-\lambda\xi}{2}\Big) \xi^6 e^{-2\beta \g(\frac{\zeta-\lambda \xi}{2},\frac{\zeta+\lambda\xi}{2})} K_1^2 (c\omega_0 \lambda |\xi|) {\rm d} \zeta {\rm d} \xi\\
    &=\frac{C}{\lambda^7}(m \omega)^2 \int_{\mathbb{R}^2}
    w_\alpha^2 \Big(\frac{\zeta+\xi}{2} \Big) w_\eta^2 \Big(\frac{\zeta-\xi}{2}\Big) \xi^6 e^{-2\beta \g(\frac{\zeta-\xi}{2},\frac{\zeta+\xi}{2})} K_1^2 (c\omega_0 |\xi|) {\rm d} \zeta {\rm d} \xi.
\end{align*}
Since the above integral is finite (by Lemma \ref{lemma:int}) and independent of $\lambda$, we have shown that
$\norm{k}_{L^2} \le C/\lambda^{3/2}$ for all 
$\lambda \ge 1$.
Thus
\begin{align}\label{eq:cM_bd}
    \norm{\cM}_2 \le \norm{(\flatL \cP)^{-1}} \norm{A_1} \norm{A_2}_2 \norm{A_3} \le C/\sqrt{\lambda}
\end{align}
and the proof is complete.
\end{proof}

\begin{remark}\label{remark:m2_2}
    The above proof further motivates the factor of $m^2$ in the definition \eqref{eq:Q} of $Q$. In particular, if we instead were to define $Q$ by
    \begin{align*}
        Q[\rho](t) := \frac{m^p}E \int_{\Rm} e^{iE|t-t'|} \rho(t'){\rm d}t'
    \end{align*}
    with $p \in \mathbb{R}$, then 
    the symbol $a^{-1}$ of $(\cL_0 \cP)^{-1}$ in \eqref{eq:ainv} would be replaced by
    \begin{align*}
        a^{-1} (\xi) = \Big(1+ \frac{m}{\sqrt{\xi^2 + \omega^2}}\Big) \Big(1 + \frac{2im^p + m^2}{\xi^2 -E^2 - 2im^p}\Big),
    \end{align*}
    whose maximum absolute value over $\xi \in \mathbb{R}$ is $O(1 + \lambda^{2-p})$. 
    It would follow that
    $$\norm{(\cL_0 \cP)^{-1}}_{L^2_\alpha \to L^2_\alpha} = O(\lambda^{\max\{0,2-p\}}), \qquad \norm{A_3}_{L^2_\alpha \to L^2} = O (\lambda^{\max\{0,p-1\}}),$$
    with the second equality following immediately from \eqref{eq:P} with $m^2$ replaced by $m^p$.
    This means our bound in \eqref{eq:cM_bd} would get weaker if $p \notin [1,2]$.
\end{remark}

Proposition \ref{thm:hs} implies that $\norm{\cM} \rightarrow 0$ as $\lambda \rightarrow \infty$, with $\norm{\cM}$ the operator norm of $\cM$ on $L^2_\alpha$. 
However, it is possible to get faster decay of $\norm{\cM}$ in $\lambda$.

\begin{prop}\label{thm:op}
    For any $0 < \alpha < \astar (\omega_0)$, $\norm{\cM}\le C/\lambda$ 
uniformly in $\lambda \in [1,\infty)$.
\end{prop}
\begin{proof}
Recall Lemma \ref{lemma:invLP}, which states that $\norm{(\flatL \cP)^{-1}} \le C$ uniformly in $\lambda$. Thus it remains to bound the norm of $(\cL - \flatL)\cP$. We will bound the terms $\cL - \flatL$ and $(\cL-\flatL) Q$ separately.

We begin with the latter. We have that $\norm{Q [\rho]}_\infty \le \frac{C m^2}{E} \norm{\rho}_{L^2_\alpha}$, thus for $(\cL - \flatL) Q$ it suffices to bound $\cL - \flatL$ from $L^\infty$ to $L^2_\alpha$.
The kernel of $\cL - \flatL$ is $-\frac{m}{\pi} \ell_\Delta (t,s)$, where
$\ell_\Delta (t,s) = K_0 (\omega |\gamma (t) - \gamma(s)|) - K_0 (\omega |t-s|)$.
Using \eqref{eq:bdl},
it follows that for $\mu \in L^\infty$ we have
\begin{align*}
    |(\cL - \flatL) [\mu] |(t) &\le \frac{m}{\pi}\norm{\mu} \int_\mathbb{R} \ell_\Delta (t,s) {\rm d}s  \\
    &\le C \norm{\mu} \lambda^2 \int_{\mathbb{R}} |t-s|^3 e^{-\beta \g(s,t)} K_1 (c \omega |t-s|) {\rm d}s \\
    &=: C \norm{\mu}\lambda^2 I_{1},
\end{align*}
where $\norm{\mu} := \norm{\mu}_\infty$.
For concreteness, suppose $t \ge 0$. This means
\begin{align*}
    \g(s,t) = \begin{cases}
        0, & s< 0\\
        s, & 0 \le s \le t\\
        t, & s > t
    \end{cases}
\end{align*}
so that
\begin{align*}
    I_{1} = \int_{-\infty}^0 \kappa_\omega (t-s) {\rm d}s + \int_{0}^t \kappa_\omega (t-s) e^{-\beta s} {\rm d}s + \int_{t}^\infty \kappa_\omega (t-s) e^{-\beta t} {\rm d}s =: I_{2} + I_3 + I_4,
\end{align*}
where we have defined $\kappa_\omega (x) := |x|^3 K_1 (c \omega |x|)$.
Recall the definition \eqref{eq:alphamax} of $\astar$, which implies that $0 < \astar (\omega_0) < \min \{\beta, c \omega_0\}$. Since $\kappa_\omega$ is continuous and decaying exponentially at infinity with rate $c\omega$, it follows that $\kappa_{\omega_0} (x) \le C e^{-\astar(\omega_0) |x|}$ for all $x \in \mathbb{R}$.
Therefore, using the shorthand $\astar = \astar (\omega_0)$,
\begin{align*}
    I_2 &\le C \omega^{-4} \int_{t}^\infty e^{-\astar s} {\rm d} s = \frac{C}{\astar \omega^{4}} e^{-\astar t}\\ 
    I_3 &=\int_{0}^t \kappa_\omega (s) e^{-\beta (t-s)} {\rm d} s \le C \omega^{-4} \int_0^t e^{-\astar s}e^{-\beta (t-s)} {\rm d} s \\ &= \frac{C}{(\beta - \astar) \omega^4}e^{-\beta t} (e^{(\beta-\astar)t} - 1)\le C \omega^{-4} e^{-\astar t},\\
    I_4 &= e^{-\beta t}\int_{-\infty}^0 \kappa_\omega (s) {\rm d} s\le \frac{C}{\astar \omega^{4}} e^{-\beta t}
    \le C \omega^{-4} e^{-\astar t}
\end{align*}
uniformly in $t \ge 0$ and $\lambda \ge 1$. 
Repeating the same argument for $t<0$ and recalling that $\alpha < \astar$, it follows that $\norm{(\cL - \flatL) [\mu]}_{L^2_\alpha} \le \frac{C \norm{\mu}}{\lambda^2}$.
Thus we have shown that $\norm{(\cL - \flatL) Q} \le C/\lambda$ on $L^2_\alpha$.

It remains to bound $\cL - \flatL$ from $L^2_\alpha$ to itself. 
We already showed in the proof of Proposition \ref{thm:hs} that $A_2 = W_\alpha (\flatL - \cL) W_\alpha$ is Hilbert-Schmidt on $L^2$ with $\norm{A_2}_2 \le C/\lambda^{3/2}$. This means $\norm{\cL - \flatL}_{L^2_\alpha \rightarrow L^2_\alpha} \le C/\lambda^{3/2}$ and the proof is complete.
\end{proof}

Finally, we prove the 
first well-posedness result of Section \ref{sec:bie}.
\begin{theorem:invL}
    For any 
$0 < \alpha < \astar (\omega_0)$,
the integral equation \eqref{eq:bif} admits a unique solution $\rho \in L^2_\alpha$
for all but a finite number of $\lambda \in [1,\infty)$.
\end{theorem:invL}
\begin{proof}
Fix $0 < \alpha < \astar (\omega_0)$ and $\delta > 0$ as small as necessary. The proof of Lemma \ref{thm:MHS} easily extends to complex $\lambda$ in a sufficiently small neighborhood of the real axis. Thus $\cM:L^2_\alpha \rightarrow L^2_\alpha$ is compact for all $\lambda \in [1,\infty) \times (-\delta, \delta) \subset \mathbb{C}$. Since $\cM$ is holomorphic in $\lambda$,
we apply Lemma \ref{thm:MHS}, Proposition \ref{thm:op} and Kato perturbation theory \cite[Theorem VII.1.9]{kato2013perturbation} to prove that $\cM$ has an eigenvalue of $-1$ for only a finite number of $\lambda \in [1,\infty)$. The result then follows from \eqref{eq:cM}.
\end{proof}




\subsection{Proofs of Theorems \ref{thm:reg0} and \ref{thm:outgoing} \label{subsec:reg}}
We first prove Theorem \ref{thm:reg0}, a regularity result for functions $\rho$ satisfying our boundary integral equation $\cL \cP \rho = 2mu_i \circ \gamma$. Recall the definition of the weighted Sobolev spaces $H^s_\alpha$ provided in~\eqref{eq:hsalpha}.

\begin{theorem:reg0}
Let $\so \in \mathbb{N} \cup \{0\}$ and $0 < \alpha < \astar (\omega)$. If the function $\rho \in L^2_\alpha$ satisfies $\cL \cP \rho\in H^\so_\alpha$, then $\rho\in H^\so_\alpha$. In particular, suppose $f_j \in C_c (\Omega_j)$ for $j=1,2$, then any function $\rho \in L^2_\alpha$ satisfying \eqref{eq:bif} also belongs to $\cap_{s \in \mathbb{N}} H_\alpha^s$.
\end{theorem:reg0}
\begin{proof}
We write
    \begin{align}\label{eq:3terms}
        \rho = (I-\cL_0 \cP) [\rho] + (\cL_0 - \cL) \cP [\rho] + \cL \cP [\rho].
    \end{align}
    Recalling \eqref{eq:symbolKP}, we see that in the Fourier domain, the first term on the above right-hand side is multiplication by the function
    \begin{align*}
        1 - \Big( 1 - \frac{m}{\sqrt{\xi^2 + \omega^2}}\Big)\Big(1&- \frac{2im^2}{\xi^2 - E^2}\Big)\\
        &= \frac{m}{\sqrt{\xi^2 + \omega^2}} + \frac{2im^2}{\xi^2 - E^2} - \frac{2im^3}{\sqrt{\xi^2 + \omega^2}(\xi^2 - E^2)}\\
        &= \frac{m}{\sqrt{\xi^2 + \omega^2}} + \frac{2im^2}{\sqrt{\xi^2 + \omega^2}(\sqrt{\xi^2 + \omega^2} + m)},
    \end{align*}
    which is holomorphic in $\xi \in \mathbb{R} \times (-\omega, \omega) \subset \mathbb{C}$ 
    and, for any fixed $-\omega < \Im \xi < \omega$, decays like $1/|\xi|$ as $\Re \xi \to \pm \infty$. Since $\alpha < \astar (\omega) < \omega$, it follows that $I - \cL_0 \cP: H^j_\alpha \to H^{j+1}_\alpha$ is bounded for any $j \in \mathbb{N} \cup \{0\}$.
    
To control the second term on the right-hand side of \eqref{eq:3terms}, we first
    recall that 
    the operator $Q$ defined in \eqref{eq:Q} acts as convolution with the function $t \mapsto \frac{m^2}{E} e^{iE|t|}$, 
    whose derivative $i m^2 \sgn (t) e^{iE |t|}$ is uniformly bounded in $t$. It follows that $Q:H^j_\alpha \to C^{j+1}$ defines a bounded map 
    for any $j \in \mathbb{N}\cup \{0\}$.

Therefore, in order to prove that $(\cL_0 - \cL) \cP = (\cL_0 - \cL) (I+Q)$ defines a bounded map from $H^{j}_\alpha$ to $H^{j+1}_\alpha$, it suffices to
    show that $\cL_0 - \cL : \hc^{j} \to H^{j+1}_\alpha$ is bounded when $\hc^j \in \{H^j_\alpha, C^j \}$. The kernel of $\cL_0 - \cL$ is proportional to $\ell_\Delta (t,s) = K_0 (\omega |\gamma (t) - \gamma(s)|) - K_0 (\omega |t-s|)$. 
    Using the regularity of $\gamma$, one can verify that for any $j \in \mathbb{N} \cup \{0\}$, the function $\ell_\Delta^{(j)} (t,s) := (\partial_t + \partial_s)^j \ell_\Delta (t,s)$ is continuously differentiable and exponentially decaying in both variables. More specifically, the arguments used to obtain \eqref{eq:bdl} extend to show that
    \begin{align}\label{eq:bdlj}
        |\ell_\Delta^{(j)} (t,s)| + |\partial_t \ell_\Delta^{(j)} (t,s)|
    \le C e^{-\beta \g (s,t)} e^{-c \omega |t-s|}
    \end{align}
    for some positive constant $C$ depending on $j$, where we recall the definition
    \begin{align*}
        \g (s,t) := \begin{cases}
    \min\{|s|,|t|\}, &st > 0\\
    0, &\text{else}.
    \end{cases}
    \end{align*}

    It follows by induction that for any $j$,
    \begin{align}\label{eq:ind}
        \partial^{j}_t \int_\mathbb{R} \ell_\Delta (t,s) \mu (s) {\rm d} s = \int_\mathbb{R} \sum_{i=0}^j \binom{j}{i} \ell_\Delta^{(j-i)} (t,s) \mu^{(i)} (s) {\rm d} s.
    \end{align}
    Indeed, \eqref{eq:ind} trivially holds when $j=0$, while assuming that \eqref{eq:ind} holds for some $j \in \mathbb{N} \cup \{0\}$ implies that
    \begin{align*}
        \partial^{j+1}_t \int_\mathbb{R} \ell_\Delta (t,s) \mu (s) {\rm d} s &= \partial_t \int_\mathbb{R} \sum_{i=0}^j \binom{j}{i} \ell_\Delta^{(j-i)} (t,s) \mu^{(i)} (s) {\rm d} s \\
        &= \int_\mathbb{R} \sum_{i=0}^j \binom{j}{i} \ell_\Delta^{(j-i+1)} (t,s) \mu^{(i)} (s) {\rm d} s\\
        &\qquad +
        \int_\mathbb{R} \sum_{i=0}^j \binom{j}{i} \ell_\Delta^{(j-i)} (t,s) \mu^{(i+1)} (s) {\rm d} s\\
        &= \int_\mathbb{R} \sum_{i=0}^{j+1} \binom{j+1}{i} \ell_\Delta^{(j+1-i)} (t,s) \mu^{(i)} (s) {\rm d} s,
    \end{align*}
    where we have integrated by parts in $s$ to obtain the second equality and used the identity $\binom{j+1}{i} = \binom{j}{i} + \binom{j}{i-1}$ to get the last equality.
    Again using the regularity of $\ell_\Delta^{(j)}$, it follows from \eqref{eq:ind} that
    \begin{align}\label{eq:kerl}
        \partial_t^{j+1} \int_\mathbb{R} \ell_\Delta (t,s) \mu (s) {\rm d} s &=\int_\mathbb{R} \sum_{i=0}^j \binom{j}{i} \partial_t \ell_\Delta^{(j-i)} (t,s) \mu^{(i)} (s) {\rm d} s.
    \end{align}
In the proof of Proposition \ref{thm:op}, we used the exponential decay of $\ell_\Delta$ to show that $\cL_0 - \cL : S^0 \to L^2_\alpha$ is bounded when $S^0 \in \{L^2_\alpha, C^0\}$. By \eqref{eq:bdlj}, the same argument establishes boundedness from $S^0$ to $L^2_\alpha$ of the integral operator with kernel $(t,s) \mapsto \partial_t \ell_\Delta^{(j-i)} (t,s)$.
It follows from \eqref{eq:kerl} that $\cL_0 - \cL$ defines a bounded map from $\hc^j$ to $H^{j+1}_\alpha$, as desired. 

The above paragraphs show that if $\rho \in H^{j}_\alpha$, then 
    $\rho$ can be written as a sum of functions in $H^{j+1}_\alpha$ and $H^\so_\alpha$. Since it is assumed that $\rho \in L^2_\alpha = H^0_\alpha$, we conclude by induction that $\rho \in H^\so_\alpha$.
    
Finally, the assumption on the $f_j \in C_{c}(\Omega_{j})$ implies that the right-hand side of \eqref{eq:bif} is in $\cap_{s \in \mathbb{N}} H_\alpha^s$ and the proof is complete.
\end{proof}

\begin{thm:outgoing}
    Suppose $f_j \in C_c (\Omega_j)$ for $j=1,2$ and let $\alpha > 0$.
Suppose $\rho \in L^2_\alpha$ satisfies \eqref{eq:bif}, and set
\begin{align*}
    \mu := \cP [\rho], \qquad u_s := S_\omega [\mu], \qquad u := u_i + u_s.
\end{align*}
Then \eqref{eqn:pde}, \eqref{eqn:jump}, \eqref{eqn:outgoing} and \eqref{eqn:decayu} hold.
\end{thm:outgoing}

\begin{proof}
That \eqref{eqn:decayu} holds is a consequence of the exponential decay of $K_0 (\omega |x-\gamma (t)|)$ in $|x-\gamma (t)|$ and the Lebesgue dominated convergence theorem. The latter also implies that \eqref{eqn:pde} holds, where we recall that $(-\Delta + \omega^2) u_i = f$ in $\mathbb{R}^2 \setminus \Gamma$ (with $f$ the function equal to $f_j$ in $\Omega_j$); see \eqref{eq:freePDE}--\eqref{eq:gf}. 
To verify the jump conditions \eqref{eqn:jump}, we use the assumption $f_j \in C_c (\Omega_j)$ 
to conclude by
Theorem \ref{thm:reg0} that $\rho \in \cap_{s \in \mathbb{N}} H^s_\alpha$, which in turn implies that 
\begin{align}\label{eq:mureg}
\mu = \cP [\rho] = \rho + Q [\rho] \in \left(\cap_{s \in \mathbb{N}} H^s_\alpha\right) + \left(C^\infty \cap L^\infty \right).
\end{align}
Applying standard results in integral equations \cite{colton2013integral}, we conclude that the limits
of $\hat{n}(t) \cdot \nabla S_\omega[\mu](\gamma(t) \pm s\,\hat{n}(t))$ as $s \to 0^+$ exist and are given by \eqref{eq:jumpGradS}.
Using that $S_\omega$ is continuous across $\Gamma$ with
$
    \lim_{s \to 0} S_\omega[\mu](\gamma(t) + s\,\hat{n}(t)) = \cS_\omega [\mu] (t),
$
the jump conditions \eqref{eqn:jump} then follow from the derivation of the integral equation \eqref{eq:mu}.

Let us now prove the outgoing condition \eqref{eqn:outgoing}.
Since $f \in C_c (\mathbb{R}^2)$, we 
know that $u_i$ is smooth outside a compact set with $u_i$ and all its derivatives decaying rapidly at infinity. Hence $u_i$ satisfies \eqref{eqn:outgoing}, so it remains to consider $u_s$.

We follow a similar strategy to the previous proofs. Fix $r \in \mathbb{R}$ and define 
\begin{align*}
    \ell_{\Delta, r}(t,s) := K_0 (\omega |\gamma (t) + r\hat{n}(t) - \gamma (s)|) - K_0 \left(\omega \sqrt{r^2 + (t-s)^2}\right),
\end{align*}
the difference kernel between the arbitrary and flat-interface geometries. Extending 
\eqref{eq:bdlj} to the $r \ne 0$ case, it follows that $\ell_{\Delta,r}$ and $\partial_t \ell_{\Delta, r}$ decay exponentially in both variables. Moreover, our assumptions \eqref{eq:beta} and \eqref{eq:gammainfty} on $\gamma$ ensure that for all $|t|$ sufficiently large, $\gamma (t) + r \hat{n} (t) \ne \gamma (s)$ for all $s \ne t$. Using the regularity \eqref{eq:mureg} of $\mu$, we conclude that
\begin{align*}
    \left| (\partial_t \pm iE) \int_{\mathbb{R}} \ell_{\Delta, r} (t,s) \mu (s) {\rm d} s \right| \le \left|\int_{\mathbb{R}} \partial_t \ell_{\Delta, r} (t,s) \mu (s) {\rm d} s \right| + |E|\left|\int_{\mathbb{R}} \ell_{\Delta, r} (t,s) \mu (s) {\rm d} s \right|
\end{align*}
goes to zero as $|t| \to \infty$. It thus suffices to prove the outgoing condition for the flat-interface $\gamma (t) = (t,0)$. In this case, 
\begin{align*}
    u_{s}(\gamma(t) + r \hat{n} (t)) = \frac{1}{2 \pi}\int_{\mathbb{R}} K_0 \left(\omega \sqrt{r^2 + (t-s)^2}\right) \mu (s) {\rm d} s
\end{align*}
is the convolution of two functions, hence 
\begin{align}\label{eq:us_flat}
    \partial_t u_{s}(\gamma(t) + r \hat{n} (t)) = \frac{1}{2 \pi}\int_{\mathbb{R}} K_0 \left(\omega \sqrt{r^2 + (t-s)^2}\right) \mu' (s) {\rm d} s
\end{align}
by the regularity of $K_0$ and $\mu = \rho + Q [\rho]$.


Above we verified that $\rho \in \cap_{s \in \mathbb{N}} H^s_\alpha$, while the fact that $\rho \in L^1$ implies that
\begin{equation*}
\begin{aligned}
    (\partial_t - iE) Q [\rho] (t) &= im^2 \int_{-\infty}^\infty (\sgn (t-t') -1) e^{iE|t-t'|} \rho (t'){\rm d} t' \\
     &= -2im^2 \int_{t}^\infty e^{-iE(t-t')} \rho (t') {\rm d}t'
\end{aligned}
\end{equation*}
goes to $0$ as $t \rightarrow \infty$. A parallel calculation reveals that $(\partial_t + iE) Q [\rho] (t) \to 0$ as $t \to -\infty$. It follows from the 
exponential decay of $K_0$ that
\begin{align}\label{eq:mu_out}
    \int_{\mathbb{R}} K_0 \left(\omega \sqrt{r^2 + (t-s)^2}\right)(\partial_s \mp iE)\mu (s){\rm d}s \longrightarrow 0
\end{align}
as $t \to \pm \infty$. With \eqref{eq:us_flat} and \eqref{eq:mu_out}, we have shown that $u_s$ satisfies \eqref{eqn:outgoing}. This completes the proof.
\end{proof}

\subsection{Proof of Theorem \ref{thm:invE} \label{subsec:invE}}
We now turn our attention to our second well-posedness result, Theorem \ref{thm:invE}.
As opposed to Theorem \ref{thm:invL}, this result 
does not require $\cM$ to be small in any limit. Instead, we show that $-1$ belongs to the resolvent set of $\cM$ 
whenever $\Im E >0$. To this end, we state the following

\begin{lemma}\label{lemma:htrivial}
Fix $m > 0$.
For all $E \in \mathbb{C}$ with $0 < |\Re E| < m$ and $\Im E \ne 0$,
if $u \in H^1 (\mathbb{R}^2)$ solves \eqref{eq:pdeh} below, then $u \equiv 0$.
\end{lemma}
\begin{proof}
Integrating by parts, we see that
\begin{align*}
    \int_{\Omega_j} |\nabla u|^2 {\rm d} x = -\omega^2 \int_{\Omega_j} |u|^2 {\rm d} x + \eta_j\int_\Gamma u \hat{n} \cdot \nabla u {\rm d}t,
\end{align*}
where $\eta_1 = 1$ and $\eta_2 = -1$.
Taking the sum over $j$ and using the second line of \eqref{eq:pdeh}, we obtain
\begin{align*}
    \int_{\mathbb{R}^2}|\nabla u|^2 {\rm d} x + \omega^2 \int_{\mathbb{R}^2} |u|^2 {\rm d} x = 2m\int_\Gamma |u|^2{\rm d}t
\end{align*}
Taking the imaginary part of both sides, we 
conclude that $\int_{\mathbb{R}^2} |u|^2 {\rm d} x = 0$
and the result is complete.
\end{proof}

Using the above lemma, we now prove the following. 

\begin{theorem:invE}
    Fix $m>0$.
For any $0 < \eps_1, \eps_2 < m$ and
$0<\alpha<\astar\left(\sqrt{2\eps_2 m - \eps_2^2}\right)$, 
the integral equation \eqref{eq:bif} admits a unique solution $\rho \in L^2_\alpha$
for all but a finite number of $E \in [-m+\eps_2, -\eps_1] \cup [\eps_1, m-\eps_2]$.
\end{theorem:invE}
\begin{proof}
Fix $\eps := (\eps_1, \eps_2) \in (0, m)^2$ and define $$Z_\eps := [-m+\eps_2, -\eps_1] \cup [\eps_1, m-\eps_2].$$
Lemma \ref{thm:MHS} implies that for any $E \in Z_\eps$, the operator $\cM$ is Hilbert-Schmidt on $L^2_\alpha$ provided $0 < \alpha < \astar (\omega)$, where $\omega := \sqrt{m^2 - E^2}$. Since $\astar$ is monotonically increasing on the positive reals, its minimum value over $E \in Z_\eps$ is $\sqrt{m^2 - (m-\eps_2)^2} = \sqrt{2\eps_2 m - \eps_2^2}$.
It follows that for any $0 < \alpha < \astar \left(\sqrt{2\eps_2 m - \eps_2^2}\right)$, $\cM$ is Hilbert-Schmidt on $L^2_\alpha$ for all $E \in Z_\eps$. The above arguments easily extend to complex-valued $E$ in a sufficiently small neighborhood of the real axis. Therefore, $\cM: L^2_\alpha \rightarrow L^2_\alpha$ is Hilbert-Schmidt
for all $E \in Z_{\eps} \times (-\delta, \delta) \subset \mathbb{C}$ provided $\delta>0$ is sufficiently small.

Fix $E \in Z_{\eps} \times (0, \delta)$. 
By contradiction suppose $\cM$ has an eigenvalue of $-1$.
Then there exists $0 \ne \rho \in L^2_\alpha$ such that $(1+\cM) [\rho] = 0$. Letting $\mu := \cP [\rho]$, it follows from Theorem \ref{thm:outgoing} (which easily extends to the $\Im E > 0$ case) that $u = S_\omega [\mu]$ solves 
the homogeneous problem
\begin{align}\label{eq:pdeh}
    \begin{cases}
    \Delta u(x) - \omega^2 u(x) = 0,\quad x\in \mathbb{R}^2 \setminus \Gamma,\\
    \left[[\hat{n}\cdot \nabla u \right]](\gamma(t)) = -2m u(\gamma(t)), \quad t \in \mathbb{R},\\
        \left[[u\right]](\gamma(t)) = 0, \quad t \in \mathbb{R},
    \end{cases}
\end{align}
where $\omega^2 = m^2 - E^2$.
Since $E$ has positive imaginary part,
$Q$ is bounded from $L^2_\alpha$ to itself, meaning that $\mu \in L^2_\alpha$.
It follows that $u$ and $\nabla u$ decay rapidly at infinity; for example, $u \in H^1 (\mathbb{R}^2)$. 
Lemma \ref{lemma:htrivial} then implies that $S_\omega [\mu] \equiv 0$;
hence $\mu \equiv 0$ by \eqref{eq:mureg} and \eqref{eq:jumpGradS}. (The arguments used to obtain \eqref{eq:mureg} when $E \in \mathbb{R}$ easily extend to the case where $\Im E > 0$.)
It follows from \eqref{eq:fp} that
$0 = \tilde{\mu}(\xi) = (1-\frac{2im^2}{\xi^2-E^2}) \tilde{\rho} (\xi)$,
from which we conclude that $\rho \equiv 0$, a contradiction. 
We have thus shown that $-1$ is not an eigenvalue of $\cM$ whenever $E \in Z_{\eps} \times (-\delta, \delta)$ has positive imaginary part. Since $\cM$ depends holomorphically on $E$,
this means we can again apply
\cite[Theorem VII.1.9]{kato2013perturbation} to complete the result.
\end{proof}

\begin{remark}
It is crucial that $Z_{\eps}$ be bounded away from $0$ as $Q$ is not holomorphic there. Thus we cannot guarantee that the number of ``bad'' $E$-values near $0$ is finite. As suggested by Remark \ref{rem:E0}, one could redefine $Q$ by
$Q [\rho] (t) = \frac{i}{2E}\int_{-\infty}^\infty (e^{iE|t-t'|} - e^{-iE|t-t'|}) \rho(t') {\rm d}t'$ to make it holomorphic, but then the outgoing conditions \eqref{eqn:outgoing} would not be satisfied.
\end{remark}

\subsection{Proof of Theorem \ref{thm:outgoing2} \label{subsec:twomasses}}
We conclude this section by proving that smooth solutions to the integral equation corresponding to two masses, satisfies the PDE (\ref{eqn:pde2}, \ref{eqn:jump2}, \ref{eqn:outgoing2}, \ref{eqn:decayu2}).
\begin{thm:outgoing2}
Suppose $f_j \in C_c (\Omega_j)$ for $j=1,2$ and let $\alpha > 0$.
Suppose $\sigma \in \cap_{s \in \mathbb{N}} (H_\alpha^s)^2$ is a solution of \eqref{eq:bif2}, and set
\begin{align*}
    \begin{pmatrix}
    \mu \\ \rho
    \end{pmatrix} := \cP_2 [\sigma], \qquad 
    u_s(x) = \begin{cases}
    D_{\omega_2}[\rho](x)+S_{\omega_2}[\mu](x), \quad & x \in \Omega_2,\\
     D_{\omega_1}[\rho](x)+S_{\omega_1}[\mu](x), \quad & x \in \Omega_1.\\
\end{cases} 
\end{align*}
Then $u:=u_{i} + u_{s}$ satisfies \eqref{eqn:pde2}, \eqref{eqn:jump2}, \eqref{eqn:outgoing2} and \eqref{eqn:decayu2}.
\end{thm:outgoing2}

\begin{proof}
As before, \eqref{eqn:decayu2} and \eqref{eqn:pde2} follow immediately from dominated Lebesgue and the definition \eqref{eq:ui} of $u_i$ in terms of the free-space Green's function. 
The derivation of (\ref{eqn:BIE_1}, \ref{eqn:BIE_2}) and assumed regularity of $\sigma$ imply \eqref{eqn:jump2}. 
We now prove the outgoing condition, 
namely that
\begin{align}\label{eqn:outgoing_pf}
    \lim_{t \rightarrow \pm \infty} (\pm\partial_t - iE) u (\gamma (t) + r \hat{n} (t))&=0
\end{align}
for all $r \in \mathbb{R}$.
Since $u_i$ and $u_s$ are no longer continuous across $\Gamma$, we handle the $r \ne 0$ and $r=0$ cases separately. 

Let us first assume that $r \ne 0$.
Recalling the definition of $(\mu,\rho)$ in \eqref{eq:mrus},
the same arguments from Theorem \ref{thm:outgoing} imply that $\mu = \mu_0 + \mu_1$ and $\rho = \rho_0 + \rho_1$, where 
$\mu_0, \rho_0 \in \cap_{s \in \mathbb{N}} H_\alpha^s$ and $\mu_1, \rho_1 \in C^\infty \cap L^\infty$ satisfy $\lim_{t \pm \infty} (\partial_t \mp iE) \mu_1 (t) = 0$ and $\lim_{t \pm \infty} (\partial_t \mp iE) \rho_1 (t) = 0$.
Hence (following the proof of Theorem \ref{thm:outgoing}) $u_i + S_{\omega_j} [\mu]$ satisfies \eqref{eqn:outgoing_pf} for $j \in \{1,2\}$. It remains to analyze $D_{\omega_j} [\rho]$.

We see that $D_{\omega_j} [\rho] (\gamma (t) + r \hat{n} (t)) =\frac{\omega_j}{2\pi} \int_\mathbb{R} k_{j,r} (t,s) \rho (s) {\rm d} s$, where
\begin{align*}
    k_{j,r} (t,s) :=
    \hat{n} (s) \cdot
    \frac{\gamma (t) + r \hat{n} (t) - \gamma (s)}{|\gamma (t) + r \hat{n} (t) - \gamma (s)|} K_1 (\omega_j |\gamma (t) + r \hat{n} (t) - \gamma (s)|).
\end{align*}
Define the reduction of $k_{j,r}$ to the flat interface by 
\begin{align*}
    k_{j,r}^0 (t,s) :=
    \frac{r}{\sqrt{r^2 + (t-s)^2}} K_1 \left( \omega_j \sqrt{r^2 + (t-s)^2}\right)
\end{align*}
and set $\ell_{\Delta, j,r} := k_{j,r} - k^0_{j,r}$. Since $r \ne 0$, the assumptions \eqref{eq:beta} and \eqref{eq:gammainfty} on $\gamma$ ensure that the restriction of $\ell_{\Delta, j, r} (t,s)$ to $|t|$ sufficiently large is smooth in both $t$ and $s$. As before, the exponential decay of $K_1$ 
implies that both $k_{j,r} (t,s)$ and $k^0_{j,r} (t,s)$ decay exponentially in $|t-s|$, while our assumption \eqref{eq:beta} that $\gamma$ is approximately flat at infinity 
gives us exponential decay of $\ell_{\Delta, j ,r} (t,s)$ in $\min\{|t|,|s|\}$ when $t$ and $s$ have the same sign. The same decay properties apply to derivatives of $\ell_{\Delta, j ,r}$. We thus conclude that $\ell_{\Delta, j,r}$ and $\partial_t \ell_{\Delta, j, r}$ decay exponentially in both variables, meaning that
\begin{align*}
    \left| (\partial_t \pm iE) \int_{\mathbb{R}} \ell_{\Delta, j,r} (t,s) \rho (s) {\rm d} s \right| \le \left|\int_{\mathbb{R}} \partial_t \ell_{\Delta,j, r} (t,s) \rho (s) {\rm d} s \right| + \left| E\int_{\mathbb{R}} \ell_{\Delta,j, r} (t,s) \rho (s) {\rm d} s \right|
\end{align*}
goes to zero as $|t| \to \infty$.

To complete the proof for $r \ne 0$, we must therefore show that 
$D_{\omega_j} [\rho]$ satisfies \eqref{eqn:outgoing_pf} when the interface is flat. In this case, $D_{\omega_j} [\rho] (\gamma (t) + r \hat{n} (t))$ is the convolution of the function $\frac{\omega_j}{2\pi} k^0_{j,r} (\cdot, 0)$ with $\rho$, and thus
\begin{align}\label{eq:D}
    (\partial_t \mp iE) D_{\omega_j} [\rho] (\gamma (t) + r \hat{n} (t)) =\frac{\omega_j}{2\pi} \int_{\mathbb{R}} k^0_{j,r} (t,s) (\rho' (s) \mp iE \rho (s)) {\rm d} s. 
\end{align}
By the decay properties of $\rho = \rho_0 + \rho_1$ and $k^0_{j,r}$ established above, we conclude that \eqref{eq:D} indeed goes to zero as $|t| \to \infty$.

We now verify \eqref{eqn:outgoing_pf} 
when $r=0$. By the regularity of $\mu$ and $\rho$, we can apply \eqref{eq:u_Gamma} which reads
\begin{align*}
    u(\gamma(t)) = \frac{1}{2} \left[\mD_{\omega_2}[\rho](t)+\mD_{\omega_1}[\rho](t)\right]+\frac{1}{2}\Bigg[\mS_{\omega_2}[\mu](t)+\mS_{\omega_1}[\mu](t) \Bigg] + u_i(\gamma (t)),
\end{align*}
where $u_i (\gamma(t)) := \lim_{r \downarrow 0}  \{ u_i (\gamma (t) + r \hat{n} (t)) + u_i (\gamma (t) - r \hat{n} (t)) \}$ and we recall that 
\begin{align*}
    \cS_\omega [\mu] (t)&:= \frac{1}{2\pi} \int_\mathbb{R} K_0 (\omega |\gamma (t)-\gamma(t')|) \,\mu(t')\, {\rm d}t'
\end{align*}
while
$
    \mD_{\omega} [\rho] (t) := \frac{\omega}{2\pi} \int_\mathbb{R} k_{\omega} (t,s) \,\rho(s)\, {\rm d}s
$ with
\begin{align*}
    k_\omega (t,s) := \hat{n} (s) \cdot
    \frac{\gamma (t) - \gamma (s)}{|\gamma (t) - \gamma (s)|} K_1 (\omega |\gamma (t) - \gamma (s)|).
\end{align*}
The proof of Theorem \ref{thm:outgoing} verifies \eqref{eqn:outgoing_pf} for the $\mS_{\omega_j}$ and $u_i$ so that
\begin{align*}
    \lim_{t \to \pm \infty} (\pm\partial_t - iE) \Bigg( \frac{1}{2}\Big[\mS_{\omega_2}[\mu](t)+\mS_{\omega_1}[\mu](t) \Big] + u_i(\gamma (t)) \Bigg) = 0.
\end{align*}
It remains to show that the $\mD_{\omega_j}$ satisfy the outgoing condition.
To this end, we 
use the fact that $\hat{n} (s) \cdot \gamma'(s) = 0$ to obtain
\begin{align*}
    \left | \hat{n} (s) \cdot (\gamma (t) - \gamma (s)) \right | \le C |t-s|^2 \norm{\gamma''}_{L^\infty [\min\{s,t\},\max\{s,t\}]}
\end{align*}
uniformly in $t$ and $s$. It follows from familiar arguments that the $k_{\omega_j}$ are continuous and exponentially decaying in both variables. We similarly verify by a direct calculation that the same is true for the $\partial_t k_{\omega_j}$. As before, we use the fact that $\rho = \rho_0 + \rho_1$ with $\rho_0 \in \cap_{s \in \mathbb{N}} H^s_\alpha$ and $\rho_1 \in L^\infty$ to conclude that
\begin{align*}
    \left| (\partial_t \pm iE) \int_{\mathbb{R}} k_{\omega_j} (t,s) \rho (s) {\rm d} s \right| \le \left|\int_{\mathbb{R}} \partial_t k_{\omega_j} (t,s) \rho (s) {\rm d} s \right| + \left| E\int_{\mathbb{R}} k_{\omega_j} (t,s) \rho (s) {\rm d} s \right|
\end{align*}
goes to zero as $|t| \to \infty$. This completes the result.
\end{proof}

\section{Numerical apparatus}\label{sec:num}
In this section we describe an algorithm for solving equations (\ref{eqn:pde}, \ref{eqn:jump}, \ref{eqn:outgoing}, \ref{eqn:decayu}) via the boundary integral equations \eqref{eq:bif} and \eqref{eq:ubif}. 
In subsection \ref{subsec:disc} we briefly discuss the details of the discretization used. Following this, in subsection \ref{subsec:accel} we describe several accelerations that were made to improve the computational efficiency.

\subsection{The discretization}\label{subsec:disc}
In this paper, we use the boundary integral equation package \texttt{chunkie} \cite{chunkie} to discretize the interface and construct the entries of the discrete approximation to the operators $\cL$ and $\cP$ given by \eqref{eq:LP}. 
Recall that the integral equation we wish to solve is $\cL \cP [\rho] = 2mu_i$.

The interface is truncated to $[a,b]$ in parameter space, i.e. we restrict $\gamma(t): [a,b]$ to its image. This truncated curve is then adaptively split into ``chunks'' where each chunk is the image of a subinterval of $[a,b]$,  and is discretized using $16$ Gauss-Legendre points. Chunks are refined until the tails of the Legendre coefficients of the speed of parameterization, and of the $x,y$ coordinates of the curve are resolved. In particular, suppose that $x^{(j)}_{n}, y^{(j)}_{n}, s^{(j)}_{n}$ are the the Legendre coefficients computed using $32$ nodes on a chunk $[a_{j},b_{j}]$, i.e.
\begin{equation}
\begin{bmatrix}
x(t) \\
y(t) \\
s(t) 
\end{bmatrix}
= \sum_{n=1}^{32} \begin{bmatrix} x^{(j)}_{n} \\ y^{(j)}_{n} \\ s^{(j)}_{n} \end{bmatrix} P_{n} \left(a_{j} + \frac{(t+1)}{2} (b_{j}-a_{j}) \right) \, ,
\end{equation}
where $s(t) = |\gamma'(t)|$, and $P_{n}(t)$ is the Legendre polynomial of degree $n$ on $[-1,1]$. Then the chunk $[a_{j},b_{j}]$ is resolved if
\begin{equation}
\max{\left(\sqrt{\frac{\sum_{n=17}^{32} |x_{n}|^2}{16}}, \sqrt{\frac{\sum_{n=17}^{32} |y_{n}|^2}{16}}, \sqrt{\frac{\sum_{n=17}^{32} |s_{n}|^2}{16}} \right)} \leq \varepsilon \, ,
\end{equation}
for a specified tolerance $\varepsilon$. If the chunk is not resolved, then it is split into two chunks of equal length in parameter space $[a_{j}, (a_{j}+b_{j})/2]$, and $[(a_{j}+b_{j})/2,b_{j}]$. Once all of the chunks are resolved, they are subsequently balanced so that adjacent chunks satisfy a $2:1$ length restriction: if $\gamma_{j}$ and $\gamma_{\ell}$ are adjacent to each other then they satisfy $|\gamma_{j}|/|\gamma_{\ell}|  \in [0.5,2]$. At the end of the adaptive procedure, the restriction of $\gamma$ to $[a,b]$ is represented via a collection of $n_{c}$ chunks, $[a,b] = \cup_{j=1}^{n_{c}} [a_{j},a_{j+1}]$ with the understanding that $a_{1} = a$, and $a_{n_{c}+1} = b$.

\begin{remark}
In the proofs, we have assumed $\gamma' \equiv 1$ for convenience. 
The proofs can be suitably modified as long as $||\gamma '| - 1| \le C e^{-\alpha |t|}$ for some positive constants $C$ and $\alpha$. Relaxing this restriction provides greater flexibility for parameterizing complicated curves.
\end{remark}

Recall that the solution $\rho \in L^{2}_{\alpha}(\mathbb{R})$, and hence there exists an $M$ such that 
\begin{equation}
\sqrt{\frac{\int_{[-\infty, -M] \cup [M,\infty]} |\rho|^2 \,  ds}{\int_{-\infty}^{\infty} |\rho|^2 \,  ds}} \leq \varepsilon \, .
\end{equation}
This estimate justifies the existence of a truncation $[a,b]$ such that the solution can be accurately represented via its restriction to some bounded interval of $\mathbb{R}$.
In all the examples, the interval $[a,b]$ is taken to be the smallest interval satisfying the following two criteria: a) that the boundary is nearly flat outside of $[a,b]$, i.e. there exists a constant vector $c \in \mathbb{R}^{2}$, such that 
$|\gamma'(t) - c| \leq \varepsilon$ for all $c \in \mathbb{R}^{2} \setminus [a,b]$, and b) the boundary data $u_{i}$ is numerically supported on $[a,b]$ to precision $\varepsilon$, i.e. $|u_{i}|_{L^{2}(\mathbb{R} \setminus [a,b])}/|u_{i}|_{L^2(\mathbb{R})} \leq \varepsilon$.

Given these restrictions, the composition $\cL \cP$ is discretized as an intergral operator on $L^{2}$ functions defined on the interval $[a,b]$. Even though $\rho$ will be numerically supported on $[a,b]$, the operator $\cP$ maps compactly supported functions to an oscillatory function that is $O(1)$ on the whole real line. On the other hand, the kernel in the integral operator $\cL$ decays exponentially as $\exp{(-\omega d)}$, where $d$ denotes the Euclidean distance between points on the interface. 
Thus, in order to compute the solution $\rho$ accurately, one needs to discretize the integral operator $\cP$ from $L^2$ functions on $[a,b]$ to $L^2$ functions on $[a',b']$, and the integral operator $\cL$ from $L^2$ functions on $[a',b']$ to $L^2$ functions $[a,b]$, where $a' = a - \log{(1/\varepsilon)}/\omega$, and $b' = b+ \log{(1/\varepsilon)}/\omega$.
The necessity and sufficiency of this choice of the buffer region is illustrated through the results in Section~\ref{subsec:verification}.  

We now turn our attention to the discretization of the integral operators $\cL$ and $\cP$. Suppose that there are an additional $n_{\textrm{buffer}}$ points introduced in each of the buffer regions $[a',a]$ and $[b,b']$.  Suppose that $t_{j} \in [a',b']$, $j=1,2\ldots n_{\textrm{over}} = 16 n_{c} + 2n_{\textrm{buffer}} $ are the discretized values of $t$ in parameter space.
Let $n_{0} = n_{\textrm{buffer}} + 16 n_{c} $.  Suppose that the points are ordered in increasing values of $t$. 
Then the points $j=1,2\ldots \nbuffer$ correspond to the left buffer region $[a',a]$, the points $j=\nbuffer+1,\ldots n_{0}$ correspond to the interval $[a,b]$, and the points $j=n_{0}+1,\ldots n_{\textrm{over}}$ correspond to the right buffer region $[b,b']$.
Note that the density $\rho$ is discretized through its values at $\rho(t_{j})$, for $j=n_{\textrm{buffer}}+1,\ldots n_{0}$, and hence the discretized linear system corresponding to $\cL \cP$ will be of size $n_{0} - \nbuffer= 16 n_{c}$.

Both the kernels $\cL$ and $\cP$ have non-smooth and at most weakly singular kernels for small distances and require specialized quadrature rules for integrating them. We use the generalized Gaussian quadrature rules of~\cite{bremer2010universal,bremer2010nonlinear} for the accurate computation of these integrals. In particular, the quadrature rule is a target dependent locally-corrected quadrature rule which accurately integrates the specific non-smooth behavior of the kernel in the vicinity of the origin. 
For every point $t_{j} \in [a_{i}, a_{i+1}]$, there exist weights $w_{j,\ell}$ for all $t_{\ell} \in [a_{i-1}, a_{i+1}]$ such that the discretized versions of $\cP$ and $\cL$, denoted by $\mtx{P}$ and $\mtx{L}$ respectively, are given by

\begin{equation}
\begin{aligned}
\mtx{P}[\rho](t_{j}) &= \frac{m^2}{E}\sum_{\ell=\nbuffer+1}^{n_{0}}  (\delta_{j,\ell} + e^{iE|t_{j}-t_{\ell}|}  w_{\ell}) \rho(t_{\ell}) \, \, \\ 
&\qquad + 
\sum_{\substack{\ell= \nbuffer+1 \\ t_{\ell} \in [a_{i-1},a_{i+1}]}}^{n_{0}}  \rho(t_{\ell}) w_{j,\ell} \, , \quad j=1,2,\ldots n_{\textrm{over}} \\
\mtx{L}[\mu](t_{j}) &= \mu(t_{j}) - 2m\sum_{\substack{\ell=1 \\ \ell \neq j}}^{n_{\textrm{over}}} G_{\omega}(\gamma(t_{j}),\gamma(t_{\ell}))   \mu(t_{\ell}) w_{\ell}\,\, \\
&\qquad+ \sum_{\substack{\ell=1 \\ t_{\ell} \in [a_{i-1},a_{i+1}]}}^{\nover}  \mu(t_{\ell}) w_{j,\ell} \, , \quad j=\nbuffer+1,\nbuffer+2,\ldots n_{0}  \, .
\end{aligned}
\end{equation}

Here $\delta_{j,\ell}$ is the Kronecker delta, i.e. $\delta_{j,j} = 1$ and $\delta_{j,\ell} = 0$ otherwise, 
and $w_{j}$, $j=1,2\ldots \nover$, denote the quadrature weights for integrating smooth functions on the interface.  The factorized linear systems $\mtx{L}$ and $\mtx{P}$ are illustrated in Figure~\ref{fig:linalg}.

\begin{remark}
    Recall that $\cL$ may have a continuous spectrum passing through zero while $\cP$ is not bounded; see \eqref{eq:symbolKP}. Thus it is expected (and observed numerically) that in general, $\mtx{L}$ and $\mtx{P}$ both have very large condition numbers. 
    This could in principle lead
    to catastrophic cancellation arising from the numerical implementation of the operator product $\mtx{L} \mtx{P}$. We postpone a thorough treatment of this potential issue to future study. Practically speaking, the accuracy of our numerical method does 
    not seem to suffer from the poor condition numbers of $\mtx{L}$ and $\mtx{P}$. Indeed, as illustrated by Section \ref{subsec:verification}, we observe low errors for a variety of interfaces and wide range of parameters.
\end{remark}

\begin{figure}[h!]
    \centering
    \includegraphics [scale=.22]{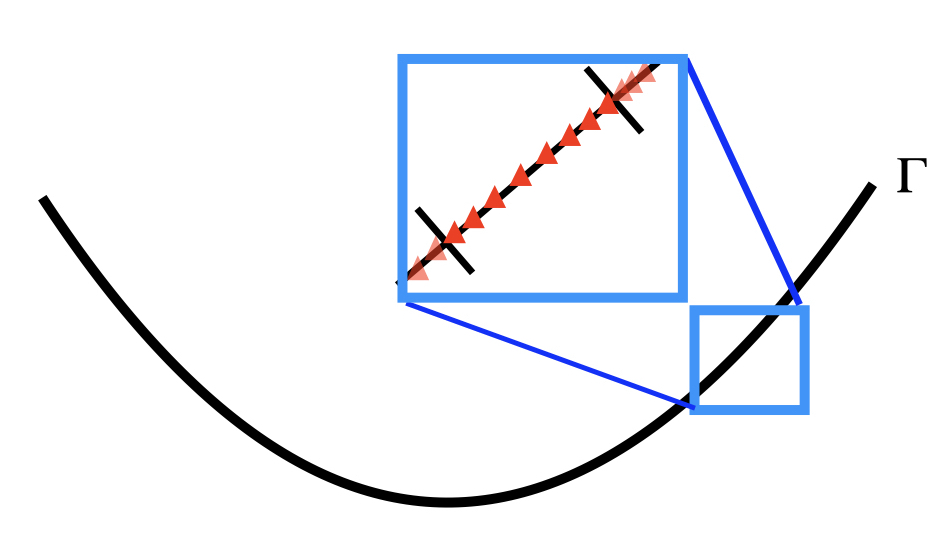}
    \caption{Schematic of the discretization approach used by \texttt{chunkie}. In the inlay, bounds between `chunks' are shown with vertical lines, and discretization nodes are denoted by red triangles. For clarity, the `panel' shown is $8^{\rm th}$ rather than $16^{\rm th}$.}
    \label{fig:chunkie}
\end{figure}

\begin{figure}[h!]
    \centering
    \includegraphics [scale=.22]{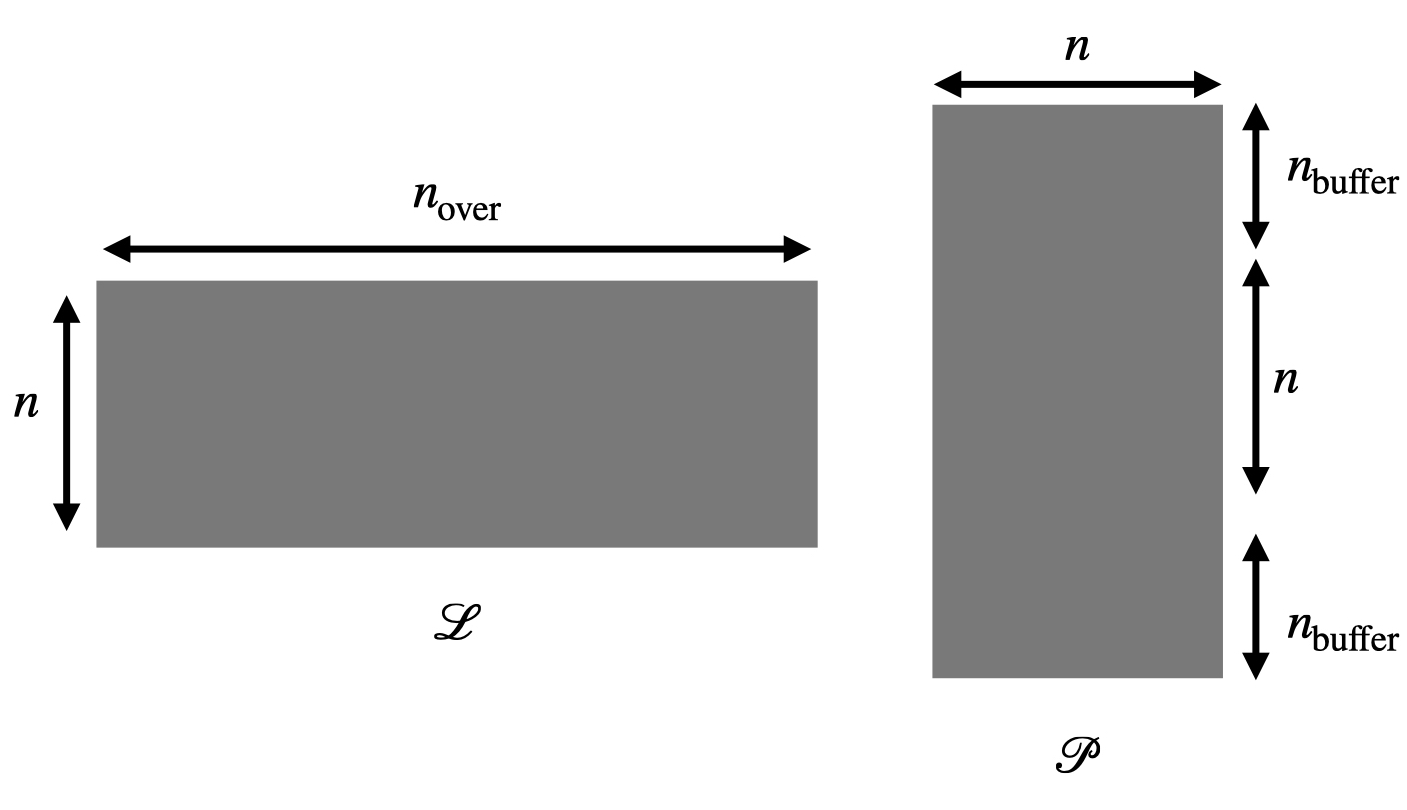}
    \caption{The factorized linear system, after discretization.}
    \label{fig:linalg}
\end{figure}
\subsection{Acceleration of the numerical method}\label{subsec:accel}
In this section, we discuss a fast algorithm for the evaluation of the matrix vector product $\mtx{L} \mtx{P}[\rho]$. The kernel of $\mtx{P}$ is the Green's function of a one-dimensional translation-invariant elliptic ordinary differential equation and hence can be accelerated using a sweeping algorithm. In order to apply $\mtx{P}$ rapidly, we require a fast algorithm for the evaluation of 
\begin{equation}
\mtx{P}[\rho](t_{j})= \frac{m^2}{E}\sum_{\ell=\nbuffer+1}^{n_{0}}  e^{iE|t_{j}-t_{\ell}|}  \rho(t_{\ell}) w_{\ell} \, ,  \quad j=1,2,\ldots \nover \, ,
\end{equation}
since the rest of the interaction is sparse and can be computed in $O(\nover)$ CPU time. Sums of this kind are frequently encountered in the application of integral equation based methods to elliptic PDEs, and can be efficiently computed using standard \emph{sweeping/sum-of-exponential} type algorithms (see \cite{GIMBUTAS2020815,jiang2014,jiang2021}). For completeness, we briefly summarize the approach as it applies in this context.
The main idea of the sweeping algorithm is to split the solution into two parts for any point $t_{j}$, $t\leq t_{j}$, and $t>t_{j}$, where both pieces can be updated in $O(1)$ operations as we move from $t_{j} \to t_{j+1}$ or $t_{j} \to t_{j-1}$. 
Let $v^{\uparrow}$ and $v^{\downarrow}$ denote the accumulation of the rightward moving solution (corresponding to $t\leq t_{j}$) and the leftward moving solution(corresponding to $t>t_{j}$) respectively.

In particular, we split the solution as follows,
\begin{equation}
\begin{aligned}
P[\rho](t_{j}) &= \frac{m^2}{E}\sum_{\substack{\ell=\nbuffer+1\\ \ell \leq j}}^{n_{0}} e^{iE (t_{j} - t_{\ell})} \rho(t_{\ell}) w_{\ell} + \frac{m^2}{E}\sum_{\substack{\ell=\nbuffer+1\\ \ell > j}}^{n_{0}} e^{iE (t_{\ell} - t_{j})} \rho(t_{\ell}) w_{\ell} \, \\
&= v^{\uparrow}_{j} + v^{\downarrow}_{j} \, .
\end{aligned}
\end{equation}
A simple calculation shows that $v^{\uparrow}$ and $v^{\downarrow}$ satisfy the following recurrence relations
\begin{equation}
\begin{aligned}
v^{\uparrow}_{j} &= v^{\uparrow}_{j-1} e^{iE(t_{j} - t_{j-1})}  + \frac{m^2}{E} \rho(t_{j}) w_{j} I_{j\in[\nbuffer+1,n_{0}]} \, , \\
v^{\downarrow}_{j} &= e^{iE(t_{j+1}-t_{j})} \left(v^{\downarrow}_{j+1}  + \frac{m^2}{E} \rho(t_{j+1}) w_{j+1} I_{j+1\in[\nbuffer+1,n_{0}]} \right)  \, ,
\end{aligned}
\end{equation}
where $I_{j\in A}$ is the indicator function of the set $A$, which is equal to $1$ if $j \in A$, and $0$ otherwise. 
Thus, $v^{\uparrow}$ satisfies an upward recurrence in $j$, while $v^{\downarrow}$ satisfies a downward recurrence, and both $v^{\uparrow}$ and $v^{\downarrow}$ can be computed for all $j$ in $O(\nover)$ work. The recurrences are initialized with $v^{\uparrow}_{1} = 0$, and $v^{\downarrow}_{\nover} = 0$.

On the other hand, the kernel of $\mtx{L}$ is the Green's function of the two dimensional Helmholtz equation with imaginary wave number $\omega$ and the bulk of the computation is given by
\begin{equation}
\label{eq:lsum}
 - 2m\sum_{\substack{\ell=1 \\ \ell \neq j}}^{n_{\textrm{over}}} G_{\omega}(\gamma(t_{j}),\gamma(t_{\ell}))   \mu(t_{\ell}) w_{\ell} \, , \quad j=\nbuffer+1, \, \ldots n_{0} \, .
\end{equation}
The above sum can be computed at all $t_{j}$, $j=\nbuffer+1, \ldots n_{0}$ in $O(\nover)$ CPU time using the standard fast multipole method~\cite{rokhlin1990rapid,greengard1987fast}. We use the fast multipole implementation in \texttt{fmm2d}~\cite{fmm2d} for evaluating the sum in~\eqref{eq:lsum}. The rest of the computation in $\mtx{L}$ is sparse whose number of nonzero elements is also $O(\nover)$.

Combining both of these fast algorithms, the matrix vector product $\mtx{L} \mtx{P}[\rho]$ can be applied in $O(\nover)$ CPU time. Thus, the solution $\rho$ can be obtained in $O(\nover \niter)$ CPU time using iterative methods like the generalized minimum residual (GMRES) method, where $\niter$ is the number of GMRES iterations required for the relative residual to drop below a prescribed tolerance. In practice, the integral equation tends to be well-conditioned which results in $\niter = O(1)$, and thus the computational complexity of obtaining the solution $\rho$ is $O(\nover)$.

\begin{remark}
The key components of the fast numerical solver include an adaptive piece-wise high-order discretization of the boundary using chunkie~\cite{chunkie}, a sweeping algorithm~\cite{GIMBUTAS2020815,jiang2014,jiang2021} for applying $\mtx{P}$ a discretization of the preconditioner $\cP$, and a Yukawa fast multipole~\cite{fmm2d} for accelerating the computation of $\mtx{L}$, a discretization of the integral operator $\cL$. While all the components of the numerical solver are existing discretization and acceleration tools, to the best of our knowledge, their use in the efficient and high-order accurate numerical solution of the boundary value problem~\ref{eq:pde} is the first of its kind.
\end{remark}

\section{Numerical illustrations and examples}\label{sec:illustrations}
In this section, we 
provide several examples of the numerical method described in Section \ref{sec:num}. We demonstrate accuracy or self-convergence of the algorithm and test its speed for a variety of interfaces; see subsection \ref{subsec:verification}. In subsection \ref{subsec:examples}, we plot corresponding solutions and compute reflection coefficients for a scattering theory. 
We refer to Figure \ref{fig:interfaces} for an illustration of the interfaces used in our examples.

\begin{remark}
Some of the examples presented below correspond to the two-mass problem detailed in Section \ref{subsec:extension}. The numerical method from Section \ref{sec:num} (with trivial modifications) still applies to this more general case.
\end{remark}

\begin{figure}
    \centering
    \includegraphics[scale=.15]{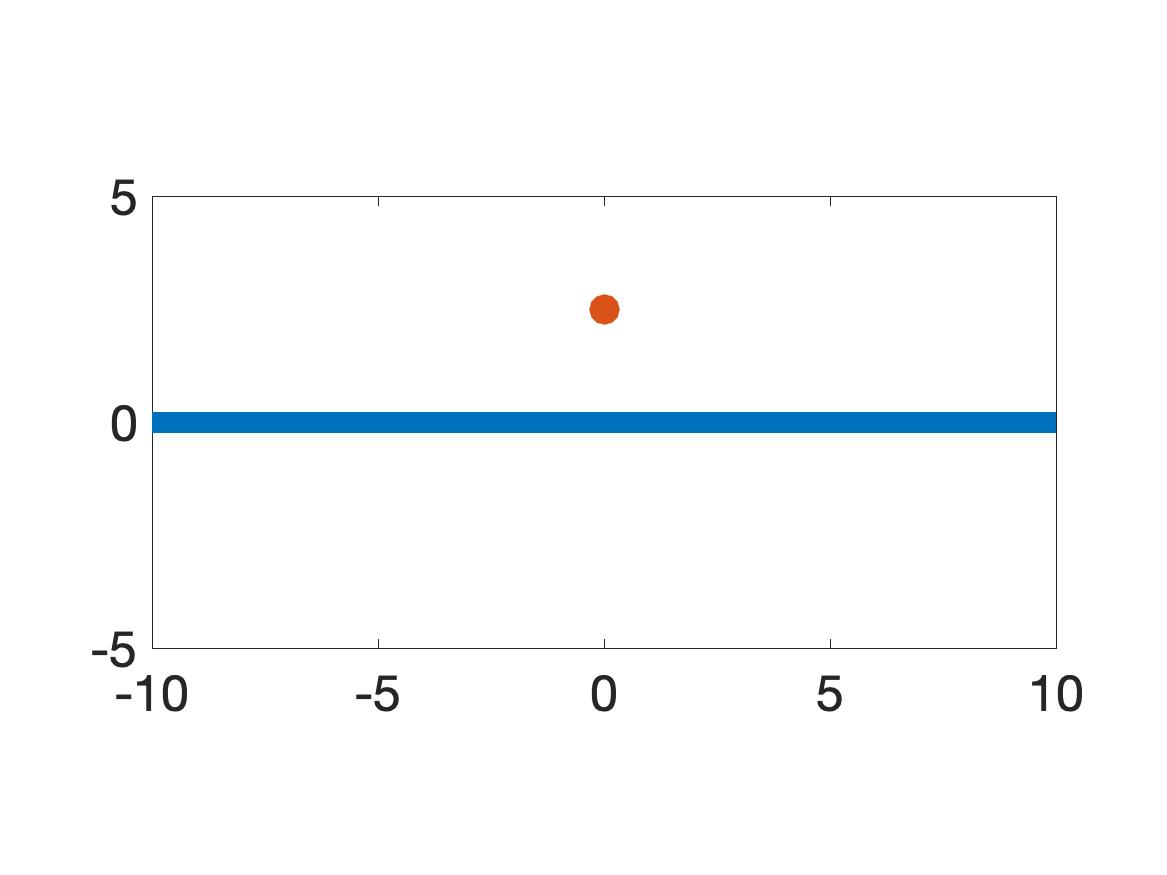}
    \includegraphics[scale=.15]{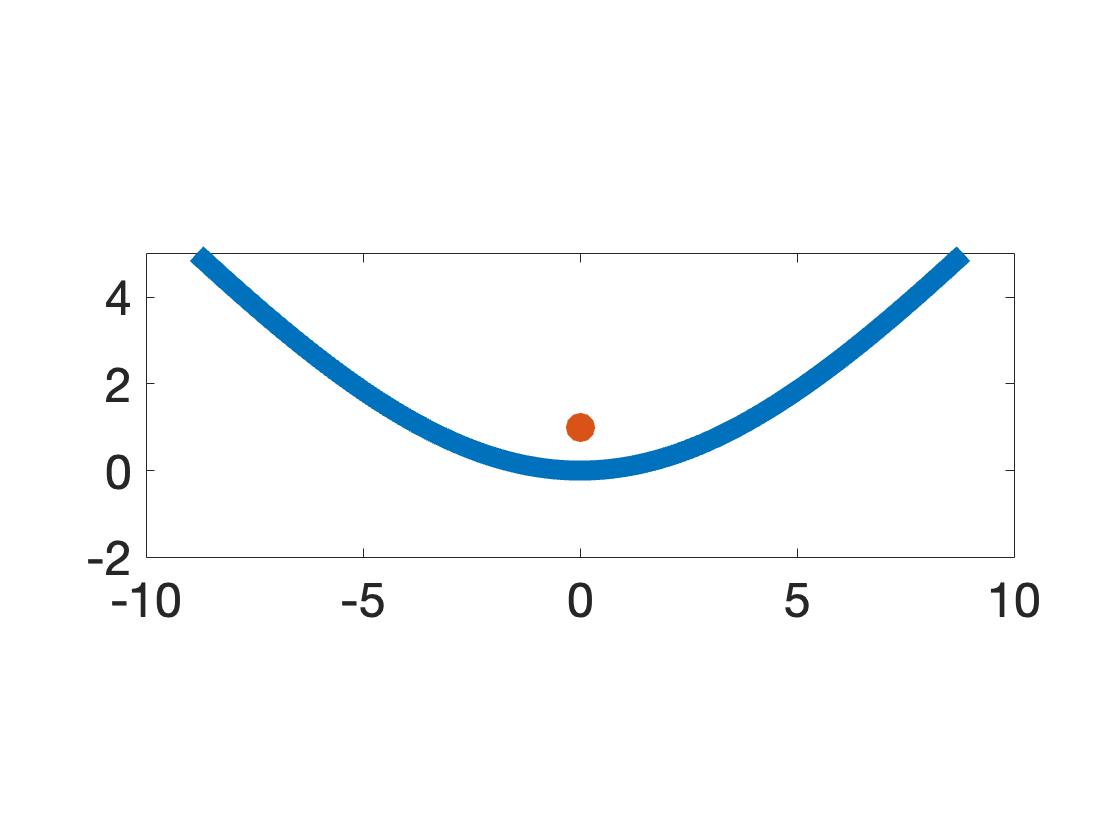}\\
    \includegraphics[scale=.15]{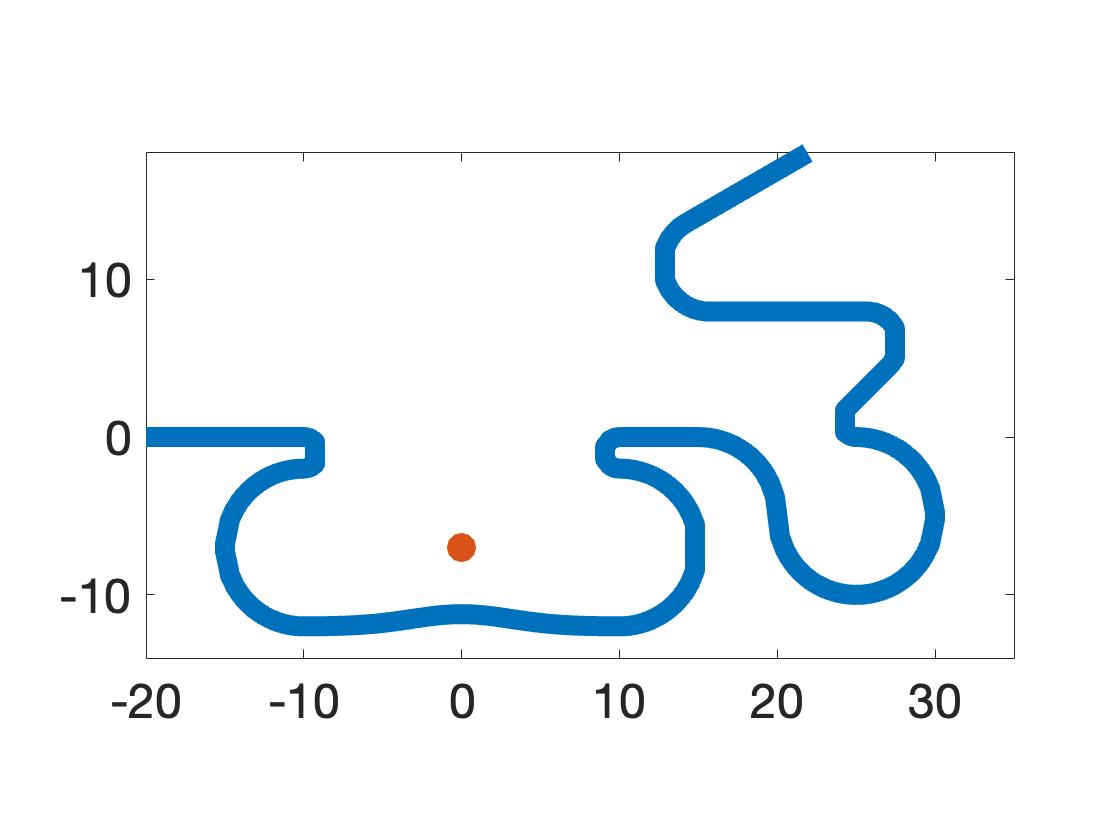}
    \includegraphics[scale=.15]{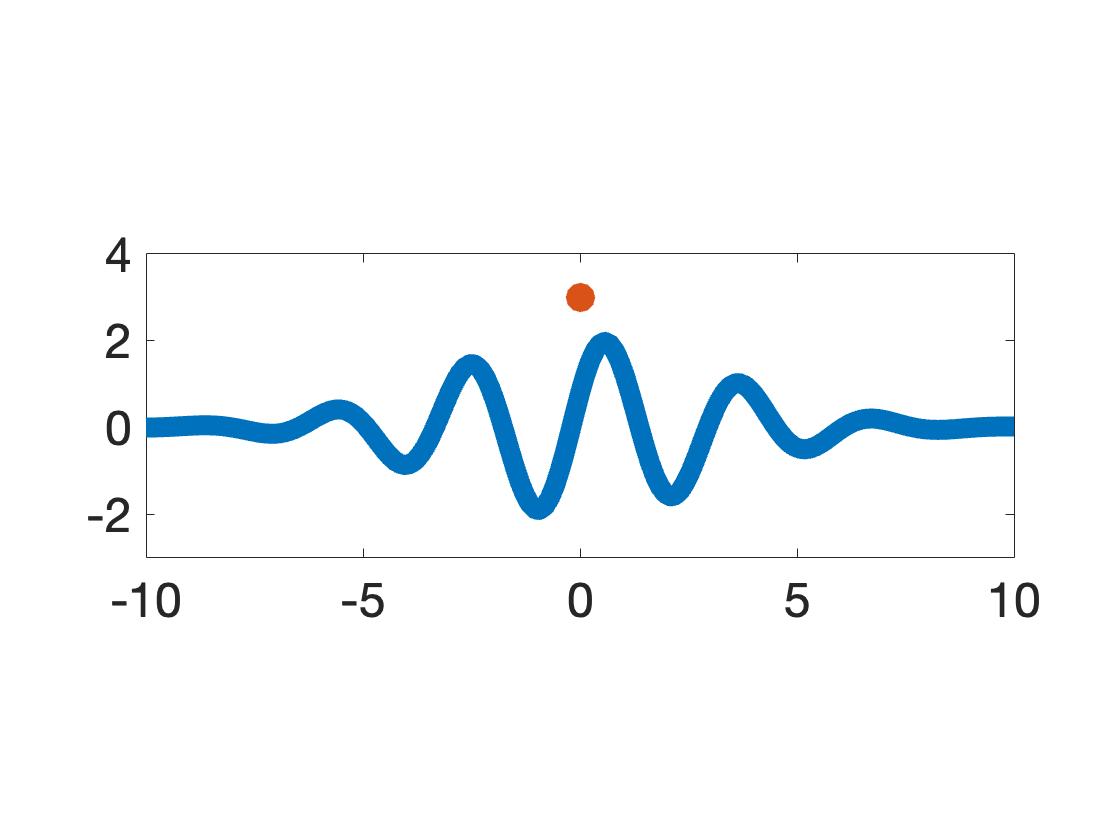}
    \caption{The interfaces $\Gamma_0$ (top left), $\Gamma_1$ (top right), $\Gamma_2$ (bottom left) and $\Gamma_3$ (bottom right), with respective sources at $(0,2.5)$, $(0,1)$, $(0,-7)$ and $(0,3)$ as indicated by the red dot. Outside the plotted region, the interfaces extend linearly to infinity.}
    \label{fig:interfaces}
\end{figure}

\subsection{Illustration of the numerical method}\label{subsec:verification}
This section presents the accuracy and speed of our numerical method for various interfaces.
We begin with the flat-interface ($\Gamma_0$) case, where there is an analytic expression for the Green's function. 
If $u$ denotes our computed solution,
we define the relative error of $u$ at the point $x_T$ by $|u (x_T) - u_* (x_T)|/|u_* (x_T)|$, where $u_*$ is the true solution.
Figure \ref{fig:flat} contains
a plot of this relative error (computed at four arbitrary points) as a function of $n_c$, as well as illustrations of the computed Green's function and densities.
Observe that our solution is highly accurate even for small values of $n_c$.

\begin{figure}
    \centering
    \includegraphics [width=0.45\linewidth]{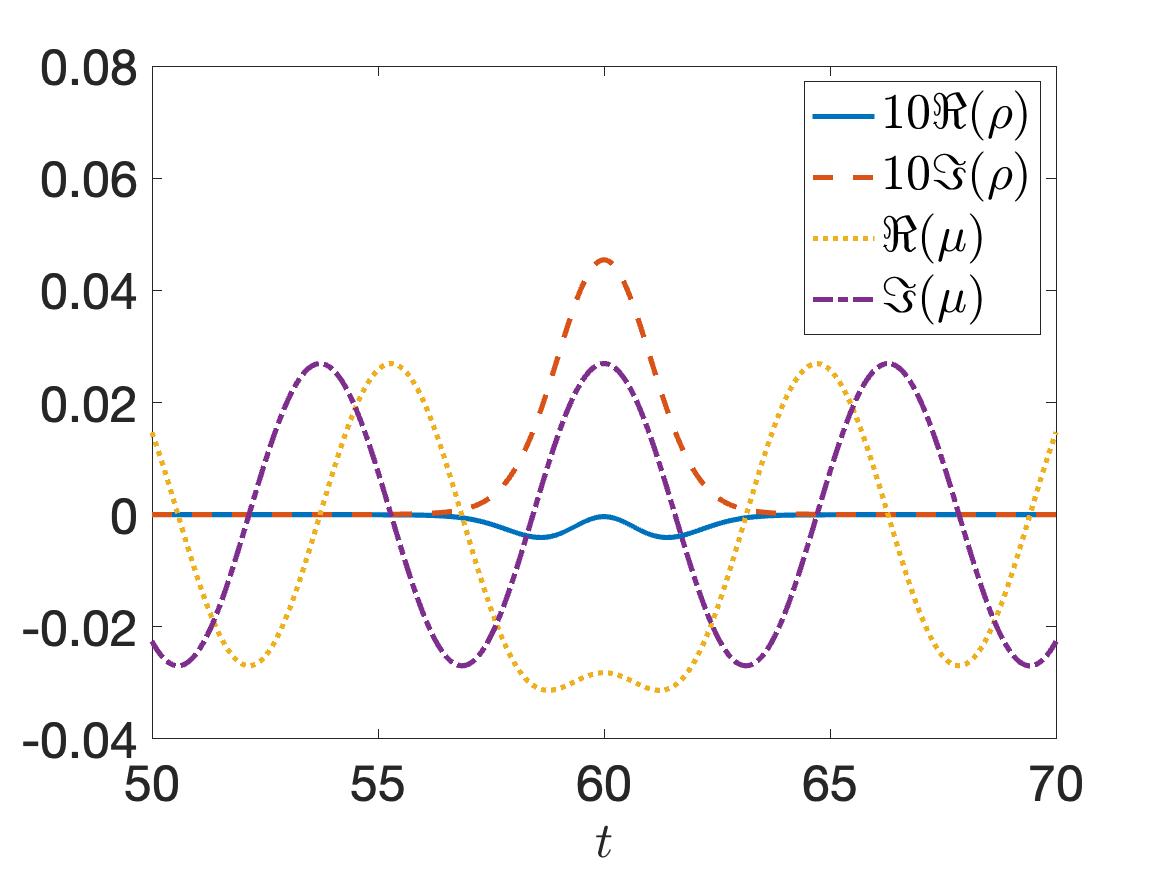}
    \includegraphics [width=0.45\linewidth]{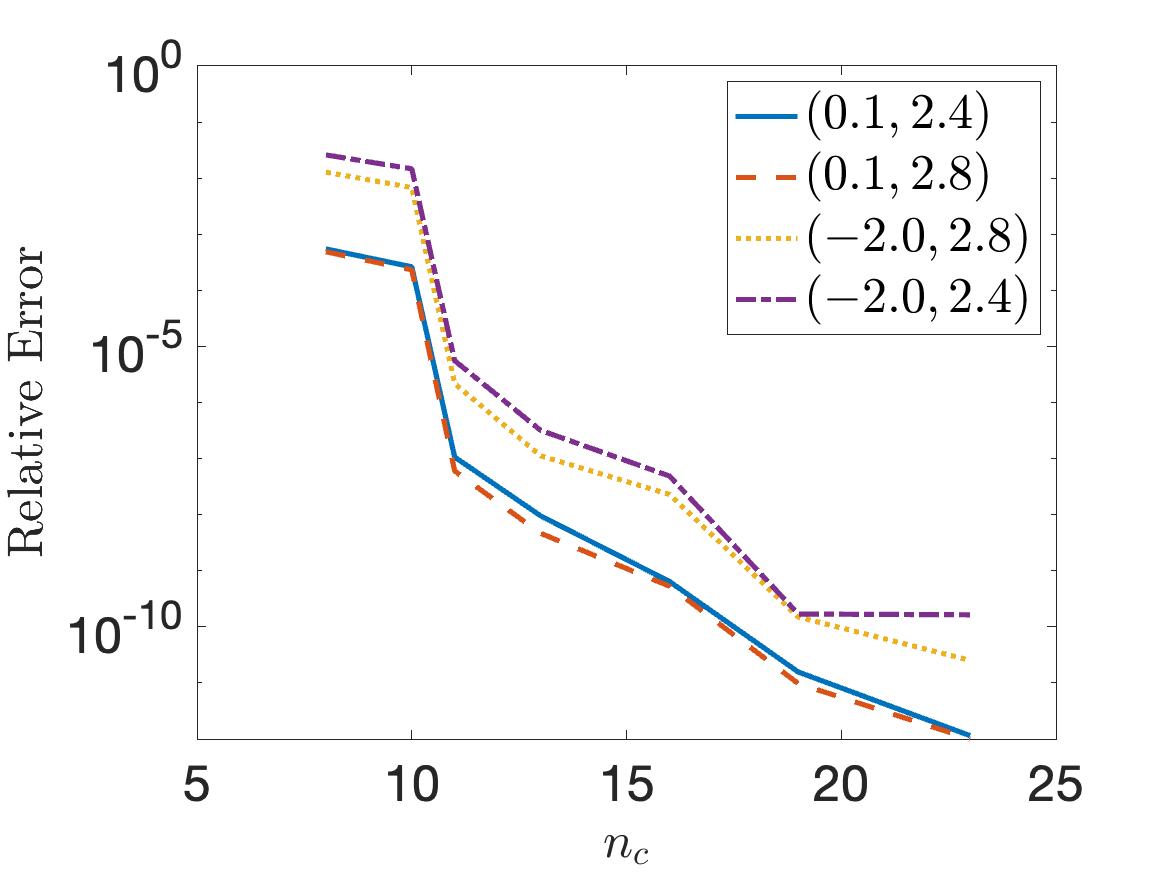}\\
    \includegraphics [width=0.45\linewidth]{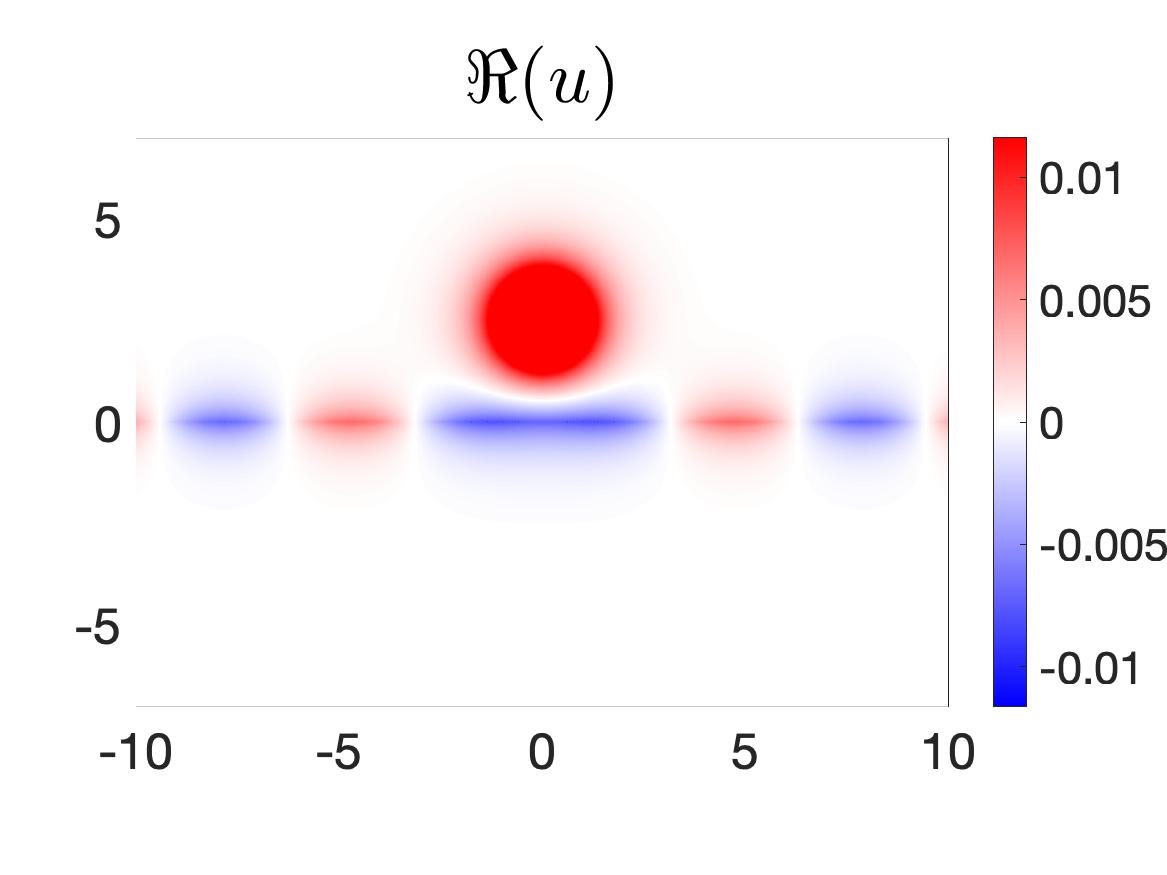}
    \includegraphics [width=0.45\linewidth]{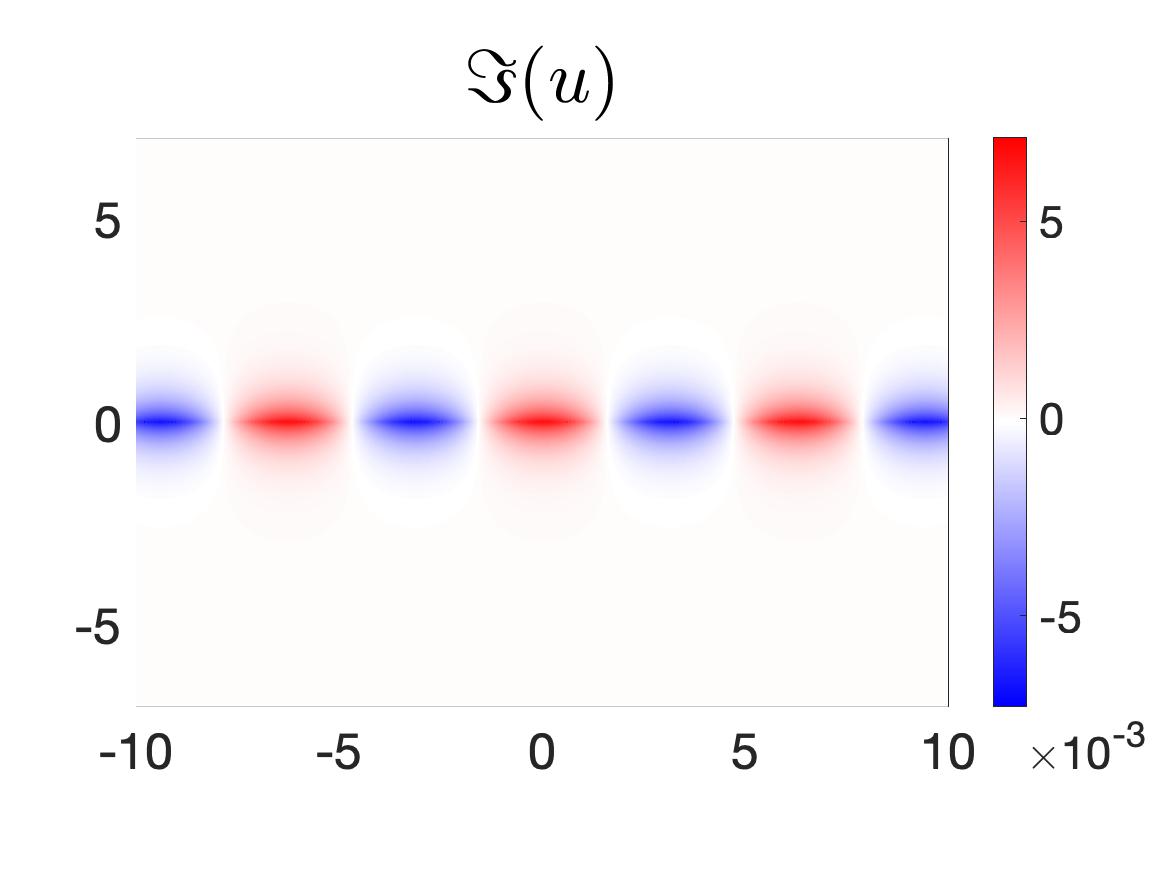}\\[-0.25cm]
    \includegraphics [width=0.45\linewidth]{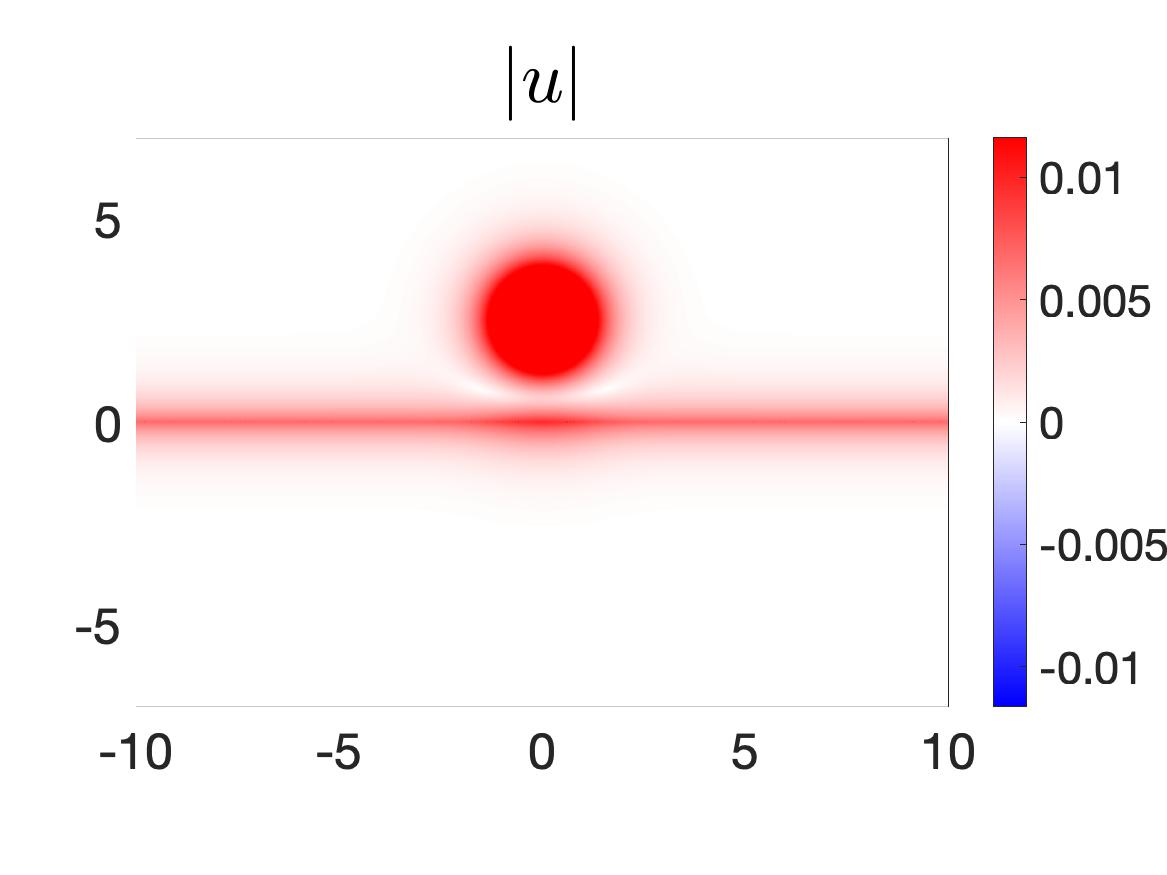}
    \includegraphics [width=0.45\linewidth]{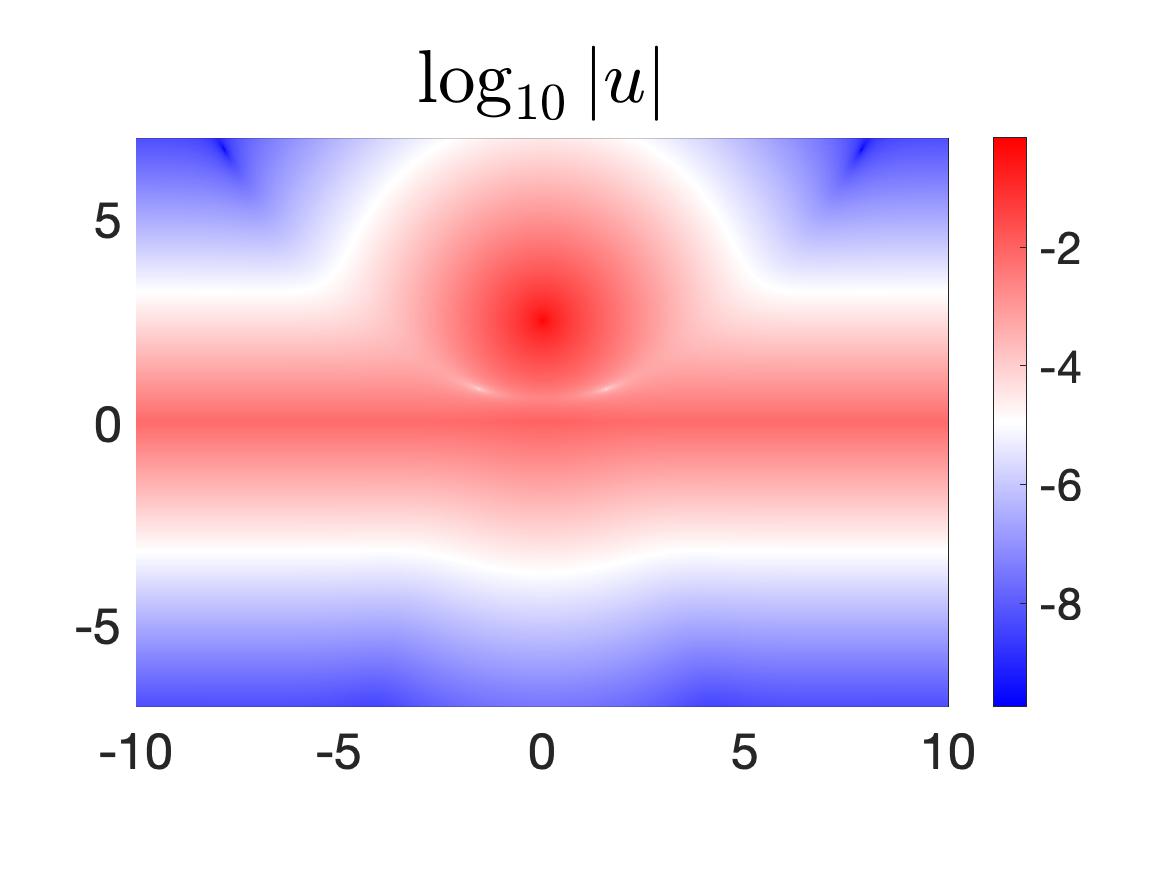}
    \caption{
    Densities $\mu$ and $\rho$ (top left panel), relative error of the computed solution at 
    the indicated points as a function of $n_c$ (top right panel), and Green's function $u$ (bottom two rows) corresponding to the flat interface $\Gamma_0$ with $m=2$ and $E=1$. The top left panel zooms in on the region $[10,10]\times \{0\} \subset \Gamma_0$, with $t=60$ corresponding to the point $(0,0) \in \Gamma_0$.}
    \label{fig:flat}
\end{figure}

For arbitrary non-flat interfaces 
(such as $\Gamma_1, \Gamma_2, \Gamma_3$), analytic solutions are not known and hence we cannot compute the exact relative error. 
Instead, 
we perform a self-convergence test, 
which involves approximating the true solution by the numerical solution at some large value $N^*$ of $n_c$. 
If we now let $u_N$ denote the computed solution with $n_c=N$, our approximate relative error at the target $x_T$ is given by
$|u_{N} (x_T) - u_{N^*} (x_T)|/|u_{N^*} (x_T)|$.
For a corresponding plot with $N^*=512$, see Figure \ref{fig:error} (top left panel).
We again observe fast convergence in $n_c$, though understandably not as fast as the flat-interface case.

Another parameter of interest is the truncation length $n_{\text{buffer}}$ (introduced in Section \ref{subsec:disc}).
Its default value (unless otherwise specified) is
$$\nb = 2\Bigg \lceil \frac{\nc \log (10^{16})}{m_{0} \Delta_t}\Bigg \rceil=:2\lceil M_b \rceil,$$
where $m_{0} := \min (m_1, m_2)$ and $\Delta_t$ is the arclength of $\Gamma$ over the entire discretized region (that is, the region discretized by all $\no = n+2\nb$ grid points; see Section \ref{subsec:disc}).
Here, $M_b\in \mathbb{N}$ is sufficiently large so that interactions between points separated by at least $M_b$ grid points are bounded by $10^{-16}$ in absolute value. 
Some routines in the numerical experiments presented here were set to a tolerance of $10^{-12}$, so the error of $10^{-16}$ more than ensures that any error we observe would not be due to truncation.

To test the effect of $\nb$ on the convergence of our method, we introduce $\tau > 0$ and set $\nb = 2\lceil \tau M_b \rceil$.
In the top right panel of Figure \ref{fig:error}, we plot
$|u_{n_c,\tau}(x_T) - u_{512, \tau} (x_T)|/|u_{512, \tau} (x_T)|$ as a function of $\tau$, where $u_{N,\tau}$ is our computed solution with $n_c = N$ and $\nb =2\lceil \tau M_b \rceil$. We set $n_c=128,64,256$ for $\Gamma_1, \Gamma_2, \Gamma_3$, respectively.
Given 
the small truncation error tolerance at $\tau = 1$, it makes sense that decreasing the value of $\tau$ from $1$ would not immediately increase the relative error. Still, once $\tau$ gets small enough (say, less than $0.5$), the truncation length is too small and the convergence of our method suffers. 
The relative error increases when $\tau$ increases from $1$, as we do not keep enough grid points in this case.

In summary, 
the top left panel illustrates relative error at fixed $\tau=1$ and varying $n_c$, while the top right panel illustrates relative error at fixed $n_c$ (depending on the interface) and varying $\tau$.
We refer to the bottom panel of Figure \ref{fig:error} for 
a plot of the speed of our method as a function of $n_c$.
As predicted, the computation time grows only linearly in $n_c$.
The respective slopes of the line of best fit for $\Gamma_1, \Gamma_2, \Gamma_3$ are $1.06 \times 10^{-2}, 1.03 \times 10^{-2}, 8.43 \times 10^{-3}$, with an average value of $9.78 \times 10^{-3}$. These slopes were computed using only the data for $n_c > 50$ to eliminate the effect of the nonlinear behavior of the curves for small $n_c$. 
The lines of best fit are used to extrapolate the data to $n_c > 256$ in the plot. 

\begin{figure}
    \centering
    \includegraphics[width=0.45\linewidth]{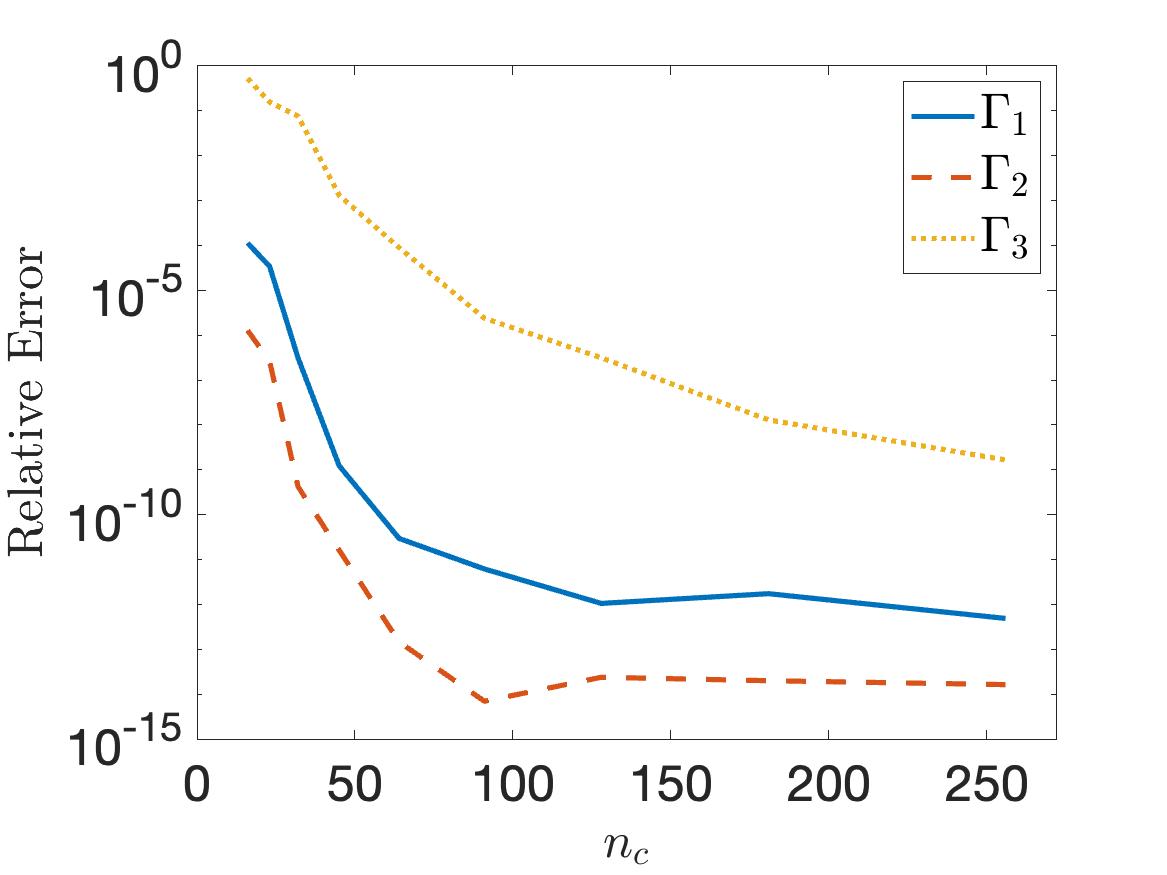}
    \includegraphics[width=0.45\linewidth]{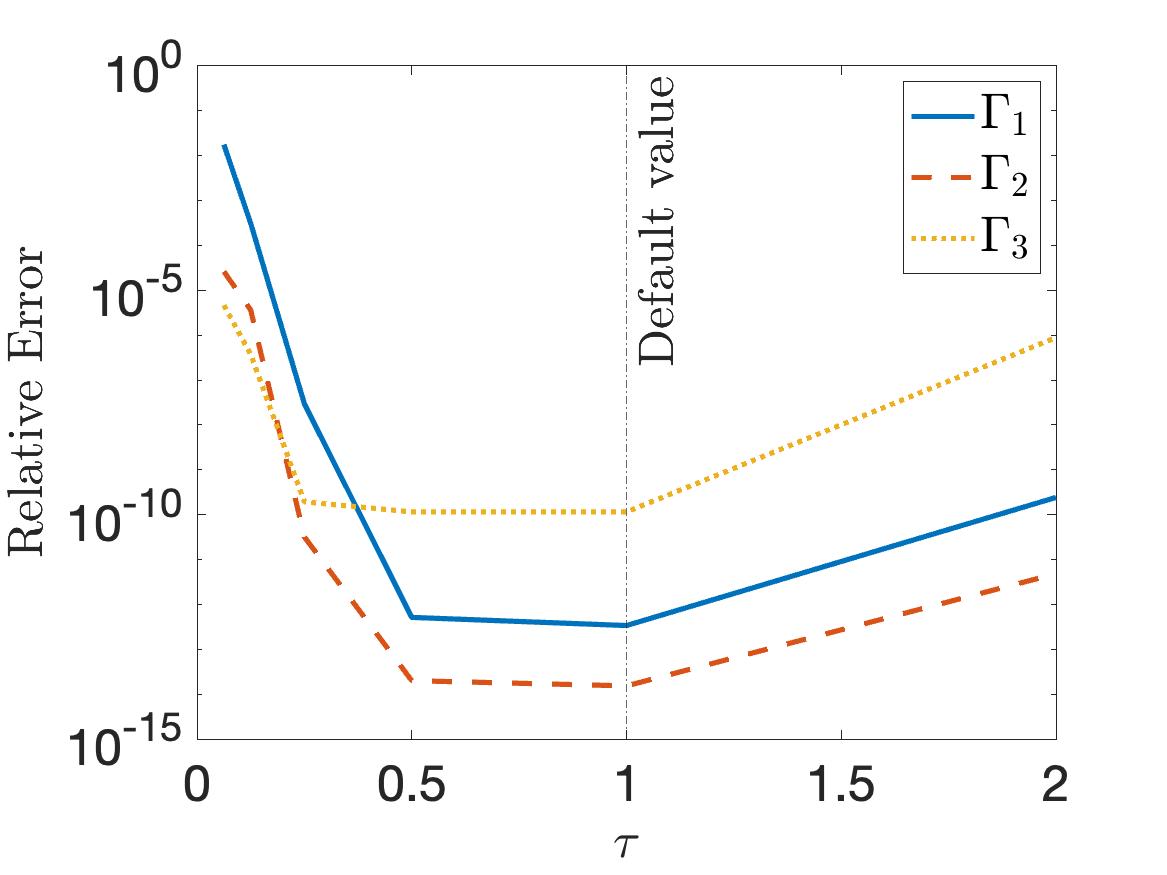}\\
    \includegraphics[width=0.45\linewidth]{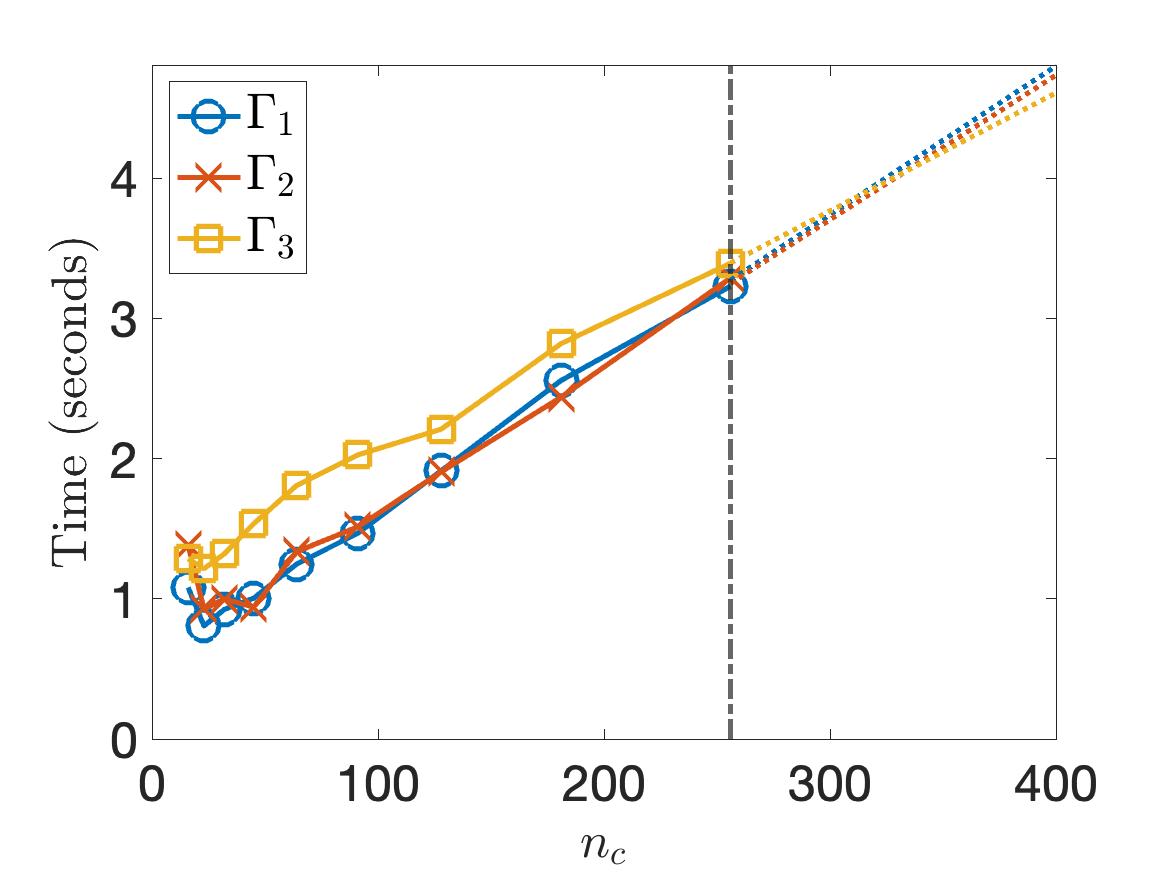}
    \caption{Self-convergence tests for varying $n_c$ (top left panel) and truncation length (top right panel). 
    The relative error is computed at $x_T = (-1,1), (7,-7), (1,-3)$ for $\Gamma_1, \Gamma_2,\Gamma_3$ respectively. For each interface, the source location is given by Figure \ref{fig:interfaces}. The respective values of $m$ and $E$ are those from Figures \ref{fig:exv} (middle panel), \ref{fig:exc} and \ref{fig:exw} below. The computational cost of obtaining $u(x_T)$ is illustrated by the bottom panel. 
    }
    \label{fig:error}
\end{figure}

\subsection{Examples of applications}\label{subsec:examples}
We now present various examples corresponding to the non-flat interfaces from Figure \ref{fig:interfaces}. The Green's function, $u$, 
and densities, $\mu$ and $\rho$, for a source (whose location is given by Figure \ref{fig:interfaces}) near a non-flat section of the interface are plotted in Figures \ref{fig:exv}--\ref{fig:exw2}.
As expected, $\omega = \sqrt{m^2-E^2}$ dictates the rate at which solutions decay away from the interface. In particular, larger values of $\omega$ result in faster decaying solutions. In the case of different masses, we observe that $u$ decays more rapidly in the domain with larger $\omega$; see Figures \ref{fig:exv} and \ref{fig:exw2}.

\begin{figure}
    \centering
    \includegraphics[width=0.75\linewidth]{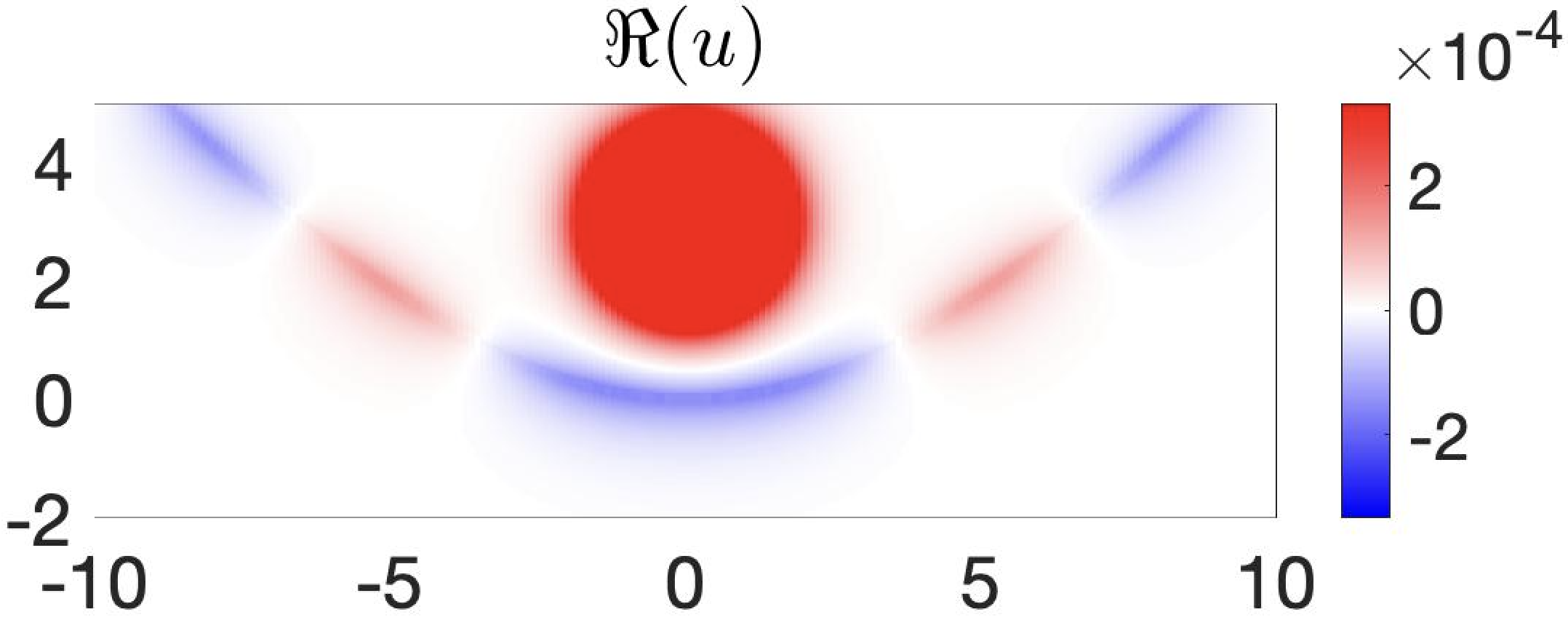}\\[0.5cm]
    \includegraphics[width=0.75\linewidth]{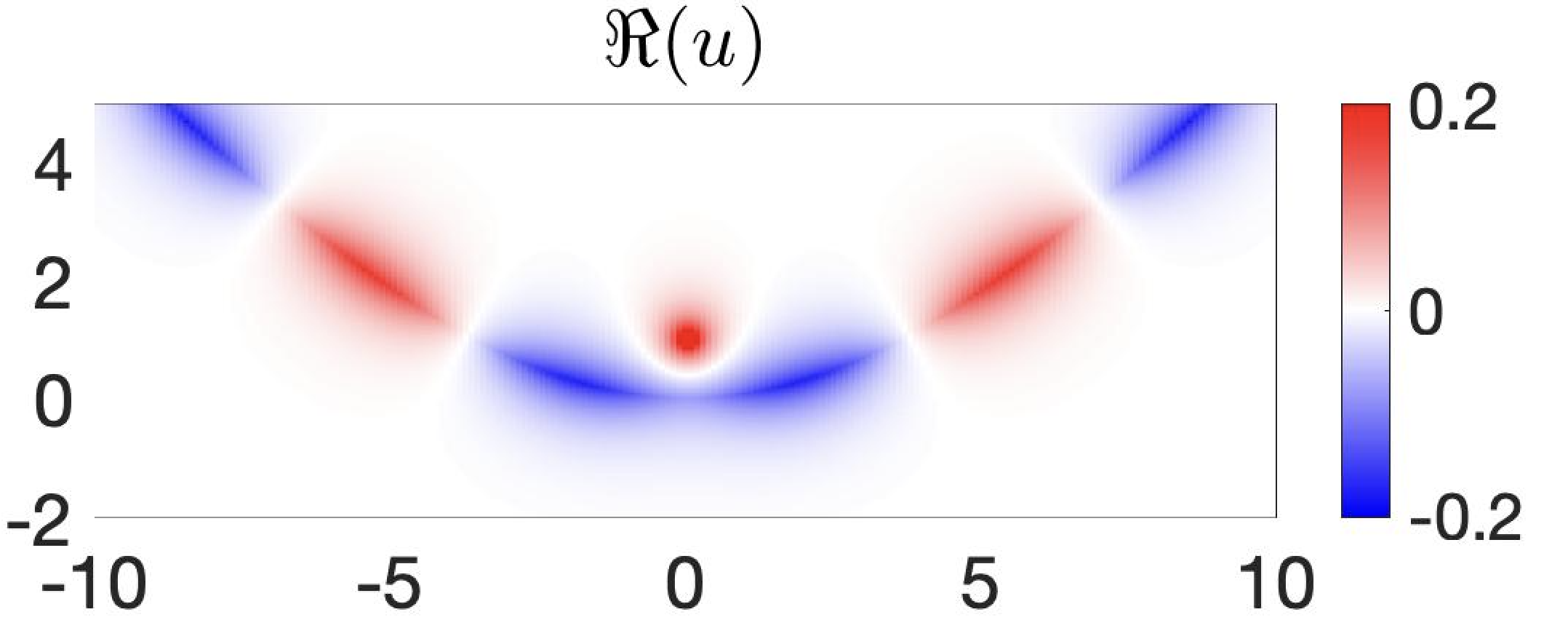}\\[0.5cm]
    \includegraphics[width=0.75\linewidth]{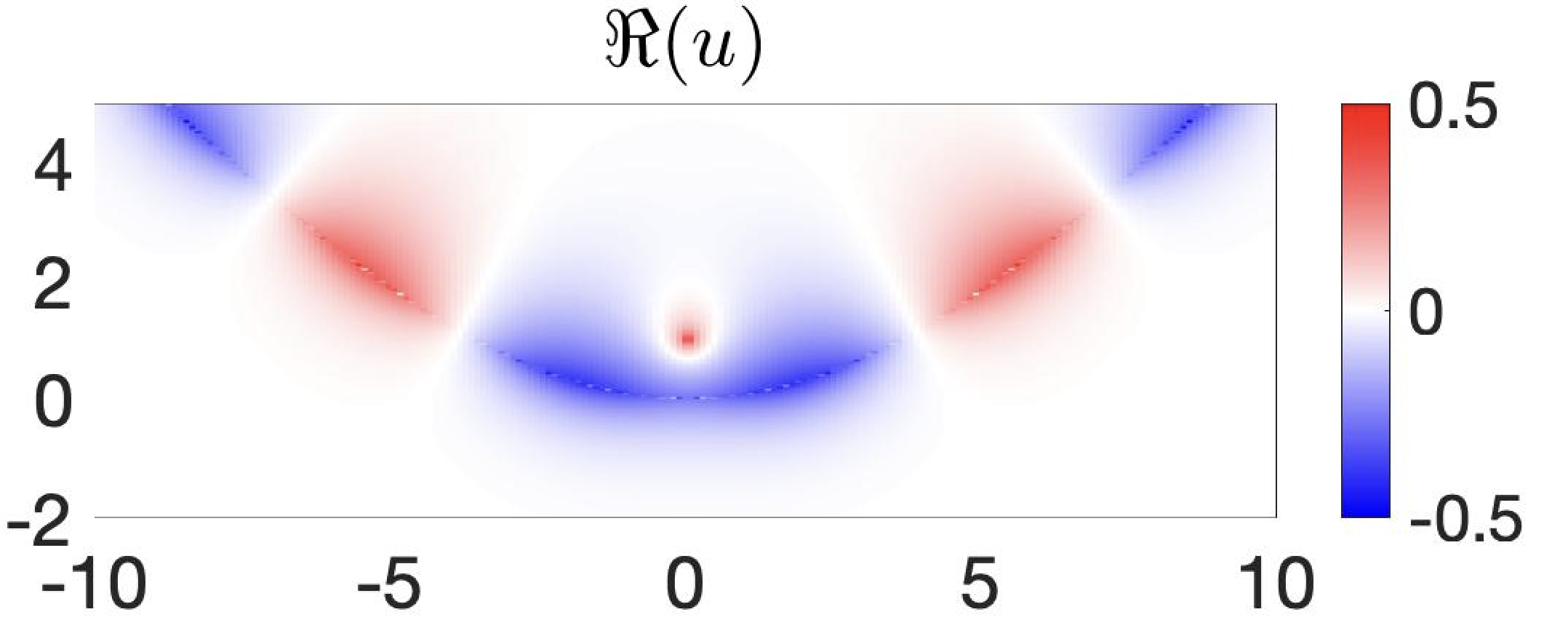}
    \caption{Green's function, $u$, for interface $\Gamma_1$ with $m_1=2$, $E=0.8$ and $m_2=3$ (top) $m_2=2$ (middle) and $m_2=1$ (bottom).}
    \label{fig:exv}
\end{figure}

\begin{figure}
    \centering
    \includegraphics[width=0.75\linewidth]{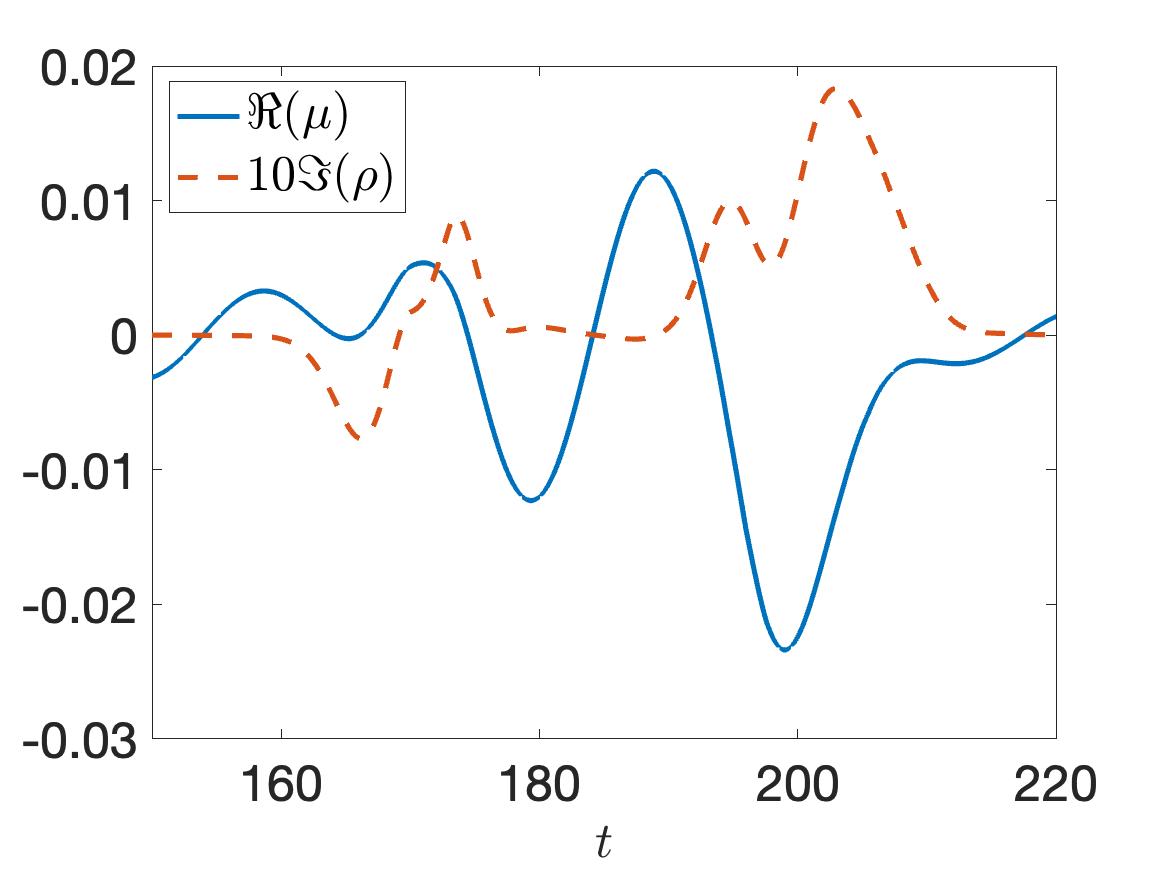}\\
    \includegraphics[width=0.45\linewidth]{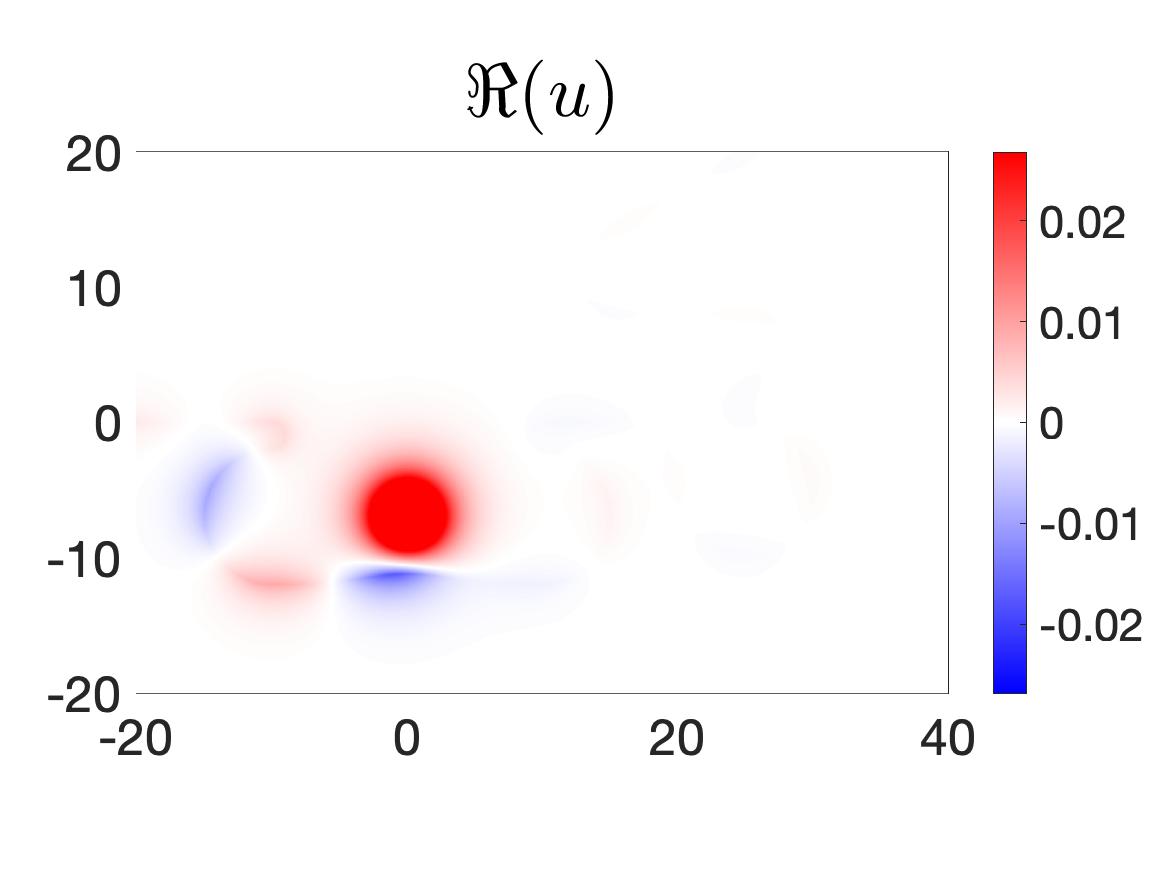}
    \includegraphics[width=0.45\linewidth]{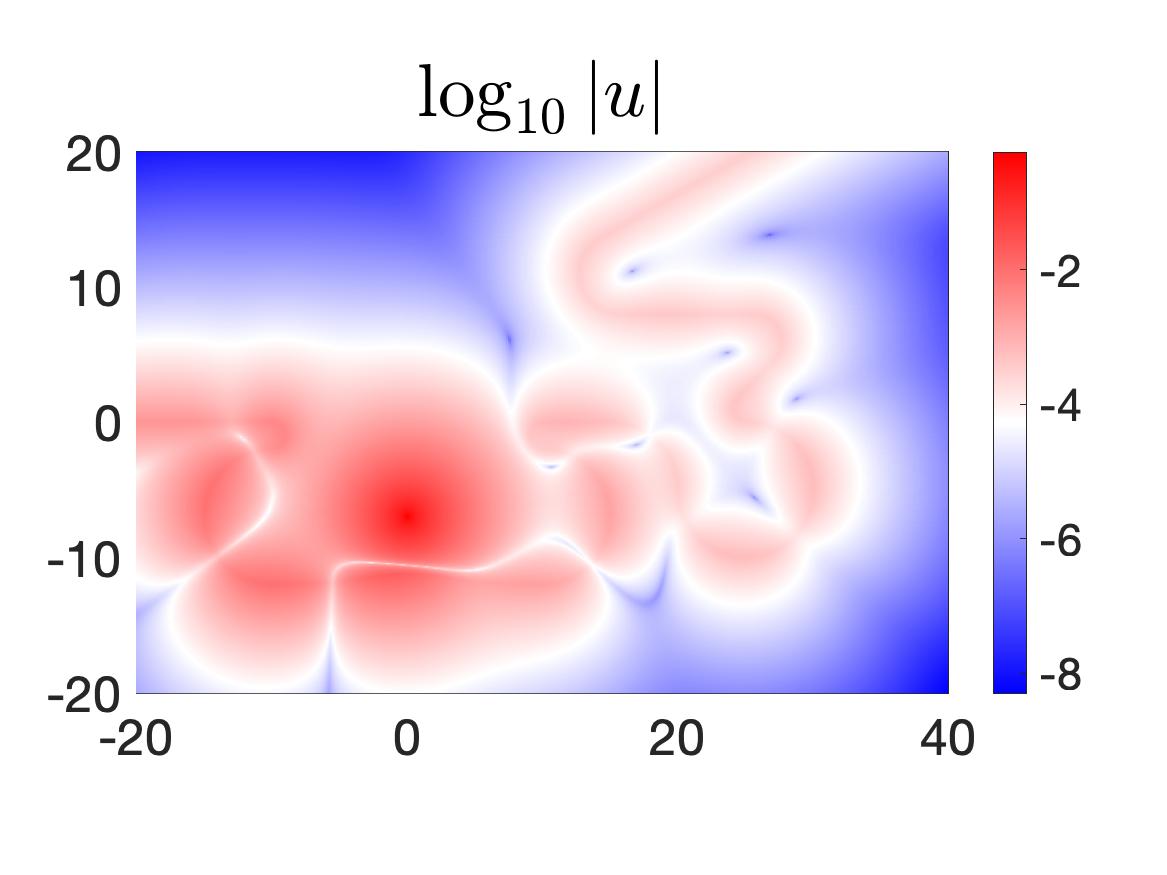}
    \caption{Densities, $\mu$ and $\rho$, (top panel) and Green's function, $u$, (bottom panels) for the interface $\Gamma_2$ with $m = 2/3$ and $E=1/3$. The top plot zooms in on the part of $\Gamma_2$ connecting the points $(-30.0,0.0)$ and $(14.8,-8.5)$, with $t=185$ corresponding to the point $(-13.5,-10.6)$.}
    \label{fig:exc}
\end{figure}

\begin{figure}
    \centering
    \includegraphics[width=0.75\linewidth]{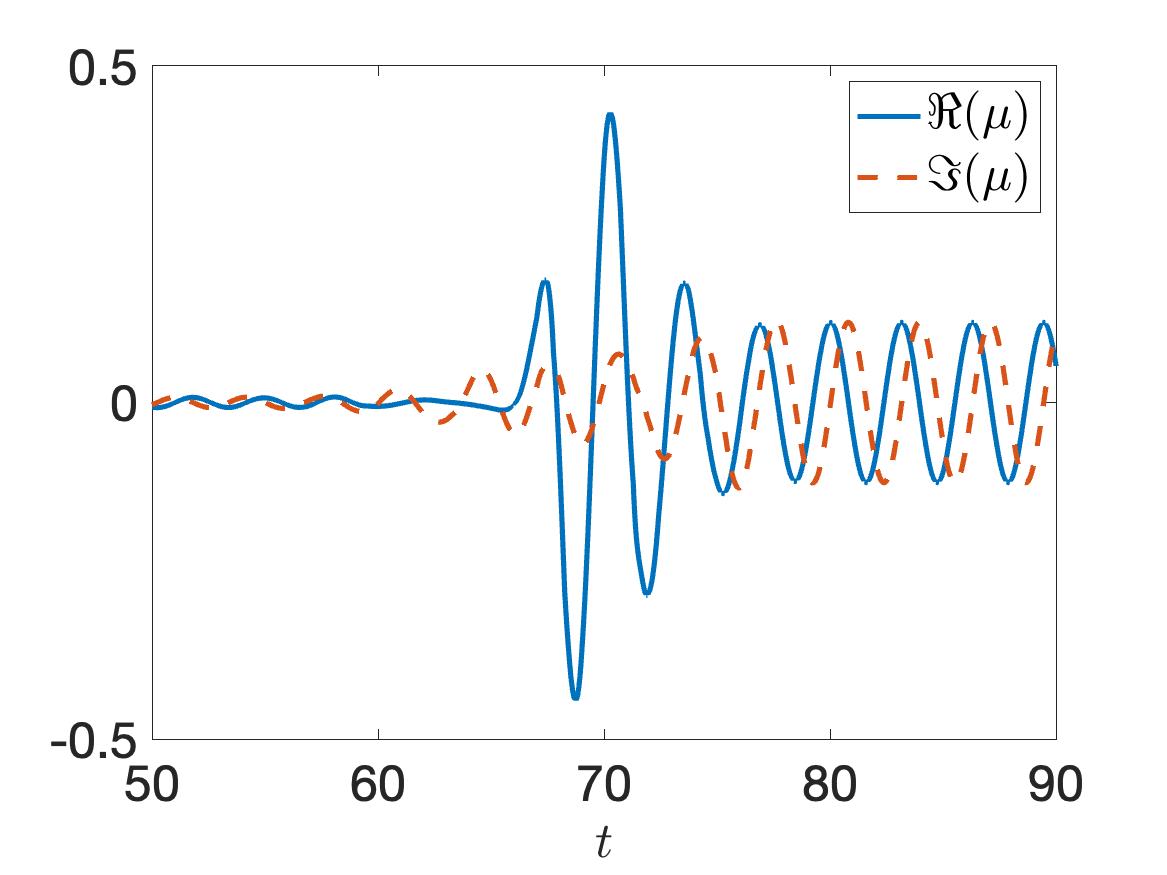}\\
    \includegraphics[width=0.45\linewidth]{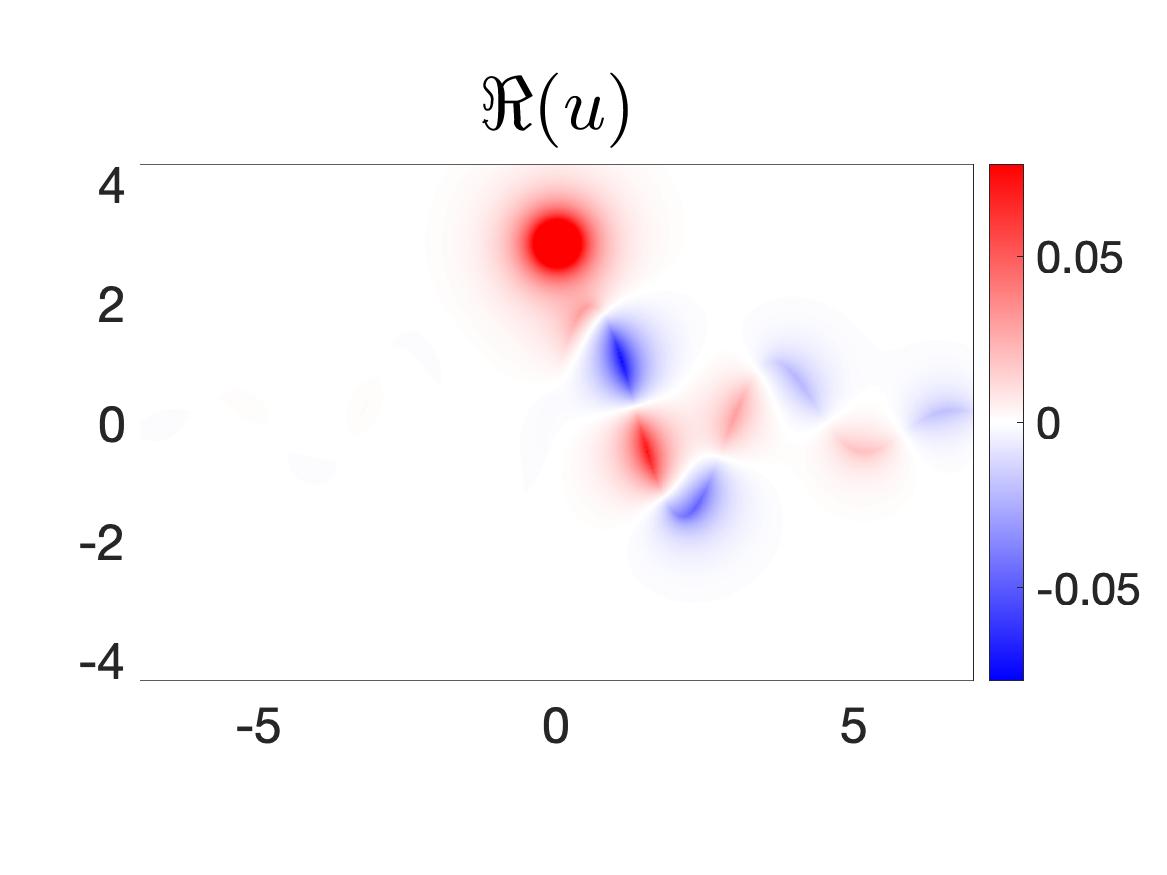}
    \includegraphics[width=0.45\linewidth]{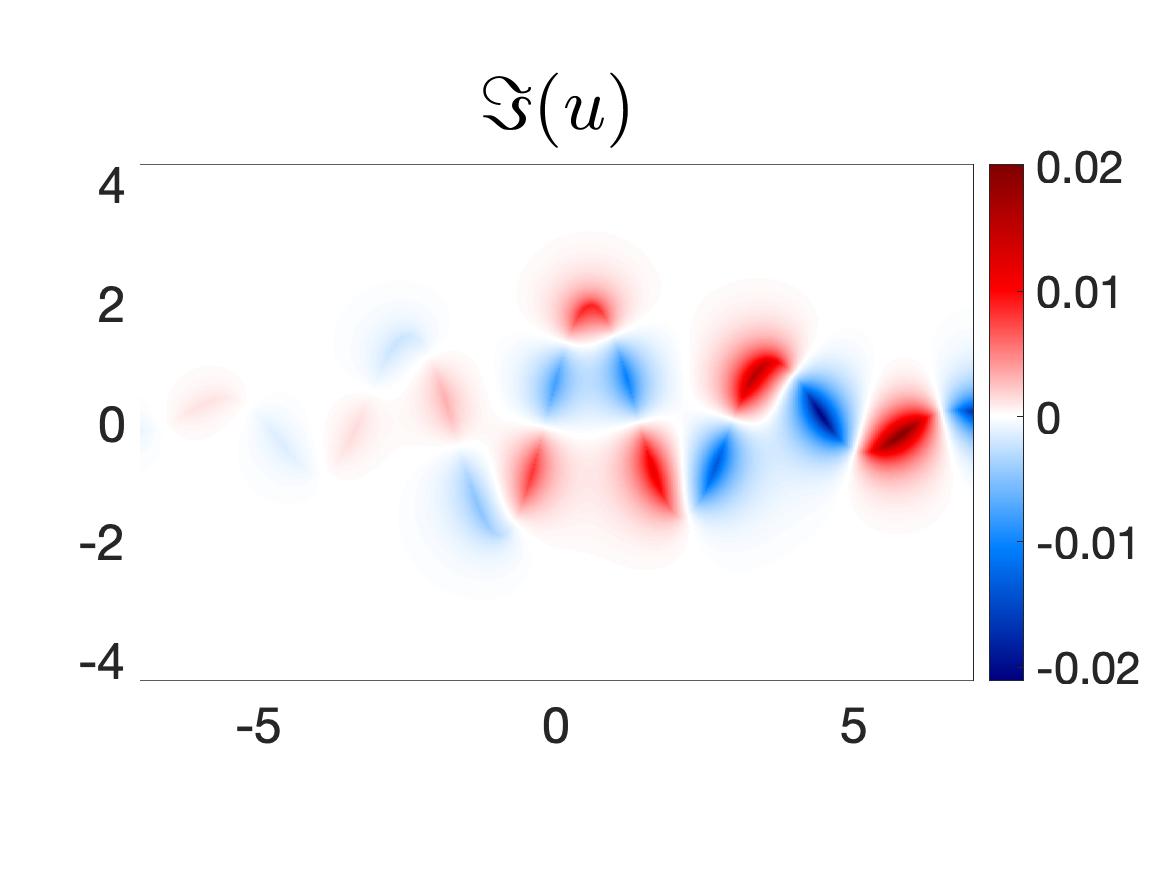}
    \caption{Density, $\mu$, (top panel) and Green's function, $u$, (bottom panels) for the interface $\Gamma_3$ with $m= 3$ and $E=2$. The top plot zooms in on the part of $\Gamma_3$ connecting the points $(-10.0,0.0)$ and $(18.0,0.0)$, with $t=70$ corresponding to the point $(1.4,-0.2)$.}
    \label{fig:exw}
\end{figure}

\begin{figure}
    \centering
    \includegraphics[width=0.7\linewidth]{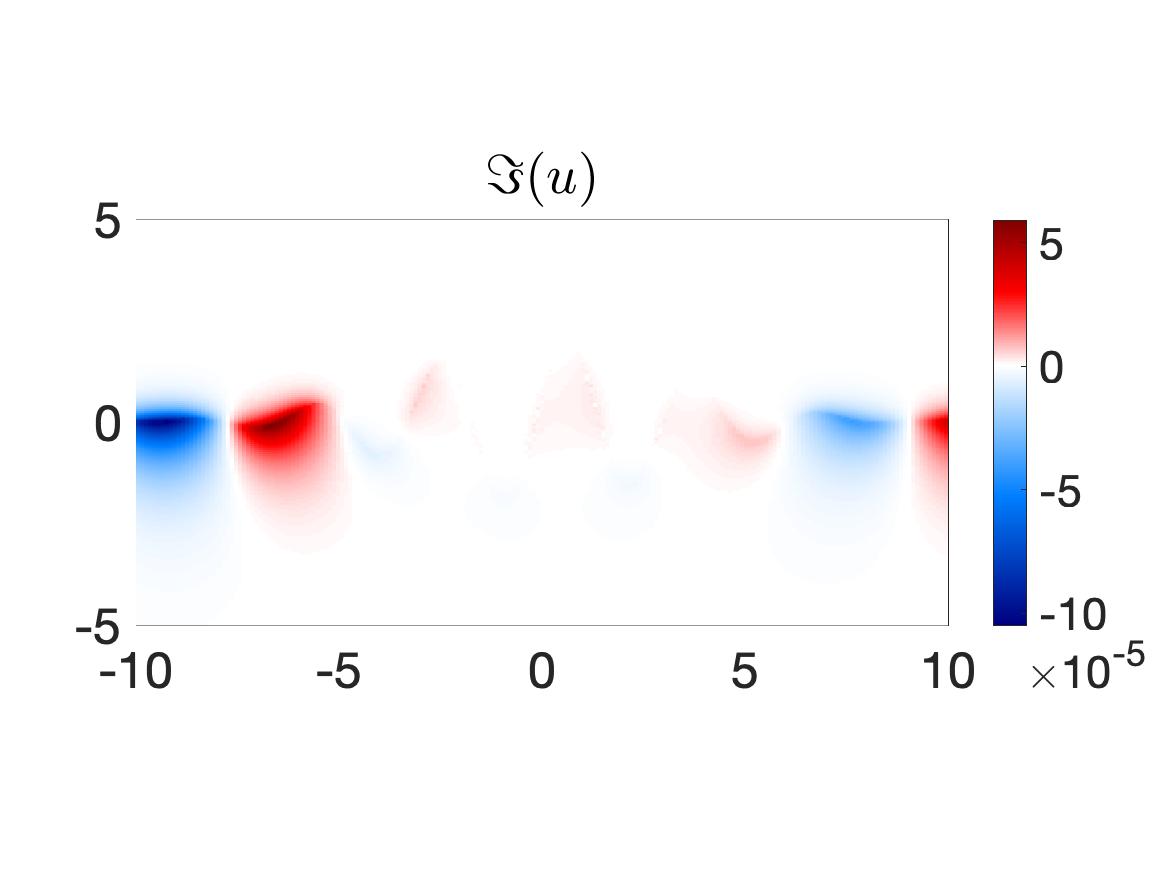}\\[-1.4cm]
    \includegraphics[width=0.7\linewidth]{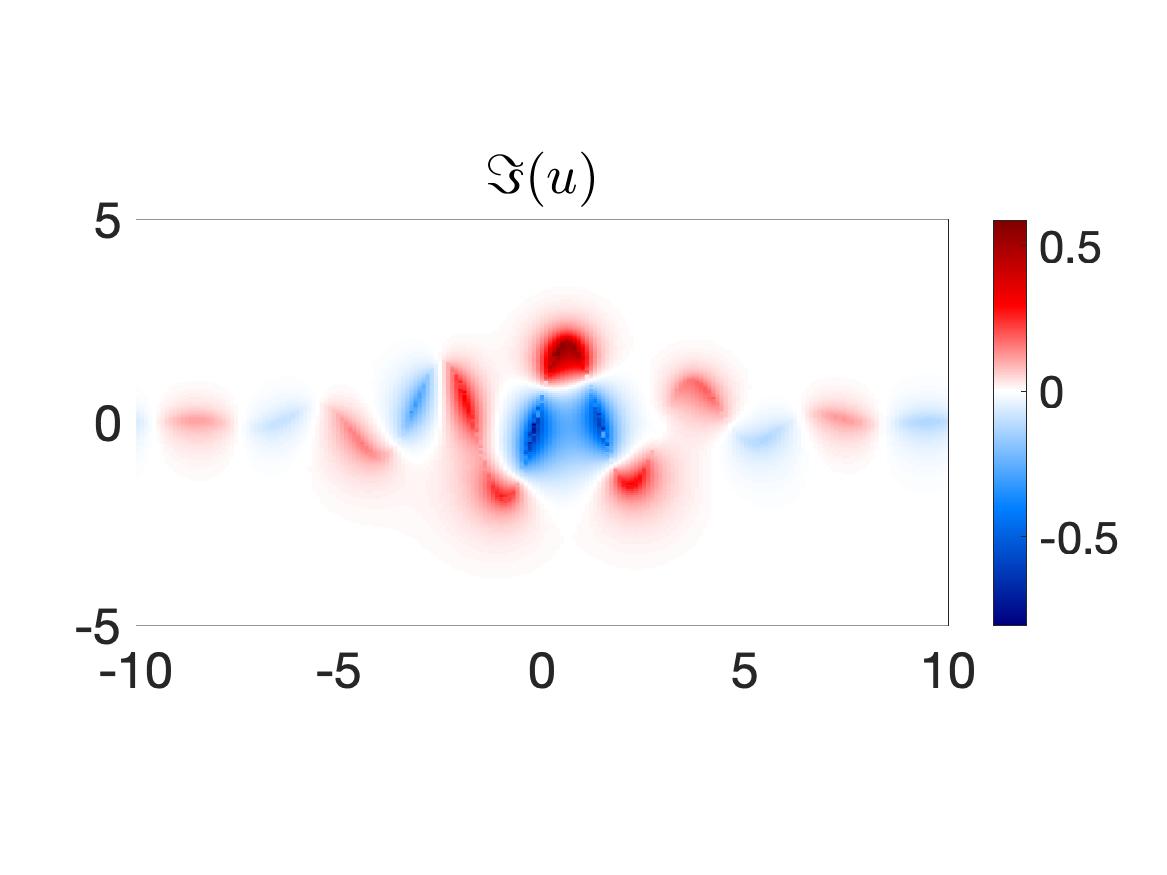}\\[-1.4cm]
    \includegraphics[width=0.7\linewidth]{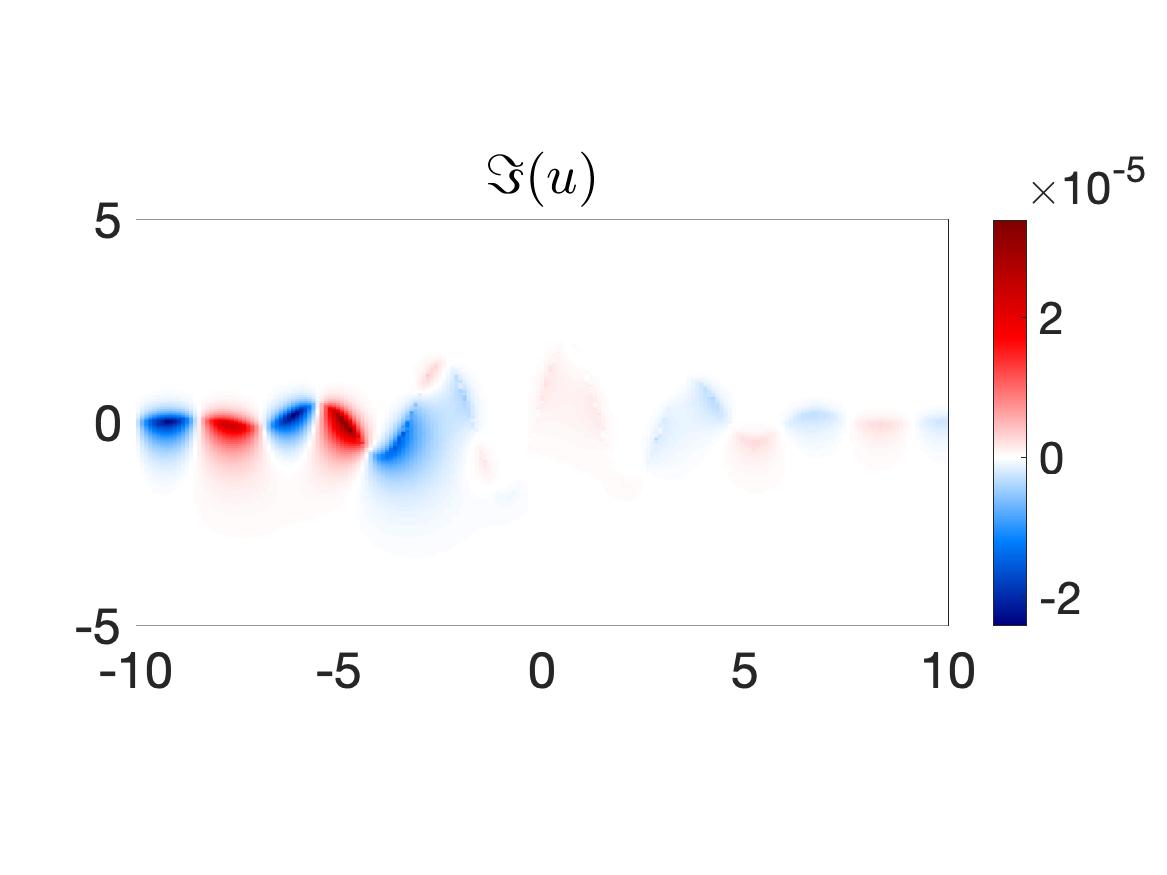}
    \caption{Green's function, $u$, for the interface $\Gamma_3$, where $(m_1,m_2,E)=(1.20,4.00,1.00)$ for the top panel, $(m_1,m_2,E)=(2.00,3.00,1.50)$ for the center panel, and $(m_1,m_2,E)=(2.25,6.00,2.00)$ for the bottom panel.}
    \label{fig:exw2}
\end{figure}

Our numerical experiments (Figures \ref{fig:flat} and \ref{fig:exw}) verify that the density $\rho$ from \eqref{eq:ubif} is rapidly decaying, while $\mu := \cP [\rho]$ obeys an outgoing radiation condition, oscillating without decay. For an illustration of $\rho$ and $\mu$ in the vicinity of a source (zooming in on the region where $\rho$ is not small), see Figure \ref{fig:exc}.

In Figure \ref{fig:exc_scattering} we move the source to $(-40,1)$ and observe that 
the propagation of the resulting wave along the interface depends on the choice of $(m,E)$. As stated in Remark \ref{rem:lambda}, increasing the value of $\omega$ is equivalent to smoothing out the interface, thus it makes sense that the solution on the top row gets reflected while the one on the bottom gets transmitted. Similarly, the small values of $m$ and $E$ in Figure \ref{fig:exc} (combined with the corresponding source location; see Figure \ref{fig:interfaces}, bottom left panel) result in a solution that is concentrated near the oscillatory part of $\Gamma_2$.

\begin{figure}
    \centering
    \includegraphics[width=0.7\linewidth]{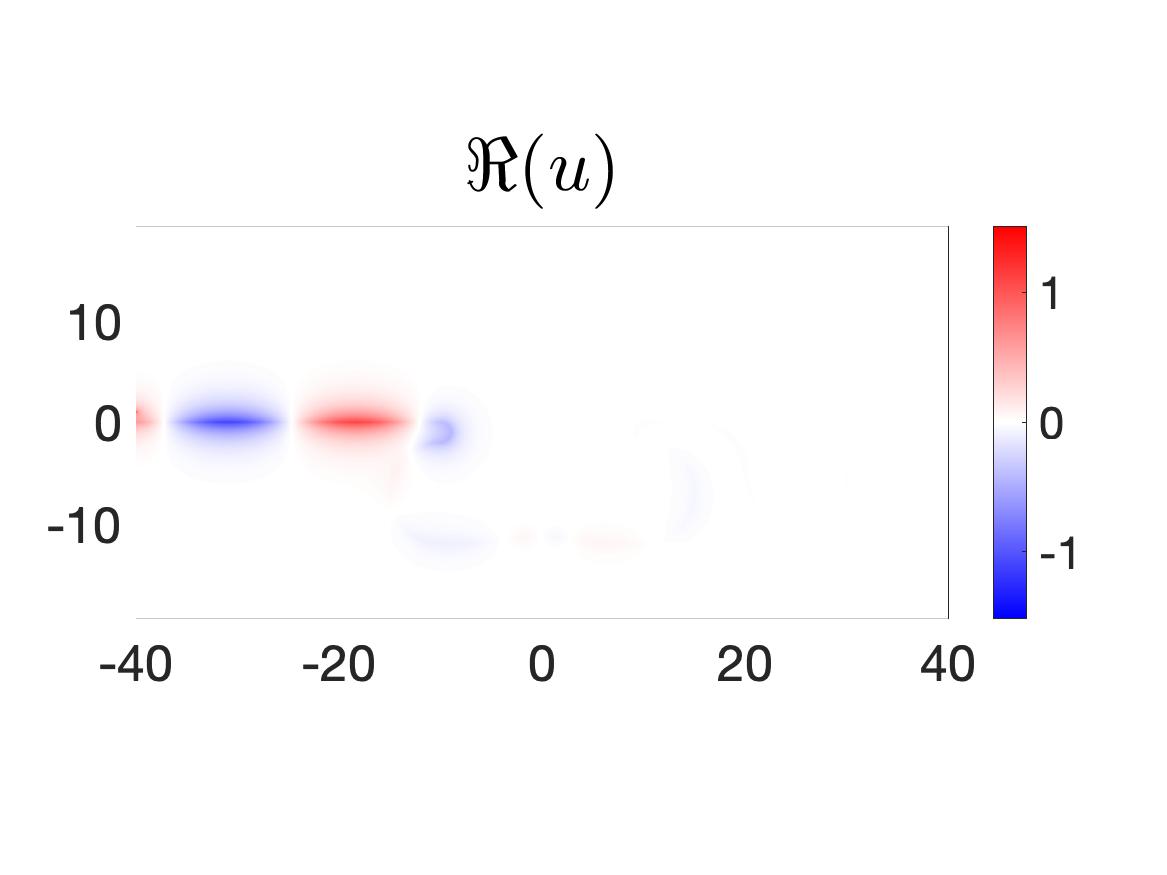}\\[-1.5cm]
    \includegraphics[width=0.7\linewidth]{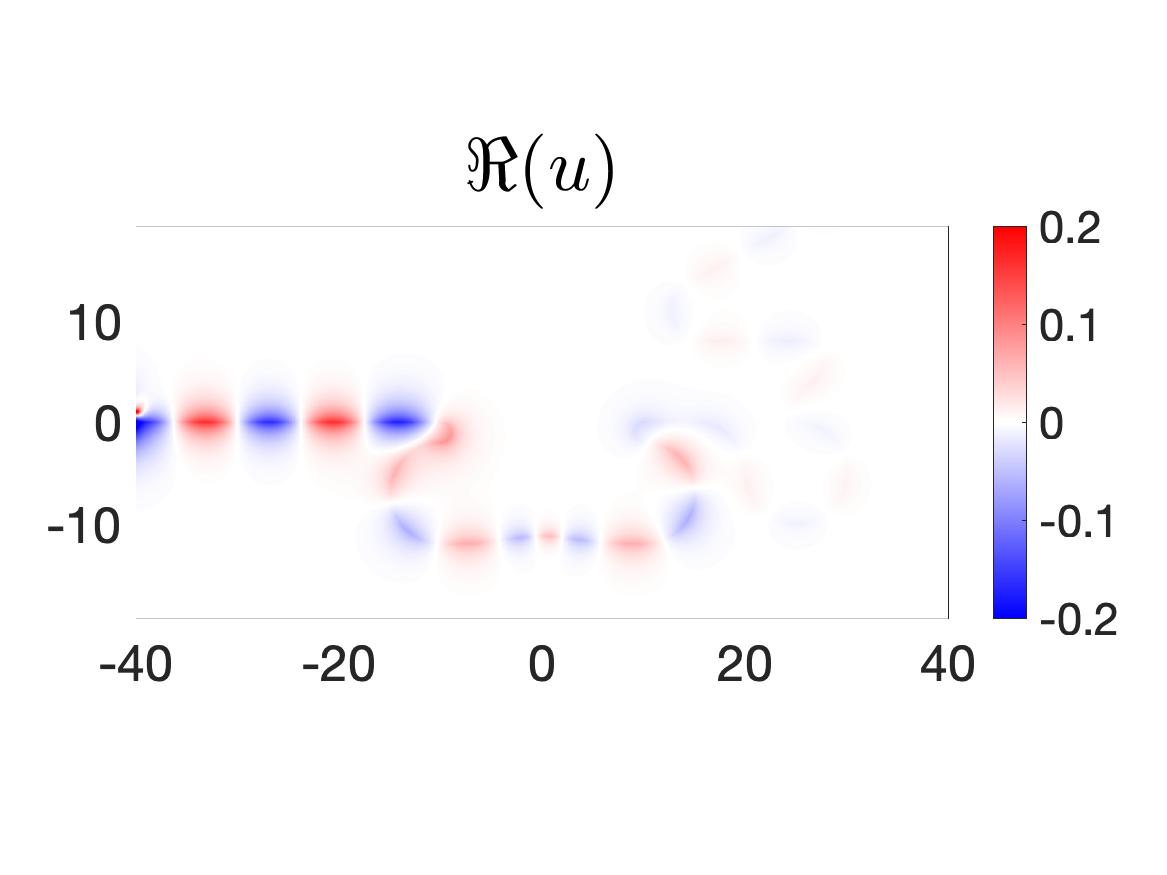}\\[-1.5cm]
    \includegraphics[width=0.7\linewidth]{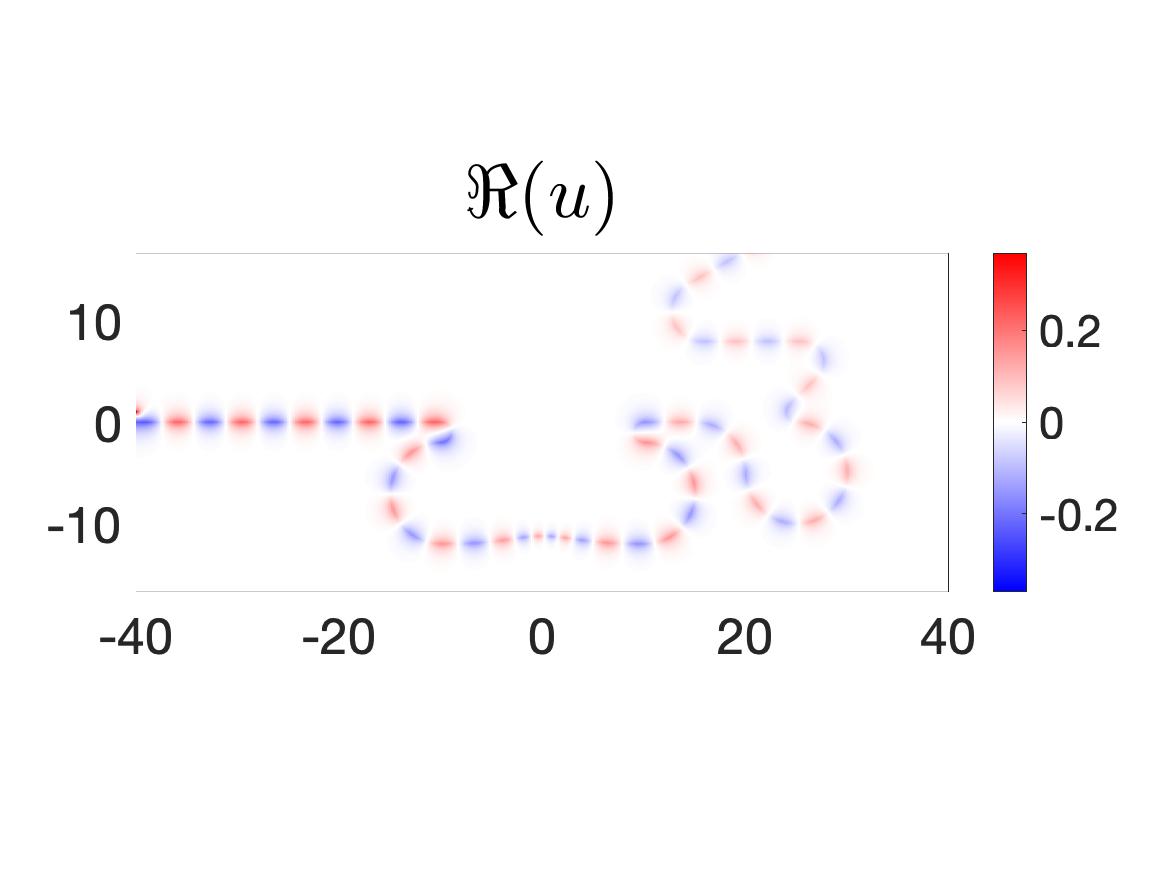}
    \caption{Green's function, $u$, for the interface $\Gamma_2$ with source located at $(-40,1)$.
Here, $(m,E)$ is $(0.75, 0.25), (0.75,0.5)$ and $(1.5,1)$ 
    going top to bottom.}
    \label{fig:exc_scattering}
\end{figure}

\medskip

We conclude this section with a scattering experiment.
Let us take an interface parametrized by $\gamma (t) = 2 e^{-0.05 t^2} \sin (bt + 0.4)$, for some $b\ge 0$. Note that the value $b=2$ gives $\Gamma_3$. We place a source near the interface, to the left of and far away from its oscillations.

The density $\mu$ is proportional to the solution $u$ along the interface, as \eqref{eq:mu} implies that $\mu = 2m (u_i + u_s) = 2m u$ on $\Gamma$.
The outgoing condition \eqref{eqn:outgoing} implies that $\mu (t) \approx C e^{iEt}$ for $t > 0$ sufficiently large.
Between the source and oscillations, we have that $\mu (t) \approx A e^{iEt} + B e^{-iEt}$, where $|A|$ and $|B|$ are the respective amplitudes of the incoming and reflected waves. The transmission and reflection coefficients are then defined by $T_L := |C|^2/|A|^2$ and $R_L := |B|^2/|A|^2$.

Since $\mu = (I+Q) [\rho]$ with $\rho \in L^1$, it follows that $\mu (t) \approx \frac{m^2}{E} e^{\pm iEt} \hat{\rho} (\pm E)$ as $t \rightarrow \pm \infty$.
It follows that $C = \frac{m^2}{2E} \hat{\rho} (E)$.
As $t \rightarrow -\infty$, we get contributions from both the reflection and the source. Hence $B = L - B_0$, where $L = \frac{m^2}{2E} \hat{\rho}(-E)$ and $B_0 = \frac{m^2}{2E} \hat{\rho}_0 (-E)$ with $\rho_0$ the solution corresponding to $k=0$. Thus, 
using the identity $T_L + R_L= 1$, 
we compute the transmission and reflection coefficients using only $\hat{\rho}_0 (-E)$ and $\hat{\rho} (\pm E)$.

For a plot of $R_L$ as a function of $b$, see Figure \ref{fig:trb} (top panel). 
When $b$ is small, $\Gamma$ resembles a flat interface and thus $R_L$ is close to $0$. For larger values of $b$ (say $b>2)$, the oscillations of the interface 
cause solutions to back-scatter, as $R_L$ is approximately $1$.
The transition of $R_L$ from $0$ to $1$ contains a small interval in $b$ at which there is a sudden dip in $R_L$. See Figure \ref{fig:trb} (bottom panels) for an illustration of the qualitatively different behavior of solutions corresponding to nearly identical interfaces. We suspect there may be other values of $b$ (out of the range of values in Figure \ref{fig:trb}) for which sharp transitions in $R_L$ occur, 
but postpone a thorough investigation of these critical values to future analyses.

\begin{figure}
    \centering
    \includegraphics[width=0.75\linewidth]{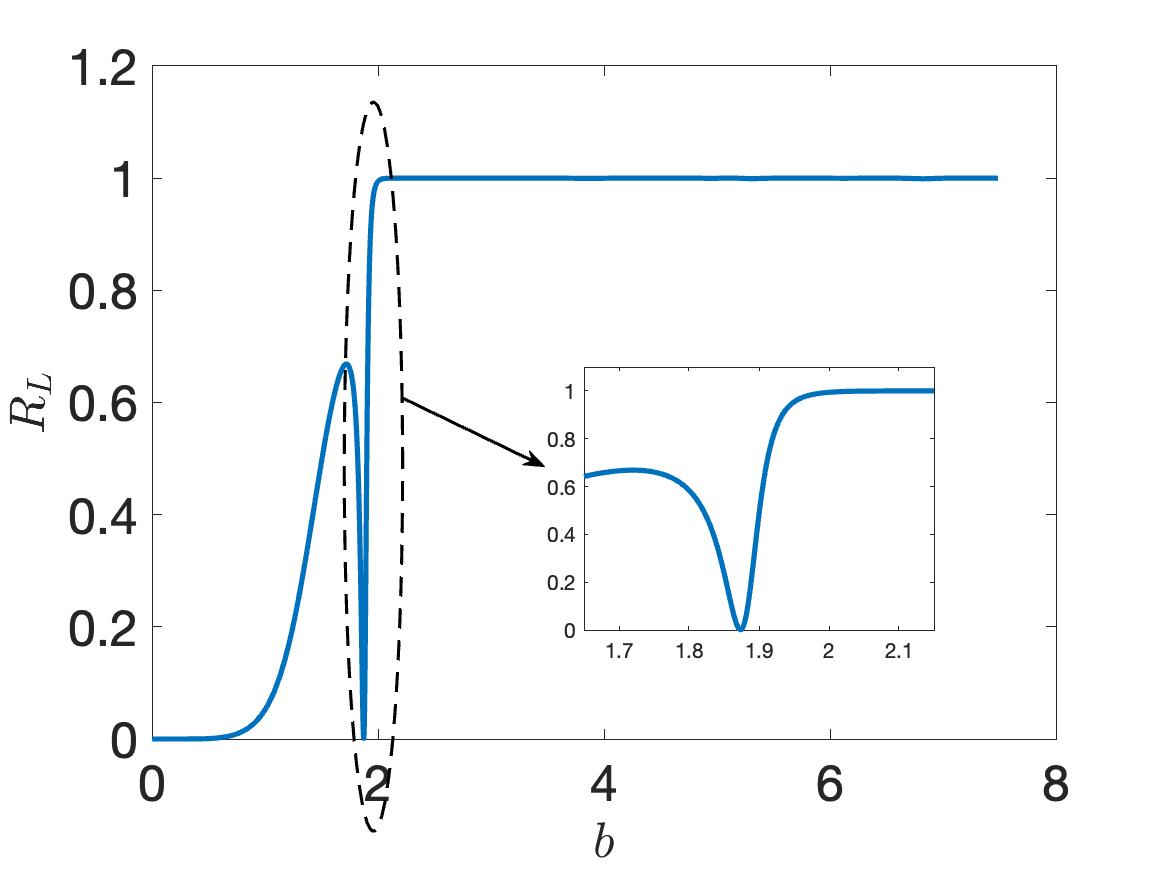}\\
    \includegraphics[width=0.45\linewidth]{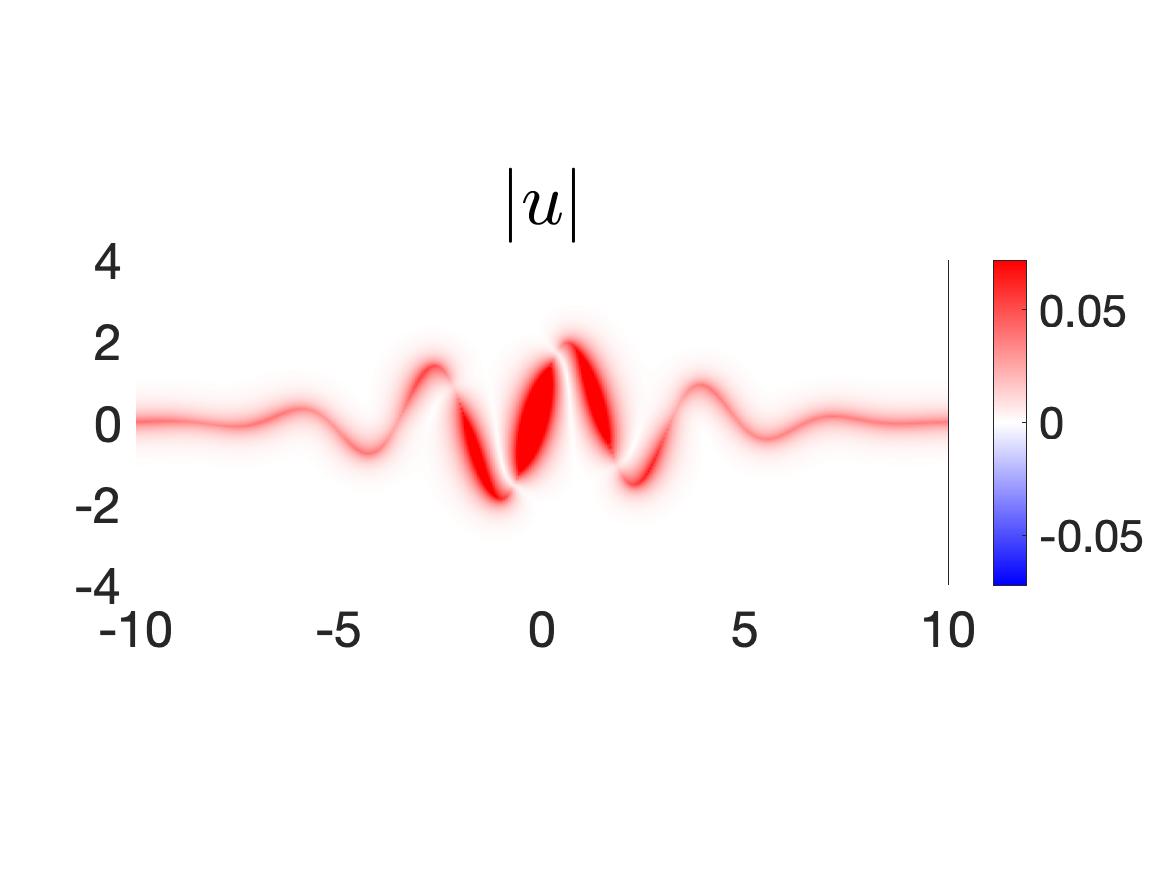}
    \includegraphics[width=0.45\linewidth]{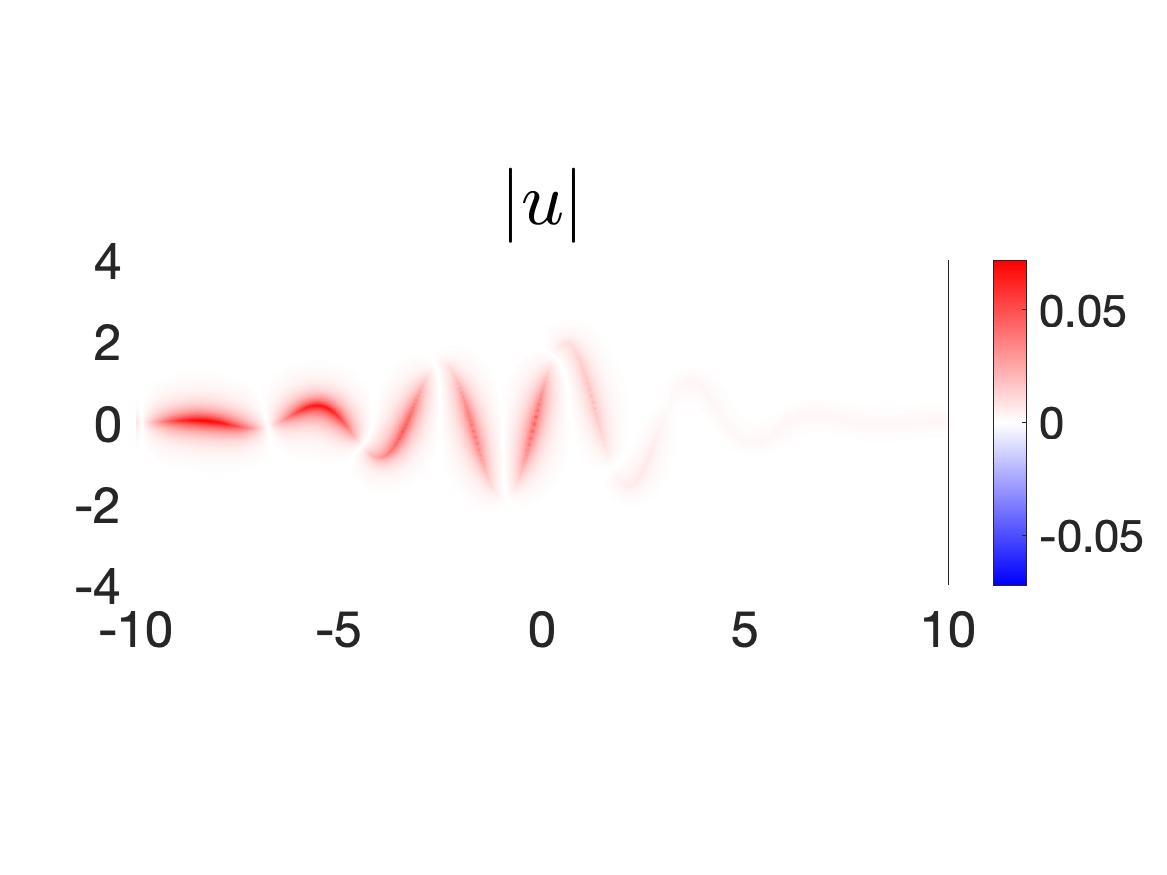}
    \caption{Scattering experiment for an interface parametrized by $\gamma (t) = 2 e^{-0.05 t^2} \sin (bt + 0.4)$, for $b\ge 0$. 
    The source is located at $(-40,1)$, and $(m,E)=(4,1)$. The bottom panels illustrate the Green's function, $u$, corresponding to $b=1.87$ (left) and $b=2.00$ (right).}
    \label{fig:trb}
\end{figure}

\section{Concluding remarks \label{sec:conc}}

In this paper, we derived a novel integral equation formulation for the efficient solution of Klein-Gordon singular waveguides which model surface waves along the interface of two insulating media. The propagation of waves when the interface is infinite, while physically interesting, poses analytical and computational challenges, owing to the slow decay of the solution along the interface, and the presence of continuous spectrum for standard integral representations. In order to overcome this issue, we derived an analytic preconditioner for the standard integral equation which captures the oscillatory behavior of the solution on the flat and infinite pieces of the boundary, and proved that the resultant preconditioned integral equation has a bounded inverse for the range of physical parameters of interest.

The analytic preconditioner is derived to ensure that the preconditioned integral equation has a bounded inverse on a flat interface. The proof of the preconditioned operator having a bounded inverse on interfaces that are asymptotically flat relies on showing that the difference between the integral operators on the flat and perturbed interfaces is compact in an exponentially weighted $L^2$ space (since the kernel of the integral operator is both is exponentially localized, and smoothing), deriving uniform bounds on the operators for scaled physical parameters (the decay rates in the bulk, and the propagation frequency along the interface), and Kato perturbation theory.

The well-conditioned nature of the preconditioned integral equation in exponentially weighted $L^2$ spaces lends itself to its efficient numerical solution via finite truncation of the infinite interface. Its well-conditioned nature implies that iterative methods like GMRES for solving the discretized integral operator would converge in $O(1)$ iterations, and the use of existing acceleration tools, i.e. the Yukawa fast multipole method and sweeping algorithms allows for its application in linear computational complexity; the combination of which results in an efficient and high-order accurate numerical method for the solution of the preconditioned integral equation. We illustrate both the accuracy and speed of our approach through several numerical examples.

The method proposed in this work is not limited solely to Equation \ref{eq:pde}. Indeed, it relies only on two key features: the exponential decay of the fundamental solution in each separate region (i.e. they are `insulating'); and an explicit characterization of the wavelength of the surface waves. As such, it should be relatively straightforward to extend it to Dirac models for graphene \cite{3,bernevig2013topological}, Dirac models for twisted bilayer graphene \cite{bal2023mathematical}, and linearized shallow water equations used in models of equatorial waves \cite{bal2024topological, delplace2017topological}. One would expect that it should also be easily extendable to three dimensional problems and PDEs with multiple interfaces.

        





\ack 
G. Bal and S. Quinn were supported in part by NSF grants DMS-1908736 and EFMA-1641100. The Flatiron Institute is a division of the Simons Foundation. The authors would like to thank the anonymous referees for many helpful comments that led to a much-improved manuscript. J. Hoskins, S. Quinn, and M. Rachh would also like to thank the American Institute of Mathematics SQuaREs program.



\frenchspacing
\bibliographystyle{plain}
\bibliography{refs}

\end{document}